\documentclass[acmsmall,screen,nonacm]{acmart}
\pdfoutput=1 
\usepackage{xcolor}
\usepackage{listings}
\usepackage{verbatim}

\usepackage{cleveref}

\usepackage{tikz}
\usepackage{amsmath}
\usepackage{pifont}
\usepackage[binary-units]{siunitx}

\newcommand{\codestyle}[1]{{{\fontsize{9}{10}\ttfamily #1}}}

\newcommand{\size}{\codestyle{size}}
\newcommand{\ins}{\codestyle{insert}}
\newcommand{\del}{\codestyle{delete}}
\newcommand{\contains}{\codestyle{contains}}
\newcommand{\cas}{\texttt{CAS}}
\newcommand{\nul}{\texttt{NULL}}
\newcommand{\scan}{\codestyle{scan}}
\newcommand{\up}{\codestyle{update}}
\newcommand{\htb}{\codestyle{HashTable}}
\newcommand{\skl}{\codestyle{SkipList}}
\newcommand{\bst}{\codestyle{BST}}
\newcommand{\shtb}{\codestyle{SizeHashTable}}
\newcommand{\sskl}{\codestyle{SizeSkipList}}
\newcommand{\sbst}{\codestyle{SizeBST}}
\newcommand{\snapskl}{\codestyle{SnapshotSkipList}}
\newcommand{\vcasbst}{\codestyle{VcasBST-64}}

\newcommand{\term}[1]{\textit{#1}}

\newtheorem{theorem}{Theorem}[section]
\newtheorem{claim}[theorem]{Claim}
\newtheorem{observation}[theorem]{Observation}

\setcopyright{none}
\copyrightyear{2022}
\acmYear{2022}
\acmDOI{}
\acmConference[]{}{}{}
\acmBooktitle{}
\acmPrice{}
\acmISBN{}
    
\begin{document}

\title{Concurrent Size}
\titlenote{This work was supported by the Israel Science Foundation Grant No. 1102/21.}

\author{Gal Sela}
\affiliation{
  \institution{Technion}
  \country{Israel}
}
\email{galy@cs.technion.ac.il}

\author{Erez Petrank}
\affiliation{
  \institution{Technion}
  \country{Israel}
}
\email{erez@cs.technion.ac.il}

\begin{abstract}
The size of a data structure (i.e., the number of elements in it) is a widely used property of a data set. 
However, for concurrent programs, obtaining a correct size efficiently is non-trivial.
In fact, the literature does not offer a mechanism to obtain a correct (linearizable) size of a concurrent data set without resorting to inefficient solutions, such as taking a full snapshot of the data structure to count the elements, or acquiring one global lock in all update and size operations.
This paper presents a methodology for adding a concurrent linearizable \size{} operation to sets and dictionaries with a relatively low performance overhead. Theoretically, the proposed size operation is wait-free with asymptotic complexity linear in the number of threads (independently of data-structure size). Practically, we evaluated the performance overhead by adding \size\ to various concurrent data structures in Java---a skip list, a hash table and a tree. The proposed linearizable \size{} operation executes faster by orders of magnitude compared to the existing option of taking a snapshot, while incurring a throughput loss of $1\%-20\%$ on the original data structure's operations.

\end{abstract}



\maketitle

The conference version of this paper is available at \cite{sela2021concurrentSize}, 
and the code is publicly available at \cite{artifactConcurrentSize}.

\section{Introduction}\label{section:intro}

Concurrent data structures are fundamental building blocks of concurrent programming, utilized to leverage modern multi-core processors. There has been substantial work on the design of efficient and scalable concurrent data structures with good progress guarantees, in order to benefit concurrent algorithms at large. A fundamental, widely used, property of a data structure is its size (i.e., the number of elements it contains). 
In Java, for example, any collection or map class that implements one of the elementary interfaces \texttt{java.util.Collection} or \texttt{java.util.Map} \cite{java18} must implement a \size{} method. Interestingly, implementing an efficient and correct \size{} operation for a concurrent data structure is non-trivial. For a formal treatment, we use linearizability as the correctness criterion of concurrent executions~\cite{herlihy1990linearizability,sela2021linearizability}, but the discussion below also applies to other intuitive correctness criteria.

The literature does not offer an acceptable solution to implementing a correct \size{} operation, and existing implementations give up  correctness in order to avoid a significant performance deterioration.
For example, the non-blocking collections and maps in the \texttt{java.util.concurrent} package \cite{lea2004java} implement a non-linearizable \size{} method that returns an estimate of the size. The returned estimate may be inaccurate when the object is concurrently modified during the execution of \size{}. In contrast, a linearizable \size{} operation would tolerate concurrent update operations and retrieve the exact number of elements in the data structure at some point during the execution of the \size{} operation. 

Existing solutions are incorrect or inefficient. 
Ignoring concurrency, one can determine the size of a data structure simply by traversing it and counting the number of encountered items. 
This is the approach taken by the \size{} method of Java's \texttt{ConcurrentLinkedQueue} and \texttt{ConcurrentLinked\-Deque}.
This approach is fine for a sequential execution, but for a concurrent execution this implementation is not linearizable. The following is a worst-case scenario for this implementation. Consider an execution on a linked list with the single item $1$. Assume a thread $T$, running this \size{} implementation, starts the traversal from the node containing $1$ and then gets preempted. At this point, the following steps may occur repeatedly: some thread appends a node with the item $2$ to the end of the list, increasing the list's size to $2$; next, $T$ gets scheduled, resumes its traversal and proceeds to this new node; then, some thread deletes the node containing $1$, so the list's size is $1$ again. Next, some thread inserts $3$ and deletes $2$, letting $T$ see a third element, etc.; until eventually---after some item $s$ is appended---the thread $T$ gets to the end of the list before another thread gets the chance to insert an additional item. In this scenario, $T$ will erroneously return a possibly large $s$ as the list size, while in practice the list size never exceeded 2. 
While this is a worst-case scenario, one can envision many other scenarios in which the returned value would be incorrect. 

Alternatively, it is possible to obtain a correct \size{} implementation by obtaining a linearizable snapshot of the data structure (e.g., using any of the methods in \cite{wei2021constant,petrank2013lock,nelson2022bundling,arbel2018harnessing}) and then iterate over the returned snapshot to count the number of elements in the snapshot. While correct, this solution is inefficient, yielding a time complexity linear in the number of elements in the data structure. 

If we want to avoid such a high cost for the \size{} operation, then we need to keep some metadata that allows computing the size of the data structure efficiently when needed, and let the data-structure operations maintain this metadata. Naively, the metadata would just be the current size value. 
A natural attempt to implement such a \size{} operation would be to keep the size in a designated field of the data structure and let the operating threads update it with each operation that affects the size.  An \ins{} operation would execute a size increment after inserting its item, and a \del{} operation would execute a decrement of the size field after performing the deletion. 
Java's \texttt{ConcurrentSkipListMap} and \texttt{ConcurrentHashMap} use such an implementation. 
However, the separation between the data-structure update and the metadata update foils linearizability. As an example, consider $n$ threads that are preempted exactly after inserting an element to the data structure and before updating the size field. At this point the size field would be off by $n$ and thus inconsistent with a view of another thread that actually reads the data structure. 

A simplified execution for one updating thread is depicted in  \Cref{fig:non-atomic-update-causes-conflicts}. There, only one update operation is ever executed on the data structure: one thread inserts $1$ into an empty structure. Another thread that starts by calling \contains(1), sees that the data structure already contains the single element in it, but then it calls \size() and receives 0. We executed this simple program on Java's \texttt{ConcurrentSkipListMap} several times, and we actually witnessed executions that reproduced the contradicting result as depicted in \Cref{fig:non-atomic-update-causes-conflicts}. 
This demonstrates that the \size{} method of \texttt{ConcurrentSkipListMap} is not linearizable. The core issue is the separation between the actual data structure update and the subsequent update of the metadata.

\setlength{\unitlength}{0.1mm}
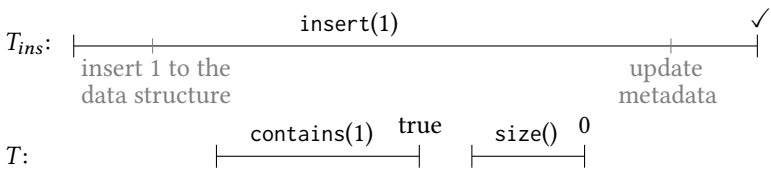
\begin{figure}[b]
\begin{picture}(1050,230)
\put(10,150){$T_{ins}$:}
\put(10,0){$T$:}

\put(100,170){\line(0,-1){30}} 
\put(100,155){\line(1,0){910}} 
\put(1010,170){\line(0,-1){30}} 
\put(400,175){\ins(1)}
\put(1000,180){$\checkmark$} 
\put(205,165){\color{gray} \line(0,-1){20}} 
\put(110,115){\color{gray} insert 1 to the}
\put(110,80){\color{gray} data structure}
\put(895,165){\color{gray} \line(0,-1){20}} 
\put(840,115){\color{gray} update}
\put(825,80){\color{gray} metadata}

\put(290,25){\line(0,-1){30}} 
\put(290,10){\line(1,0){270}} 
\put(335,30){\contains(1)}
\put(560,25){\line(0,-1){30}} 
\put(532,40){true} 

\put(630,25){\line(0,-1){30}} 
\put(630,10){\line(1,0){150}} 
\put(660,30){\size()}
\put(780,25){\line(0,-1){30}} 
\put(772,40){0} 

\end{picture}
\caption{An execution with conflicting \contains{} and \size{} results due to the separation between updating the data structure and the size metadata}
\label{fig:non-atomic-update-causes-conflicts}
\end{figure}

\setlength{\unitlength}{0.1mm}
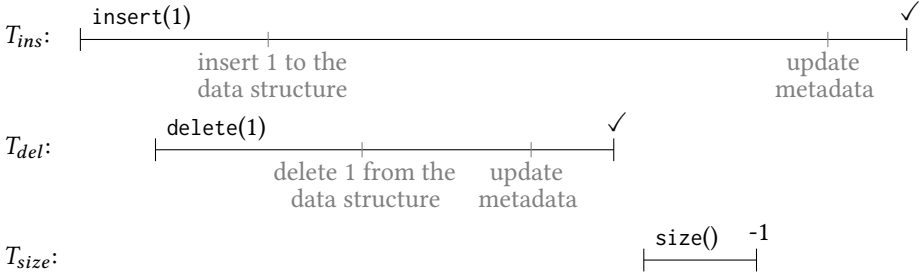
\begin{figure}[ht]
\begin{picture}(1200,370)
\put(0,295){$T_{ins}$:}
\put(0,150){$T_{del}$:}
\put(0,5){$T_{size}$:}

\put(100,315){\line(0,-1){30}} 
\put(100,300){\line(1,0){1100}} 
\put(1200,315){\line(0,-1){30}} 
\put(1190,325){$\checkmark$} 
\put(115,320){\ins(1)}
\put(350,310){\color{gray} \line(0,-1){20}} 
\put(256,260){\color{gray} insert 1 to the}
\put(256,225){\color{gray} data structure}
\put(1095,310){\color{gray} \line(0,-1){20}} 
\put(1040,260){\color{gray} update}
\put(1025,225){\color{gray} metadata}

\put(200,170){\line(0,-1){30}} 
\put(200,155){\line(1,0){610}} 
\put(810,170){\line(0,-1){30}} 
\put(800,180){$\checkmark$} 
\put(215,175){\del(1)}
\put(475,165){\color{gray} \line(0,-1){20}} 
\put(357,115){\color{gray} delete 1 from the}
\put(380,80){\color{gray} data structure}
\put(700,165){\color{gray} \line(0,-1){20}} 
\put(645,115){\color{gray} update}
\put(630,80){\color{gray} metadata}

\put(850,25){\line(0,-1){30}} 
\put(850,10){\line(1,0){150}} 
\put(1000,25){\line(0,-1){30}} 
\put(990,37){-1} 
\put(865,30){\size()}

\end{picture}
\caption{An execution that yields a negative size due to the separation between updating the data structure and the size metadata}
\label{fig:non-atomic-update-causes-negative-size}
\end{figure}

Furthermore, updating the metadata separately from updating the data structure may yield a \size{} execution that returns a negative number. This means that the size operation is not only non-linearizable, but it can also not satisfy any correctness criterion that requires method calls to appear to happen in a one-at-a-time sequential order, e.g., it is not quiescently consistent \cite{aspnes1994counting,shavit1996diffracting,herlihy2008art} nor sequentially consistent \cite{lamport1979make,herlihy2008art}. 
Consider the following execution (depicted in \Cref{fig:non-atomic-update-causes-negative-size}).
Thread $T_{ins}$ inserts an item to the data structure, and before it updates the metadata, thread $T_{del}$ deletes that item and updates the metadata, registering its decrement. At this point, thread $T_{size}$ calls \size() that returns $-1$ based on the metadata, which currently reflects only the deletion and not yet the insertion.
The separation between the data-structure and metadata updates results here in updating the metadata in a reversed order, which is impossible since the deletion cannot succeed if it happens before the insertion. The returned size exposes this impossible operation order. 
The method calls in this execution do not appear to happen in a one-at-a-time sequential order since \size{} would never return a negative result in a legal sequential execution.

A more complex metadata maintenance is proposed by Afek et al. for computing the \size{} more efficiently \cite{Afek2012size}. But they, too, update the metadata after the data-structure update, and so their implementation suffers from the same problems. (We elaborate on issues in \cite{Afek2012size} in   \Cref{section:prev paper incorrect}.)

A third alternative for implementing the \size{} operation is to use locks to prevent a \size{} operation from exposing a temporary inconsistency between the data structure's state and the metadata. This too would create a severe bottleneck and deteriorate performance significantly.

In this paper we propose an efficient linearizable \size{} implementation. To the best of our knowledge, this is the first \size{} solution that provides both linearizability and efficiency (namely, not iterating over all elements of the data structure or using coarse-grained locking). 
We present a methodology for adding a linearizable \size{} operation to concurrent data structures that implement a set or a dictionary. Our methodology yields data structures that satisfy the following attractive theoretical properties:
\begin{enumerate}
    \item \label{property:linear size} The time complexity of the \size{} operation is linear in the number of threads.
    \item \label{property:wait-free size} The \size{} operation is wait-free, namely, a thread running a \size{} operation completes the operation within a finite number of steps, regardless of the activity of other threads.
    \item \label{property:preserve time complexity} The (asymptotic) time complexity of the original data-structure operations is preserved.
    \item \label{property:preserve progress} The progress guarantees of the original data structure are preserved. Namely, wait-free methods of the original data structure remain wait-free in the transformed data structure, and the same goes for lock-free or obstruction-free methods.
\end{enumerate}

To achieve Property (\ref{property:linear size}), we keep always-consistent metadata, from which the size can be correctly computed.
To prevent operations from exposing inconsistencies similar to the examples of  \Cref{fig:non-atomic-update-causes-conflicts,fig:non-atomic-update-causes-negative-size}, we work hard to achieve a single linearization point in which the data structure is modified and the metadata gets updated simultaneously. This is obtained by letting an operation appear as completed to other operations only when the metadata update occurs. Formally, the update of the metadata becomes the single linearization point of the entire data structure operation.  
Dependent data-structure operations are adapted to comply with the new linearization point, and help completing concurrent operations when necessary. For instance, a \del($k$) by thread $T_2$ that encounters an ongoing \del($k$) by another thread $T_1$ which has already deleted the key from the data structure, must help $T_1$ to update the metadata in order to complete the obstructing \del($k$) before returning a failure. It cannot block and wait for $T_1$ to update the metadata, since that might change the progress guarantees of the \del{} operation and foil Property (\ref{property:preserve progress}).

Helping another operation implies updating the metadata on its behalf. As always with multiple threads helping to execute a single operation, care has to be taken for the operation to be executed only once. We keep per-thread counters as the size metadata, and use a corresponding mechanism that enables helpers to determine whether the metadata already reflects the helped operation, to prevent a wrong double update of the metadata on behalf of the same operation. This mechanism enables helpers to efficiently make a determination and update the metadata if necessary, thus achieving Property (\ref{property:preserve time complexity}).

The size of a data structure is a fundamental property of a data set and having a methodology for obtaining an efficient accurate solution for it seems like an important point in the design space, which is currently missing in the literature. Using inaccurate solutions may yield unexpected results, e.g., sizes that the data structure never had and even a negative size. Such results may in turn yield unexpected bugs that may put the reliability of an entire system at risk.
A reliable solution is especially desirable in dynamic programming languages that favor correctness over performance, such as Python and Ruby, which use a global interpreter lock in their reference implementations and are expected to behave reliably even in optimized implementations that shed the global interpreter lock to obtain parallelism.
This follows the line of previous works \citep[e.g.,][]{daloze2018parallelization,meier2016parallel} that present solutions for reliable efficient parallelism.

In order to evaluate the performance overhead of the linearizable \size{} operation, we added the \size{} operation using the methodology described in this paper to various concurrent data structures in Java: a skip list, a hash table and a tree \cite{artifactConcurrentSize}. The proposed linearizable \size{} operation executes faster by orders of magnitude compared to counting the elements of a linearizable snapshot. It also demonstrates scalability and insensitivity to the data-structure size. However, obtaining a linearizable \size{} operation does come with some cost, incurring a throughput loss of $1\%-20\%$ on the original data structure's operations. 

The paper is organized as follows. 
\Cref{section:terminology} introduces some basic terminology. \Cref{section:related-work} surveys relevant work on snapshots.
We then describe our methodology, starting with the transformation of a linearizable data structure into one that uses our size mechanism in \Cref{section:data-structure transformation}, and proceeding with the size mechanism itself: the metadata design is covered in \Cref{section:metadata}, and \Cref{section:wait-free size} describes how the size is obtained in a wait-free form. We describe possible optimizations to our methodology's implementation in \Cref{section:optimizations}.
In \Cref{section:properties} we argue about the properties the methodology satisfies. \Cref{section:performance} presents an evaluation of the methodology applied to different data structures in a variety of workloads. We conclude in \Cref{section:conclusion}.
\section{Terminology}\label{section:terminology}

An execution is considered \term{linearizable} \cite{herlihy1990linearizability,sela2021linearizability} if each method call appears to take effect at once, between its invocation and its response events, at a point in time denoted its \term{linearization point}, in a way that satisfies the sequential specification of the objects.
A concurrent data-structure is \term{linearizable} if all its executions are linearizable.

A concurrent object implementation is \term{wait-free} \cite{herlihy1991wait} if any thread can complete any operation in a finite number of steps, regardless of the execution speeds of other threads. 

A \term{set} is a collection of keys without duplicates, supplying the following interface operations: an \ins($k$) operation which inserts the key $k$ if it does not exist or else returns a failure; a \del($k$) operation which deletes $k$ if it exists or else returns a failure; and a \contains($k$) operation which returns true if and only if $k$ exists in the set.

A \term{dictionary} (synonymously \term{map} or \term{key-value map}) is a collection of distinct keys with associated values, with operations similar to the ones of a set but with values integrated in them.
Throughout the paper we will refer only to sets for brevity, but all our claims apply to dictionaries as well.

A compare-and-swap instruction (henceforth \cas{}) on an object takes an expected value and a new value. It atomically obtains the object's current value and swaps it with the new value if the current one equals the expected value. The return value indicates whether the substitution was performed: its compareAndSet variant returns a corresponding boolean value; its compareAndExchange variant returns the obtained current value.
\section{Related work}\label{section:related-work}

A {\em snapshot object} \cite{afek1993atomic} is an abstraction of shared memory made of an array of $m$ cells, supporting two operations: \up($i$, $v$) that writes $v$ to the $i$-th cell of the array, and \scan() that returns the current values of all $m$ locations (i.e., a snapshot of the array). 
The {\em atomic snapshot} problem is to implement such an object such that its two operations are linearizable and wait-free.  
\citet{jayanti2005optimal} presents algorithms that solve the problem with optimal time complexity. We build on the fundamental ideas of \citet{jayanti2005optimal} in this work to design a wait-free \size{} operation.

However, this scheme is not suited for multiple concurrent \scan{} operations and does not allow other operations (such as reading a specific cell) to occur concurrently.
\citet{petrank2013lock} extend Jayanti's idea and introduce a technique for adding a linearizable wait-free snapshot operation to a concurrent set data structure. In supporting concurrent \size{} operations, we use their method to support multiple concurrent snapshot operations.

An alternative approach by \citet{wei2021constant,nelson2022bundling} obtains snapshots of concurrent data structures more efficiently, at the cost of higher space overhead. They keep copies of modified nodes and let the snapshot operation advance a timestamp. This timestamp is then used to read the content of the data structure at the time the snapshot was taken.
To support such a read of old values, operations on the data structure are responsible to maintaining lists of previous values of mutable fields. 
Specifically for obtaining the size, one may take a snapshot and use the returned timestamp to traverse the data structure at that time and count elements. 

Literature on range queries may be also utilized to take a full snapshot of a data structure. For instance, \citet{arbel2018harnessing} propose to implement range queries using epoch-based memory reclamation. 

The above snapshot algorithms can be used to obtain a linearizable size, but using them for this purpose is an overkill.
The comparison we make in \Cref{section:performance} to the algorithms of  \citet{petrank2013lock} and \citet{wei2021constant} demonstrates the clear benefit of using our algorithm which is tailored for obtaining the size.

\section{Data-structure transformation}\label{section:data-structure transformation}
In this section we specify how the fields and methods of a linearizable data structure can be modified in order to transform it into a data structure that uses our size mechanism. 
To efficiently obtain a linearizable size, we keep metadata from which the size may be computed. 
But unlike previous work, we make the data structure and the metadata change (linearize) simultaneously.
The data-structure operations are responsible to maintain the metadata.
The main idea is to make sure that updates are not visible to other operations until their metadata is updated. The way to enforce that, is to let each operation complete work for previous related operations, so that it does not view any intermediate states. The details follow. 

\paragraph{Successful operations update metadata} 
The first modification is to let each successful \ins{} or \del{} operation (i.e., an operation that succeeds to insert a new key or delete an existing key respectively) update the metadata to reflect the operation's effect on the size. 

\paragraph{Operations help concurrent operations on the same key update metadata}
To prevent operations from exposing inconsistencies similar to the examples of \Cref{fig:non-atomic-update-causes-conflicts,fig:non-atomic-update-causes-negative-size}, we linearize data-structure operations that alter the size at the time the metadata is updated to reflect them (informally, linearizing means logically considering them as applied). Dependent data-structure operations are adapted to comply with the new linearization point: if they observe that successful \ins{} or \del{} operations that they depend on have accomplished their original linearization point, they help them update the metadata so that they reach their new linearization point. For example, a \contains($k$) that encounters a node with the key $k$ inserted by a concurrent \ins($k$) cannot return \codestyle{true} before ensuring that the \ins{} is reflected in the metadata.

We focus on data structures implementing a set (i.e., a collection of distinct keys) or a dictionary (i.e., a collection of distinct keys with associated values) that provide standard \ins{}, \del{} and \contains{} operations. In such data structures, an operation on some key logically depends only on operations on the same key.
Accordingly, when an operation on some key encounters a node with that key, it acts as follows: if the node is unmarked, it updates the metadata on behalf of the \ins{} operation that inserted the node, to guarantee it is complete (in case the metadata is not yet updated with this \ins{}); if the node is marked as deleted, it ensures the metadata reflects the \del{} operation that marked the node before proceeding with its own execution.

\paragraph{Successful operations leave a trace for helpers} 
For operations to help unfinished operations on the same key to update the metadata, they must observe these unfinished operations. To facilitate this, successful \ins{} and \del{} operations prepare an \codestyle{UpdateInfo} object with the information required by helpers for updating the metadata on their behalf, and reference it from the node on which they operate. An \ins{} creates an \codestyle{UpdateInfo} object and places a reference to it in the node it is about to link, in a new \codestyle{insertInfo} field we add to node objects. 
A \del{} also creates an \codestyle{UpdateInfo} object, and needs to reference it from the node it deletes. To this end, we rely on a deletion pattern introduced by \citet{harris2001pragmatic} and commonly used in concurrent data structures \citep[e.g.,][]{harris2001pragmatic,heller2005lazy,herlihy2007simple,fraser2004practical,sundell2005fast}, where a node is first marked as deleted and then physically unlinked. We install the delete information together with the marking, as demonstrated in the following examples.

When the original marking step already marks the node as deleted by installing an object with information about the delete operation (this is true, for instance, for the binary search tree of \citet{ellen2010non}), then a \codestyle{deleteInfo} field referencing the \del{}'s \codestyle{UpdateInfo} object may be simply placed inside that object.
When the original marking step nullifies the node's value field (as in Java's \texttt{ConcurrentSkipListMap}), in the transformed data structure instead of setting the value field to \nul{}, it may be set to a reference to the \codestyle{UpdateInfo} object.
When the original marking step sets a bit in the node's \codestyle{next} field (like in Harris's linked list \cite{harris2001pragmatic}), a new \codestyle{deleteInfo} field in the node may be set to reference the \codestyle{UpdateInfo} object before the marking step.

\paragraph{Metadata is updated before unlinking a marked node}
The metadata must be updated on behalf of a \del{} before the relevant node is unlinked.
To see why, assume the metadata is updated to reflect a \del($k$) only after it completes to operate on the data structure, including unlinking the node with the key $k$. In this circumstance, a dependent operation like \contains($k$) might run in between, and then it will not observe the relevant node and will thus not assist the \del{} operation update the metadata. Such a \contains($k$) would return \codestyle{false} though \del($k$) has not yet updated the metadata, hence, is not yet linearized. 
Therefore, the metadata is updated before any unlinking attempt:
the \del($k$) operation itself updates the metadata after marking the node and before unlinking it;
and any other operation that attempts to help unlinking the marked node is also required to update the metadata on behalf of \del($k$) beforehand.

\paragraph{Adding size functionality}
An instance of a \codestyle{SizeCalculator} object (described in detail in \Cref{section:SizeCalculator}), responsible for keeping the metadata and computing the size, is referenced from the transformed data structure, and a \size{} method that uses it to retrieve the size is added to the data structure.

\begin{figure}
\begin{lstlisting}
INSERT = 0, DELETE = 1
@\underline{class TransformedDataStructure}@:
    @\underline{TransformedDataStructure()}@:
        initialize as originally
        sizeCalculator = new SizeCalculator()
    @\underline{contains(k)}@:
        @search$^*$ for a node with k as originally@
        if not found: return false
        else if found unmarked node:
            sizeCalculator.updateMetadata(node.insertInfo, INSERT)
            return true
        else: // found marked node
            sizeCalculator.updateMetadata(node@'@s deleteInfo, DELETE)
            return false
    @\underline{insert(k)}@:
        @search$^*$ for the place to insert k as originally@
        if k is already present in an unmarked node: 
            sizeCalculator.updateMetadata(node.insertInfo, INSERT)
            return failure
        if k is present in a marked node:
            sizeCalculator.updateMetadata(node@'@s deleteInfo, DELETE)
        insertInfo = sizeCalculator.createUpdateInfo(INSERT)
        @allocate newNode as originally with k and the other relevant data, and additionally with insertInfo@
        @insert newNode as originally (in case of failure proceed as originally)@
        sizeCalculator.updateMetadata(insertInfo, INSERT)
        return success
    @\underline{delete(k)}@:
        @search$^*$ for a node with k as originally@
        if not found: return failure
        if found a marked node:
            sizeCalculator.updateMetadata(node@'@s deleteInfo, DELETE)
            return failure
        sizeCalculator.updateMetadata(node.insertInfo, INSERT)
        deleteInfo = sizeCalculator.createUpdateInfo(DELETE)
        @mark node with deleteInfo (in case of failure proceed as originally)@
        sizeCalculator.updateMetadata(deleteInfo, DELETE)
        unlink node
        return success
    @\underline{size()}@:
        return sizeCalculator.compute()
@$^*$For each encountered marked node along the search, in case of unlinking it in the original algorithm, call sizeCalculator.updateMetadata(node's deleteInfo, DELETE) before unlinking it.@
\end{lstlisting}
\caption{A transformed data structure}\label{fig: transformed data structure}
\end{figure}

\subsection{Specific Examples and the \codestyle{SizeCalculator} Object}
\Cref{fig: transformed data structure} demonstrates how the transformation described above may be applied to standard linearizable linked list, skip list and hash table that implement a set. 
A similar transformation with minor adaptations will apply to implementations of a dictionary.
The transformation may also be applied to search trees with some adaptations.

At the core of our size mechanism stands the  \codestyle{SizeCalculator} object. We elaborate on this object, responsible for the size calculation, in \Cref{section:SizeCalculator}.
For now we just need to be familiar with its interface methods: \codestyle{updateMetadata} is called with an \codestyle{UpdateInfo} object associated with an \ins{} or \del{} operation for updating the metadata stored in the \codestyle{SizeCalculator} to reflect that operation. This method may be called by both the operation initiator and helpers. We explain in \Cref{section:metadata} how \codestyle{SizeCalculator} prevents double update of the metadata on behalf of the same operation. 
\codestyle{createUpdateInfo} is called by \ins{} and \del{} operations to produce an object that will be published to helpers, with the information required for updating the metadata on their behalf. 
\codestyle{compute} is the method used to retrieve the size of the data structure efficiently (using the metadata).

A \codestyle{SizeCalculator} reference field named  \codestyle{sizeCalculator} is placed in the data structure, and initialized to hold a \codestyle{SizeCalculator} instance. Its methods are called in the appropriate places in the data-structure operations, as can be seen in \Cref{fig: transformed data structure}.
In addition, an \codestyle{insertInfo} field referencing an \codestyle{UpdateInfo} object is placed in the data structure's node objects. A similar \codestyle{deleteInfo} field is placed in the appropriate place, as described above.
Since the \codestyle{UpdateInfo} record contains the information required for updating the metadata to reflect the associated operation, its content is coupled with the size metadata, so its description is deferred to \Cref{section:metadata}.

\subsection{Applicability}
We focus on a transformation for linearizable data structures that implement the highly prevalent set or dictionary data types. However, the presented ideas may be adapted to other data types. 
Our transformation recipe requires that the \del{} operation of the original data structure be linearized at a marking step and not at an unlinking step, to ensure consistency with the size metadata. Otherwise, if operations on $k$ that encounter a marked node with the key $k$ ignore the mark, and consider $k$ as deleted only when its node is unlinked, that might be inconsistent with the size metadata which is updated to reflect the deletion before the unlinking.
Instead, by our requirement, operations on $k$ in the original data structure consider the node as removed when it is marked, and in the transformed data structure they help update the metadata on behalf of the \del($k$) that marked the node and only then treat the key $k$ as deleted. 

In case of a data structure that linearizes the \del{} operation at an unlinking step and not in the prior marking step, it is usually not difficult to adjust it to have the marking as the linearization point of \del{}. We made this adjustment to the binary search tree of \citet{ellen2010non} in order to apply the transformation to it and evaluate its performance.

\section{The size metadata}\label{section:metadata}
In our transformation, operations may help other operations update the metadata. Hence, we must prevent a double update of the metadata on behalf of the same operation. 
We use metadata which enables threads that operate on the data structure to determine whether the metadata already reflects a certain operation, and update it otherwise.
The \codestyle{SizeCalculator} object holds the array \codestyle{metadataCounters} as the metadata, containing two counters per thread: an insertion counter and a deletion counter, which indicate the number of successful insertions and deletions the thread has performed so far on the data structure.
Separating the insertion from the deletion counter allows determining whether an \ins{} (or a \del{}) operation has been reflected in the counters. If an \ins{} follows a \del{}, a single counter (incremented on each insertion and decremented on each deletion) cannot indicate if the two operations completed or none of them.
Two separate counters allow a simple concise indication of which one of the two operations is reflected in the counters. 
Next we describe how insertions are handled; deletions are handled similarly.

The per-thread monotonic insertion counters enable to immediately detect whether a certain \ins{} operation by a certain thread is reflected in the metadata, and otherwise ensure that it is reflected via a single \cas{}:
When \codestyle{updateMetadata} is called on behalf of a thread $T$'s $i$-th successful \ins{} operation by either $T$ or helpers,
if $T$'s insertion counter is $\geq i$, it leaves the counter as is since the operation is already reflected in the metadata; else, it uses a \cas{} to increment it from $i-1$ to $i$. There is no need to repeat the \cas{} in case of failure, since that might happen only when another thread succeeds to perform the same update.

To help another operation update the metadata, a helper needs to know on which counter in \codestyle{metadataCounters} it should operate and its target value.
This dictates the information that the $i$-th \ins{} operation by thread $T$ leaves for helpers in an \codestyle{UpdateInfo} object: $T$'s thread ID, which will be used as an index to the \codestyle{metadataCounters} array, and $i$, which is the counter's target value (which is simply the current value of $T$'s insertion counter in \codestyle{metadataCounters} plus $1$).

The size may be calculated from \codestyle{metadataCounters} as the difference between the sum of insertion counters and the sum of deletion counters. But naively reading the values one by one may result in an inconsistent (non-linearizable) size, because we may obtain a collection of values that never existed simultaneously in the array. We need to obtain a snapshot of values the array counters had at some point in time, but we cannot use locks to achieve this atomicity as we aim for a wait-free \size{}. Next we explain how we manage to achieve that.

\section{Mechanism for wait-free size}\label{section:wait-free size}
The \size{} operation needs to obtain a linearizable snapshot of the \codestyle{metadataCounters} array, from which it will be able to compute a consistent size. As the size of this array is twice the number of threads, our solution is the most beneficial (in comparison to computing the size by iterating over a snapshot of the data structure) for applications that usually use data structures with much more elements than the number of threads. If \size{} naively read \codestyle{metadataCounters} cell by cell, it could obtain an inconsistent view of the array. For example, consider an execution in which a \size{} operation starts scanning the array, but after it reads the insertion counter of some thread $T$, this thread inserts an item and then removes it. Now both $T$'s insertion and deletion counters equal $1$, and when the \size{} operation resumes it reads the new value of $T$'s deletion counter and returns $-1$ as the size. The problem here is that the \size{} operation captured the \del{}'s update of the array but missed the preceding \ins{}'s update.

To overcome this problem and obtain a linearizable snapshot of the counters array in a wait-free form, we adopt the basic idea of \cite{jayanti2005optimal}'s single-scanner single-writer snapshot algorithm, which is as follows.
After an \up{} operation writes to the main array, it checks if a concurrent \scan{} operation is in the process of collecting the main array's values. If so, the \scan{} has maybe already read the relevant cell and missed the new value, thus the \up{} forwards the new value from the main array to a designated second array. The \scan{} operation begins with a collection phase to collect the main array values; before starting the collection it announces it to other operations, and after the collection it announces its completion. In a second phase, the \scan{} retrieves a linearizable view of the array by combining the values it collected with newer values, forwarded to the designated second array by concurrent \up{} operations (namely, each forwarded value is adopted in place of the value that the \scan{} collected from the corresponding cell in the main array).
A \scan{} is linearized at the point it announces completing the collection. 
It might miss values that were written to the main array by some \up{} operations while it was collecting, thus, such operations are retrospectively linearized immediately after the \scan{}'s linearization point. 
We bring the linearization details of \up{} adapted to our context in \Cref{section:linearizability}.

Our \size{} operation acts in the spirit of Jayanti's \scan{} to obtain a view of the metadata array, and data-structure operations that update the metadata array (on behalf of their own operation or to help another operation) act in the spirit of Jayanti's \up{} to inform a concurrent \size{} of a new value it might have missed.
However, Jayanti's basic idea supports a single scanner.
When multiple \size{} operations execute concurrently, we cannot let each \size{} take an independent snapshot of the metadata array, because the linearization point of a \size{} operation determines the linearization points of updating operations it missed, and concurrent independent \size{} operations might determine contradicting linearization points for concurrent updates. Thus, we need to make sure that concurrent \size{} operations yield the same consistent snapshot of the metadata array.

To this end, we introduce a \codestyle{CountersSnapshot} object (on which we detail in \Cref{section:CountersSnapshot}). Concurrent \size{} operations coordinate with each other through a \codestyle{CountersSnapshot} instance, similarly to concurrent snapshot operations in \cite{petrank2013lock} that use a shared object to orchestrate taking a snapshot concurrently. 
A \size{} operation needs to first obtain a \codestyle{CountersSnapshot} instance to operate on.
At any given point in time, at most one collecting \codestyle{CountersSnapshot} instance (in which the collection has not yet completed) is announced. If a \size{} operation observes such an instance, it operates on it, so that it returns the same size as the \size{} that announced this instance. Otherwise, the \size{} operation produces a new instance, announces it and operates on it.

The \codestyle{CountersSnapshot} holds a snapshot array for taking a snapshot of the metadata array. \size{} operations that operate on a certain \codestyle{CountersSnapshot} instance collect values into its snapshot array (using a \cas{} from an initial \codestyle{INVALID} value to the value obtained from the metadata array), and operations that concurrently update the metadata array forward their values into the snapshot array. 
After a collection phase, a \size{} operation needs to compute the size based on the counters in the snapshot array. But the array is not stable---updating operations might be still forwarding values. For all \size{} operations that operate on the same \codestyle{CountersSnapshot} instance to agree on the same size, we place a \codestyle{size} field in \codestyle{CountersSnapshot}, initialized to \codestyle{INVALID}. The first \size{} operation to compute a size by traversing the snapshot array and then perform a \cas{} of the \codestyle{size} field from \codestyle{INVALID} to its computed size, determines the size value. Concurrent \size{} operations will adopt this value. Any value forwarded to a counter in the snapshot array after the thread that determined the size read this counter is ignored (and its related operation is linearized after the \size{}).

\begin{figure}[t]
\begin{lstlisting}
@\underline{\textbf{class} UpdateInfo}@:
    int tid
    long counter
@\underline{\textbf{class} SizeCalculator}@:
    long[][] metadataCounters
    CountersSnapshot countersSnapshot 
@\underline{\textbf{class} CountersSnapshot}@:
    long[][] snapshot
    boolean collecting
    long size
\end{lstlisting}
\caption{Classes fields}\label{fig:Classes fields}
\end{figure}

\begin{figure}
\begin{lstlisting}
@\underline{\textbf{class} SizeCalculator}@:
    @\underline{SizeCalculator()}@:@\label{code:SizeCalculator ctor}@
        this.metadataCounters = new long[n][PADDING] // implicitly initialized to zeros
        this.countersSnapshot = new CountersSnapshot()
        this.countersSnapshot.collecting.setVolatile(false)
    @\underline{compute()}@:
        activeCountersSnapshot = _obtainCollectingCountersSnapshot()@\label{code:SizeCalculator compute start collection}@
        _collect(activeCountersSnapshot)@\label{code:SizeCalculator compute collect}@
        activeCountersSnapshot.collecting.setVolatile(false)@\label{code:SizeCalculator compute stop collection}@
        return activeCountersSnapshot.computeSize()
    @\underline{\_obtainCollectingCountersSnapshot()}@:
        currentCountersSnapshot = this.countersSnapshot.getVolatile()@\label{code:SizeCalculator obtain get existing start}@
        if currentCountersSnapshot.collecting.getVolatile():@\label{code:SizeCalculator obtain current collecting}@
            return currentCountersSnapshot@\label{code:SizeCalculator obtain get existing end}@
        newCountersSnapshot = new CountersSnapshot()@\label{code:SizeCalculator obtain get new start}@
        witnessedCountersSnapshot = this.countersSnapshot.compareAndExchange(currentCountersSnapshot, newCountersSnapshot):@\label{code:SizeCalculator obtain cas}@
        if witnessedCountersSnapshot == currentCountersSnapshot:
            return newCountersSnapshot
        return witnessedCountersSnapshot // our exchange failed, adopt the value written by a concurrent thread@\label{code:SizeCalculator obtain get new end}@
    @\underline{\_collect(targetCountersSnapshot)}@:
        for each tid:
            for opKind in (INSERT, DELETE):
                targetCountersSnapshot.add(tid, opKind, this.metadataCounters[tid][opKind].getVolatile())@\label{code:SizeCalculator collect add}@
    @\underline{updateMetadata(updateInfo, opKind)}@:
        tid = updateInfo.tid
        newCounter = updateInfo.counter
        if this.metadataCounters[tid][opKind].getVolatile() == newCounter - 1:@\label{code:SizeCalculator updateMetadata update array start}@
           this.metadataCounters[tid][opKind].compareAndSet(newCounter - 1, newCounter)@\label{code:SizeCalculator updateMetadata update array end}@
        currentCountersSnapshot = this.countersSnapshot.getVolatile()@\label{code:SizeCalculator updateMetadata obtain CountersSnapshot}@
        if currentCountersSnapshot.collecting.getVolatile() and @\label{code:SizeCalculator updateMetadata check collecting}@
              this.metadataCounters[tid][opKind].getVolatile() == newCounter:@\label{code:SizeCalculator updateMetadata obtain and check counter}@
            currentCountersSnapshot.forward(tid, opKind, newCounter)@\label{code:SizeCalculator updateMetadata forward}@
    @\underline{createUpdateInfo(opKind)}@:
        return new UpdateInfo(threadID, this.metadataCounters[threadID][opKind].getVolatile() + 1)
\end{lstlisting}
\caption{\codestyle{SizeCalculator} methods}\label{fig:SizeCalculator}
\end{figure}

\subsection{\codestyle{SizeCalculator} Details}\label{section:SizeCalculator}
Each transformed data structure holds a \codestyle{SizeCalculator} instance associated with it, responsible for calculating the size by holding the metadata and operating on it.
The fields of \codestyle{SizeCalculator} (as well as the other classes we use) are detailed in \Cref{fig:Classes fields}, and its pseudocode appears in \Cref{fig:SizeCalculator}.

The \codestyle{SizeCalculator} object contains two fields: The first is \codestyle{metadataCounters}, holding the size metadata---an array with an insertion counter and a deletion counter per thread, with padding between the cells of each thread and the next one so that the counters of the different threads are placed in separate cache lines to avoid false sharing. The second field is \codestyle{countersSnapshot}, that holds the most recent \codestyle{CountersSnapshot} instance. In its constructor method (appearing in \Cref{code:SizeCalculator ctor}), \codestyle{SizeCalculator} initializes \codestyle{metadataCounters} with zeros, and \codestyle{countersSnapshot} with a dummy instance with its \codestyle{collecting} flag set to \codestyle{false}, to signal that it is not collecting and future \size{} operations should use a new instance.

The \codestyle{compute} method is called by the \size{} operation of a transformed data structure. It starts with a collection phase in \Crefrange{code:SizeCalculator compute start collection}{code:SizeCalculator compute stop collection}. First it needs to announce a new collection if there is no ongoing collection. To this end it calls the private method \codestyle{\_obtainCollectingCounters\-Snapshot}. The latter returns the most recent \codestyle{CountersSnapshot} if this instance is still collecting (\Crefrange{code:SizeCalculator obtain get existing start}{code:SizeCalculator obtain get existing end}), so that the current \codestyle{compute} would cooperate with ongoing \codestyle{compute} calls. Otherwise, \codestyle{\_obtainCollectingCounters\-Snapshot} tries to place a new \codestyle{CountersSnapshot} instance in \codestyle{counters\-Snapshot} using a \cas{}, and returns the new \codestyle{countersSnapshot} value, whether it is an instance placed by itself or an instance placed by another \codestyle{compute} call (\Crefrange{code:SizeCalculator obtain get new start}{code:SizeCalculator obtain get new end}).
With an \codestyle{activeCounters\-Snapshot} instance in a collecting mode, \codestyle{compute} calls the private method \codestyle{\_collect} (\Cref{code:SizeCalculator compute collect}), to add all \codestyle{metadataCounters} values to \codestyle{activeCounters\-Snapshot}. Then, it unsets \codestyle{activeCountersSnapshot}'s \codestyle{collecting} flag.
Now that its collection phase is complete, \codestyle{compute} computes the size according to the \codestyle{CountersSnapshot} instance maintained in \codestyle{activeCounters\-Snapshot}.
This is done using the \codestyle{computeSize} method of \codestyle{CountersSnapshot}, on which we detail in \Cref{section:CountersSnapshot}.

\codestyle{updateMetadata(UpdateInfo(tid, c), INSERT)} is called on behalf of the $c$-th successful \ins{} operation by thread $tid$. We describe how the method handles insertions for convenience; the same applies for deletions by passing \codestyle{opKind=DELETE}. The method first updates the relevant counter in the metadata array, i.e., \codestyle{metadataCounters[tid][INSERT]}, to be $c$ (\Crefrange{code:SizeCalculator updateMetadata update array start}{code:SizeCalculator updateMetadata update array end}), using a \cas{} to avoid overriding concurrent updates. 
At this point, the metadata reflects the discussed insertion.
Then, according to the described-above scheme, \codestyle{updateMetadata} should also forward the counter value $c$ to concurrent \size{} operations that take a snapshot of the \codestyle{metadataCounters} array and might have missed this value. 
For that, it performs the following steps:
(1) obtain the current collecting \codestyle{CountersSnpashot} instance (\Cref{code:SizeCalculator updateMetadata obtain CountersSnapshot});
(2) verify it is still collecting (\Cref{code:SizeCalculator updateMetadata check collecting});
(3) obtain the relevant counter from the metadata array and verify it still holds the value $c$ (\Cref{code:SizeCalculator updateMetadata obtain and check counter}); and then finally, if these checks pass,
(4) call the \codestyle{forward} method of the \codestyle{CountersSnpashot} instance obtained in the first step (\Cref{code:SizeCalculator updateMetadata forward}).
This series of steps is intended to prevent redundant forwarding. Though it is not yet clear now, it guarantees a constant time complexity for the \codestyle{forward} method, as we prove in \Cref{section:time complexity}.

The last method of \codestyle{SizeCalculator} is \codestyle{createUpdateInfo}, which is called by \ins{} and \del{} operations to obtain an \codestyle{UpdateInfo} instance for publication to helpers. \codestyle{createUpdateInfo} creates an \codestyle{UpdateInfo} instance with \codestyle{tid=threadID} and \codestyle{counter}=$c$, where \codestyle{threadID} is the ID of the calling thread (\codestyle{threadID} values are assumed to start from $0$, and could be obtained e.g. from a thread-local variable), and $c$ is the current value of the relevant counter (that indicates how many successful operations of the requested kind have been executed by the calling thread so far) plus $1$---as the calling thread is about to attempt its \codestyle{c}-th operation of this kind.

\begin{figure}
\begin{lstlisting}
@\underline{\textbf{class} CountersSnapshot}@:
    @\underline{CountersSnapshot()}@:@\label{code:CountersSnapshot ctor}@
        this.snapshot = new long[n][2]
        setVolatile all cells of this.snapshot to INVALID
        this.collecting.setVolatile(true)
        this.size.setVolatile(INVALID)
    @\underline{add(tid, opKind, counter)}@:
        if this.snapshot[tid][opKind].getVolatile() == INVALID:
            this.snapshot[tid][opKind].compareAndSet(INVALID, counter)@\label{code:CountersSnapshot snapshot cas in add}@
    @\underline{forward(tid, opKind, counter)}@:
        snapshotCounter = this.snapshot[tid][opKind].getVolatile()
        while (snapshotCounter == INVALID or @ @ counter > snapshotCounter): // will execute at most two iterations
            witnessedSnapshotCounter = this.snapshot[tid][opKind].compareAndExchange(snapshotCounter, counter):@\label{code:CountersSnapshot snapshot cas in forward}@
            if witnessedSnapshotCounter == snapshotCounter:
                break
            snapshotCounter = witnessedSnapshotCounter
    @\underline{computeSize()}@:
        computedSize = 0@\label{code:CountersSnapshot computeSize compute start}@
        for each tid:
            computedSize += this.snapshot[tid][INSERT].getVolatile() - @ @ this.snapshot[tid][DELETE].getVolatile()@\label{code:CountersSnapshot computeSize compute end}@
        witnessedSize = this.size.compareAndExchange(INVALID, computedSize)@\label{code:CountersSnapshot computeSize CAS}@
        if witnessedSize == INVALID:
            return computedSize
        return witnessedSize // our exchange failed, adopt the size written by a concurrent thread
\end{lstlisting}
\caption{\codestyle{CountersSnapshot} methods}\label{fig:CountersSnapshot}
\end{figure}

\subsection{\codestyle{CountersSnapshot} Details} \label{section:CountersSnapshot}
A new \codestyle{CountersSnapshot} instance is announced in \codestyle{SizeCalculator.countersSnapshot} each time a new collection phase starts (which happens every time a \size{} operation starts and observes that the last announced \codestyle{CountersSnapshot} instance is already not collecting). This instance coordinates the current size calculation among all concurrent \size{} calls that use it to compute the size.
Its methods appear in \Cref{fig:CountersSnapshot} and its fields appear in \Cref{fig:Classes fields}. 

The \codestyle{CountersSnapshot} object holds a snapshot array called \codestyle{snapshot} for taking a snapshot of the metadata array, from which the size will be computed. It also holds a \codestyle{collecting} field that indicates whether the collection of values into \codestyle{snapshot} is still ongoing, and a \codestyle{size} field that will eventually hold the computed size.
In its constructor method (appearing in \Cref{code:CountersSnapshot ctor}), \codestyle{CountersSnapshot} initializes all its fields. The cells of \codestyle{snapshot} are set to \codestyle{INVALID} (which may have the value \codestyle{Long.MAX\_VALUE} for instance), the \codestyle{collecting} flag is set to \codestyle{true} and \codestyle{size} is set to \codestyle{INVALID}.

The \codestyle{add} method is called by \size{} operations to collect values into the snapshot array. It performs a \cas{} on the requested cell to the requested value only if the current value is \codestyle{INVALID}. Otherwise, another operation has already collected a value to this cell and there is no need to perform another \cas{}. 
Indeed, the value that this \size{} operation fails to add might be missed during the size calculation if it is not forwarded on time by the updating operation associated with it, but this does not foil linearizability, as the updating operation associated with this value is retrospectively linearized after the \size{}.

\codestyle{forward(tid,INSERT,c)} is called by \codestyle{updateMetadata} that was called on behalf of the $c$-th successful \ins{} operation by thread $tid$. It is called after the insertion counter of thread $tid$ in the metadata array is set to $c$, to ensure that the insertion counter of that thread in the snapshot array contains a value $\geq c$. 
We prove in \Cref{section:time complexity} that the \codestyle{forward} method shall execute at most two \cas{} attempts before reaching this goal. 
\codestyle{forward} operates similarly for deletions when called with an \codestyle{opKind=DELETE} argument.

The \codestyle{computeSize} method is called by the \codestyle{compute} method of \codestyle{SizeCalculator} (which is called by the data structure's \size{} method), after obtaining a \codestyle{CountersSnapshot} instance and completing the collection to this instance, so that its \codestyle{snapshot} array is filled with meaningful values (rather than \codestyle{INVALID} values). 
The size is computed as the difference between the sum of insertion counters and the sum of deletion counters in the \codestyle{snapshot} array (\Crefrange{code:CountersSnapshot computeSize compute start}{code:CountersSnapshot computeSize compute end}). 
But \codestyle{computeSize} may be called by multiple concurrent \size{} operations that operate on the same \codestyle{CountersSnapshot} instance, and each of them might compute a different size because values may be concurrently forwarded to the array. Only the first \codestyle{computeSize} call to fix the size it computed in the \codestyle{size} field (in \Cref{code:CountersSnapshot computeSize CAS}), determines the size value that they will all adopt. The rest of them will fail to \cas{} and will adopt its value.

\subsection{Memory Model}\label{section:memory model}
The pseudocode brought in this section aligns with our Java implementation \cite{artifactConcurrentSize} (evaluated in \Cref{section:performance}) and accesses variables using volatile memory semantics to ensure the visibility required for correctness in accordance with the Java memory model. Read, write and \cas{} operations on non-final fields of shared objects are performed with volatile semantics (in our Java implementation this is achieved using volatile variables, \codestyle{VarHandle}s and \codestyle{AtomicReferenceFieldUpdater}s). These volatile-semantics accesses are considered synchronization actions, over which the Java memory model guarantees a synchronization order (a total order which is consistent with the program order of each thread, and where a read from a volatile variable returns the last value written to it by the synchronization order). 
A similar implementation could be designed in C++ according to its memory model guarantees, utilizing the \codestyle{std::atomic} library to order accesses to shared memory.

\section{Optimizations}\label{section:optimizations}

The following optimizations may be applied in our methodology, and we apply them in our implementation \cite{artifactConcurrentSize} measured in \Cref{section:performance}.

\subsection{Eliminate Metadata Update on Behalf of Completed Insertions}
When an insertion is complete, there is no need that future operations on the inserted node update the metadata on behalf of that insertion.
To this end, after a thread calls \codestyle{updateMetadata} to update the metadata on behalf of some \ins{} operation that inserted a node $N$, it may set $N$.\codestyle{insertInfo} to \nul{}, to signal that the metadata already reflects the insertion and there is no need to call \codestyle{updateMetadata}.
Before calling \codestyle{updateMetadata}, threads will perform a \nul{} check to the node's \codestyle{insertInfo} to rule if the call is necessary.

We do not propose a similar modification for deletions since deleted nodes are unlinked from the data structure when the deletion completes and cause no more update activity, unlike inserted nodes which, without the optimization, cause a redundant \codestyle{updateMetadata} call on each operation on the node.

\subsection{Size Backoff}
Each \size{} operation operates on a \codestyle{CountersSnapshot} instance it obtains as follows. It collects values into its \codestyle{snapshot} array using \cas{} operations, uses the collected counters to compute the size, and finally sets its \codestyle{size} field to the computed size using a \cas{}, unless another \size{} operation has done that beforehand.

Exponential backoff may be used to reduce contention among concurrent \size{} operations caused by their \cas{} operations on the \codestyle{snapshot} and \codestyle{size} fields.
Each time a \size{} operation obtains an existing \codestyle{CountersSnapshot} instance that was announced by another \size{} operation, it may wait a while to let another \size{} operation complete the size calculation. After waiting, if the calculation is not yet complete (which may be detected by an \codestyle{INVALID} value in the \codestyle{size} field), it shall collect, compute the size and try to set it on its own.

\subsection{Check for an Already-Set Size}
There are occasions where we may avoid contention and redundant work by obtaining the \codestyle{size} field of \codestyle{CountersSnapshot} and returning it in case it does not equal \codestyle{INVALID}. This may be done when \codestyle{SizeCalculator}'s \codestyle{\_obtainCollectingCountersSnapshot} method observes a concurrent \size{} operation in \Cref{code:SizeCalculator obtain get existing end,code:SizeCalculator obtain get new end}; at the beginning of \codestyle{CountersSnapshot}'s \codestyle{computeSize} method; and right before \codestyle{computeSize} performs a \cas{} attempt.
\section{Methodology properties}\label{section:properties}

\subsection{Linearizability}\label{section:linearizability}
A linearizable data structure transformed according to our methodology to support a \size{} operation, remains linearizable.
Recall that an operation has its original linearization point, when its linearization is defined in the original set data structure, but we linearize operations in the transformed data structure only when the metadata is updated.
Next, we detail the linearization points of a transformed set's operations, and use them to prove linearizability.

\subsubsection{Linearization Points}\label{section:linearization points}
A \size{} operation is linearized at the announcement of the collection completion. 
For a successful \ins{} or \del{} operation, the associated metadata counter is updated to reflect the operation (by either the operation initiator or helpers), and if this update happens when no \size{} is collecting, then the operation is simply linearized at the update.
However, if the update is performed while some \size{} is collecting, then the operation is linearized according to that \size{} to comply with its linearization point: if the \size{} takes the operation into account then the operation is simply linearized at the metadata counter update; otherwise, it is retrospectively linearized immediately after the linearization point of that \size{}.
Finally, a \contains{} operation and a failing \ins{} or \del{} operation (namely, one that fails to insert a new key or delete an existing key respectively, and returns a failure), are linearized like in the original data structure, unless the operation they ``depend on'', namely, the last successful update operation on the same key whose original linearization point precedes their original linearization point (a concurrent successful \ins{} of the same key in case of \contains{} returning \codestyle{true} and a failing \ins{};
and a concurrent successful \del{} in case of \contains{} returning \codestyle{false} and a failing \del{}) is not yet linearized at their original linearization point, in which case they are linearized immediately after this operation is linearized.

In more detail, a \size{} operation is linearized when the \codestyle{collecting} field of the \codestyle{CountersSnapshot} instance it operates on is set to \codestyle{false} for the first time (in \Cref{code:SizeCalculator compute stop collection}).
Regarding a successful \ins{} operation, the associated metadata counter is updated as follows: a \cas{} of \codestyle{sizeCalculator.metadata\-Counters[tid][INSERT]} to $c$ is performed on behalf of the $c$-th successful \ins{} operation of thread $tid$ (in \Cref{code:SizeCalculator updateMetadata update array end}), where \codestyle{sizeCalculator} is the \codestyle{SizeCalculator} instance held by the transformed data structure. For a successful \del{} operation, the only difference is that \codestyle{DELETE} is used as the array index.
As for the linearization point of such an \ins{} or \del{} operation---if a \codestyle{CountersSnapshot} instance with a \codestyle{collecting} field set to \codestyle{true} is not announced in \codestyle{sizeCalculator.countersSnapshot} when the metadata counter update is performed, then the operation is linearized at the metadata counter update (namely, at the \cas{} in \Cref{code:SizeCalculator updateMetadata update array end}). Else, the operation is linearized according to this \codestyle{CountersSnapshot} instance: if the \size{} operation that sets its \codestyle{size} field read a value $\geq c$ from the relevant counter (in \Cref{code:CountersSnapshot computeSize compute end}), then the operation is linearized at the metadata counter update (as in the previous case); otherwise, it is linearized immediately after the linearization point of that \size{} operation.

In the above specification, several operations might be linearized at the same moment---either operations defined to be linearized immediately after each other, or operations linearized at the exact same moment (e.g., several \size{} operations operating on the same \codestyle{CountersSnapshot} instance). We order operations that are linearized at the same moment one after the other as follows: \size{} operations (if any) are placed first; the order among them is arbitrary. Successful update operations (if any) are placed after the size operations according to their metadata-counter update order. Each \contains{} or failing \ins{} or \del{} call that is not linearized at its original linearization point (if any) is placed right after the successful update operation it depends on; the order among such operations which are placed after the same successful update is arbitrary.

\subsubsection{Linearizability Proof}\label{section:Linearizability proof}
We prove that our transformation in linearizable using the equivalent definition of linearizability that is based on linearization points (see \cite[Section 7]{sela2021linearizabilityfull} and the atomicity definition in \cite{lynch1996distributed}). We need to show that (1) each linearization point occurs within the operation's execution time, and that (2) ordering an execution's operations (with their results) according to their linearization points forms a legal sequential history.

We prove Property (1) in \Cref{claim:lin within op} and Property (2) in \Cref{claim:results comply} in \Cref{section:abstract set lin proof}.

\begin{claim}\label{claim:lin within op}
The linearization point of each operation occurs within its execution time.
\end{claim}

\begin{proof}

We begin with the linearization point of a \size{} operation. 
\size{} calls \codestyle{sizeCalculator.com\-pute}, which starts with calling \codestyle{\_obtainCollectingCountersSnapshot} to obtain a collecting \codestyle{Counters\-Snapshot} instance. \codestyle{\_obtainCollectingCountersSnapshot} returns an instance after its \codestyle{collecting} field has had the value \codestyle{true} at some point during this \codestyle{\_obtainCollectingCountersSnapshot} call: If this call observes that the current announced instance is collecting (in \Cref{code:SizeCalculator obtain current collecting}), it returns this instance. Otherwise, this instance cannot be used by the current \size{} because its linearization point has passed and has possibly occurred before the current \size{} started.
Thus, it creates an instance with \codestyle{collecting} set to \codestyle{true}, and if it succeeds to announce it using a \cas{} (in \Cref{code:SizeCalculator obtain cas}), it returns this instance. Else, the failure of its \cas{} indicates that another thread has in the meanwhile announced a new instance, with \codestyle{collecting} set to \codestyle{true}, and the discussed \codestyle{\_obtainCollectingCountersSnapshot} call returns such an instance.
We showed that in any of the above cases, the \codestyle{collecting} field of the obtained \codestyle{CountersSnapshot} instance was still \codestyle{true} at some point during the \codestyle{\_obtainCollecting\-Counters\-Snapshot} call, hence the \size{}'s linearization point does not occur before the \codestyle{compute} call starts. It does occur before it ends, as the \codestyle{collecting} field is set to \codestyle{false} either when this call executes \Cref{code:SizeCalculator compute stop collection}, or before if another \codestyle{compute} call has executed this line earlier.

Next, we prove that successful update operations are linearized within their execution time.
A successful \ins{} or \del{} operation calls \codestyle{updateMetadata} with its \codestyle{UpdateInfo} instance before returning. As we prove in \Cref{lemma:lin when updateMetadata returns} below, by the time \codestyle{updateMetadata} returns, the operation is guaranteed to be linearized.
Additionally, it is not linearized before the operation's execution starts, since it is linearized either at its metadata counter update or at a later point in time, and the update on behalf of a certain operation can only happen after it started and published its \codestyle{UpdateInfo} instance. 

Lastly, we show that a \contains{}, a failing \ins{} and a failing \del{} operations are linearized within their execution time. If an operation $op$ of this kind is linearized at its original linearization point, we are done\footnote{For every linearizable data structure, there exists a selection of linearization points such that each of them is placed during the execution time of the corresponding operation (see the equivalent definition of linearizability based on linearization points in Section 7 in \cite{sela2021linearizabilityfull}). Each time we refer to the linearization points of the original data structure, we refer to points that satisfy this requirement.}. Otherwise, $op$ is linearized immediately after the linearization point of an operation $op_2$ it depends on. This happens only in case $op$ observes $op_2$ and calls \codestyle{updateMetadata} on behalf of $op_2$. By \Cref{lemma:lin when updateMetadata returns}, $op_2$ is linearized by the time this \codestyle{updateMetadata} call returns. Hence, $op$ is linearized by that time as well.

\end{proof}

In the proof of \Cref{claim:lin within op} we use the following lemmas (\Cref{lemma:counter inc after first 2 lines of updateMetadata} is proven in \Cref{section:abstract set lin proof}):

\begin{lemma}\label{lemma:lin when updateMetadata returns}
When a call to \codestyle{updateMetadata} returns, the operation whose \codestyle{updateInfo} was passed to the call is guaranteed to be linearized.
\end{lemma}
\begin{proof}
Consider a call to \codestyle{updateMetadata} on behalf of $op$, being the $c$-th successful \ins{} operation by a thread $T$ (a similar proof applies for \del{}). Denote this call by \term{updateMetadataForOp}. We need to show that $op$ has been linearized by the time \term{updateMetadataForOp} returns.
By \Cref{lemma:counter inc after first 2 lines of updateMetadata}, after executing \Crefrange{code:SizeCalculator updateMetadata update array start}{code:SizeCalculator updateMetadata update array end}, the relevant metadata counter's value is $\geq c$. 
If $op$ is linearized when its related metadata counter is set to $c$, we are done. Otherwise, the following hold: (1) the counter is set to $c$ when a \codestyle{CountersSnapshot} instance with a \codestyle{collecting} field set to \codestyle{true} is announced in the \codestyle{countersSnapshot} field of \codestyle{sizeCalculator} (denote this instance by \term{snapshotAtUpdate}); (2) the \size{} operation that sets \term{snapshotAtUpdate}'s \codestyle{size} field reads, during its size computation, a value $< c$ from the corresponding snapshot counter; and (3) $op$ is linearized immediately after the linearization point of that \size{} operation, namely, immediately after the \codestyle{collecting} field of \term{snapshotAtUpdate} is set to \codestyle{false} for the first time.
So we need to prove that this \codestyle{collecting} field is set to \codestyle{false} before the \term{updateMetadataForOp} call returns.
Next, we prove this holds in the various possible scenarios.

If \term{updateMetadataForOp} obtains a newer \codestyle{CountersSnapshot} instance than \term{snapshotAtUpdate} in \Cref{code:SizeCalculator updateMetadata obtain CountersSnapshot}, then we are done,
as \codestyle{CountersSnapshot} instances announced in \codestyle{SizeCalculator} are replaced only after their \codestyle{collecting} field is set to \codestyle{false}.

Else, if \term{updateMetadataForOp} observes in \Cref{code:SizeCalculator updateMetadata check collecting} that \term{snapshotAtUpdate}'s \codestyle{collecting} field value is \codestyle{false}, we are done.

Else, if the checks in \Cref{code:SizeCalculator updateMetadata check collecting,code:SizeCalculator updateMetadata obtain and check counter} pass, \term{updateMetadataForOp} forwards the value $c$ to the snapshot counter in \Cref{code:SizeCalculator updateMetadata forward}. When its forwarding completes, the snapshot counter contains a value $\geq c$. Since we analyze here a case in which the \size{} operation, which sets \term{snapshotAtUpdate}'s \codestyle{size} field, reads a value $< c$ from the snapshot counter, this \size{} must have read the snapshot counter before the forwarding completes, and this read during the size computation occurs only after the \term{snapshotAtUpdate}'s \codestyle{collecting} field is set to \codestyle{false}.

The remaining scenario is that \term{updateMetadataForOp} observes in \Cref{code:SizeCalculator updateMetadata obtain and check counter} that the metadata counter's value is $\geq c+1$. We will show that \term{snapshotAtUpdate}'s \codestyle{collecting} field value has been earlier set to \codestyle{false}.
For the metadata counter to reach the value $c+1$, the thread $T$ must have already started its ($c+1$)-st successful insertion.
Prior to that, $T$ has completed $op$ (which is its $c$-th successful insertion), during which it has called \codestyle{updateMetadata}, in a call we denote \term{updateMetadataByT}. We will next show that by the time \term{updateMetadataByT} returns, the value of \term{snapshotAtUpdate}'s \codestyle{collecting} field is already \codestyle{false}.
After \term{updateMetadataByT} executes \Crefrange{code:SizeCalculator updateMetadata update array start}{code:SizeCalculator updateMetadata update array end}, the metadata counter's value is $\geq c$ by \Cref{lemma:counter inc after first 2 lines of updateMetadata}.
Denote by \term{currSnap} the value that \term{updateMetadataByT} obtains in \codestyle{currentCountersSnapshot} in \Cref{code:SizeCalculator updateMetadata obtain CountersSnapshot}. \term{currSnap} must be \term{snapshotAtUpdate}, because otherwise, if it were an earlier \codestyle{CountersSnapshot} instance, then when \term{snapshotAtUpdate} is later announced in the \codestyle{countersSnapshot} field of \codestyle{sizeCalculator}, the metadata counter's value would already be $\geq c$ as mentioned above, but this contradicts Attribute (2) above which implies that a value $< c$ is collected in \term{snapshotAtUpdate}.
Thus, \term{currSnap} is \term{snapshotAtUpdate}.
When \term{updateMetadataByT} checks the value of \term{snapshotAtUpdate}'s \codestyle{collecting} field in \Cref{code:SizeCalculator updateMetadata check collecting}, if it is \codestyle{false} then we are done. Otherwise, \term{updateMetadataByT} calls the \codestyle{forward} method to ensure that the value $c$ is forwarded to the snapshot counter.
Like in the previous scenario, we analyze here a case in which the \size{} operation, which sets \term{snapshotAtUpdate}'s \codestyle{size} field, reads a value $< c$ from the snapshot counter, hence, this \size{} must have read the snapshot counter before the forwarding completes, and this read during the size computation occurs only after the \term{snapshotAtUpdate}'s \codestyle{collecting} field is set to \codestyle{false}. This concludes the proof.

\end{proof}

\begin{lemma}\label{lemma:counter inc after first 2 lines of updateMetadata}
Consider a call to \codestyle{updateMetadata} on behalf of $op$, being the $c$-th successful \ins{} or \del{} operation by a thread $T$. 
After this call executes \Crefrange{code:SizeCalculator updateMetadata update array start}{code:SizeCalculator updateMetadata update array end}, the relevant metadata counter's value is $\geq c$.
\end{lemma}

\subsection{Wait-Freedom and Asymptotic Time Complexity}\label{section:time complexity}
The \size{} operation in our methodology is wait-free as it completes within a constant number of steps, regardless of other threads' progress. Its time complexity is linear in the number of threads due to its two passes on arrays with per-thread counters: during the collection process (in the \codestyle{\_collect} method) and the size computation (in the \codestyle{computeSize} method).

Our transformation preserves the time complexity and progress guarantees of the data-structure operations, as each call to the \codestyle{udpateMetadata} method adds a constant number of steps. This follows from the following claim.

\begin{claim}\label{claim: forward 2 iterations}
Each call to the \codestyle{forward} method of \codestyle{CountersSnapshot} (by \codestyle{updateMetadata}) executes at most two iterations of its while loop. 
\end{claim}

Before forwarding, \codestyle{updateMetadata} performs several checks in a certain order. Without this careful procedure, delayed threads that run \codestyle{updateMetadata} to help old operations could forward stale values to the snapshot array, causing an \codestyle{updateMetadata} on behalf of a recent operation to repeatedly fail forwarding.
Next, we show how this procedure prevents forwarding stale values.

\begin{proof}
Consider a call to \codestyle{updateMetadata} that calls \codestyle{forward} and operates on behalf of $op$, being the $c$-th successful \ins{} operation by a thread $T$ (a similar proof applies for \del{}). Denote by \term{currSnap} the value this call obtains in \codestyle{currentCountersSnapshot} (in \Cref{code:SizeCalculator updateMetadata obtain CountersSnapshot}) at time denoted $t_{obt}$. As the snapshot counters are monotonically increasing, it is enough to prove that from $t_{obt}$ and on, only values $\geq c-1$ may be written to the snapshot counter of \term{currSnap} that is associated with $op$.

Note that after obtaining \term{currSnap} at time $t_{obt}$ and before calling \codestyle{forward}, \codestyle{updateMetadata} observes that \term{currSnap} is in a collecting mode (in \Cref{code:SizeCalculator updateMetadata check collecting}), thus, it has been in this mode at $t_{obt}$.
Now, let $t_{c-1}$ be the time in which the metadata counter associated with $op$ is set to $c-1$. If \term{currSnap} has been announced in \codestyle{sizeCalculator.countersSnapshot} before time $t_{c-1}$, then it keeps being announced and in collecting mode at least until time $t_{obt}$. Thus, the value $c-1$ is forwarded to the snapshot counter associated with $op$ before time $t_{obt}$, as otherwise the thread $T$ would not have proceeded from its ($c-1$)-st successful \ins{} to its $c$-th successful \ins{}, and we are done since the snapshot counter is monotonically increasing.

Otherwise, \term{currSnap} is announced in \codestyle{sizeCalculator.countersSnapshot} after time $t_{c-1}$. 
Two methods write to the snapshot array: \codestyle{add} and \codestyle{forward}.
\codestyle{add} is called (in \Cref{code:SizeCalculator collect add}) with a counter value that is obtained from the metadata array after \term{currSnap} is announced, hence, after time $t_{c-1}$, so it writes a value $\geq c-1$.
As for \codestyle{forward}, a thread that calls it (in \Cref{code:SizeCalculator updateMetadata forward}) to forward a value to the snapshot counter associated with $op$, performs the following steps in this order:
(1) obtain \term{currSnap} in \Cref{code:SizeCalculator updateMetadata obtain CountersSnapshot}---which must happen after \term{currSnap} is announced and hence after time $t_{c-1}$;
(2) obtain the value of the metadata counter associated with $op$ in \Cref{code:SizeCalculator updateMetadata obtain and check counter}; and then 
(3) forward this value to \term{currSnap}'s snapshot array.
Since the value is obtained after time $t_{c-1}$, it must be $\geq c-1$.
\end{proof}

\section{Evaluation}\label{section:performance}

In this section, we present the evaluation of our methodology on several data structures. The code for the data structures and the measurements is available at \cite{artifactConcurrentSize}. We first measure the overhead that the addition of the size mechanism introduces to operations of transformed data structures (in \Cref{section:overhead-breakdown-by-op-type} we break down the overhead by operation type). Then, we demonstrate the benefits of computing a linearizable size in our methodology. 
We show that it yields a performance better in orders of magnitude than the alternative of taking a linearizable snapshot of the data structure and counting its elements. We further demonstrate the scalability of our methodology and its performance insensitivity to the data-structure size. 

\paragraph{Evaluated data structures}
We start with three baseline data structures that do not support a linearizable \size{}: a skip list, a hash table and a binary search tree, denoted \skl{}, \htb{} and \bst{} respectively.
The implementation of \skl{} is taken from the \texttt{ConcurrentSkipListMap} class of the \texttt{java.util.concurrent} package of Java SE 18. We eliminated methods irrelevant to our measurements and kept the main \ins{}, \del{} and \contains{} functionality.
As for a hash table, we implemented \htb{} as a table of linked lists whose implementation is based on the linked list in the base level of \skl{}. We use a table of a static size (chosen in a way similar to Java's \texttt{ConcurrentHashMap} to be a power of $2$ between $1\times$ and $2\times$ the number of elements;
we detail below how we keep the number of elements stable during the measurements). 
We do not use Java's lock-based \texttt{ConcurrentHashMap} as our hash-table baseline because it deletes items by unlinking without marking (as it does not use a delicate synchronization mechanism but rather coarse-grained locking), thus, our transformation is not applicable to it as is.
For \bst{} we use Trevor Brown's implementation \cite{brownJava} of the lock-free binary search tree of \cite{ellen2010non} that places elements in leaf nodes.

We applied our methodology to each of the baseline algorithms, to produce the transformed data structures \sskl{}, \shtb{} and \sbst{} that support a linearizable \size{}. In the case of the tree, \bst{} linearizes the \del{} operation at the unlinking and not in the prior marking of the deleted node's parent. Hence, we formed a variant of \bst{} that linearizes \del{} at the marking step, and then applied our methodology to this variant.

We compare performance of the \size{} operation in the data structures produced using our methodology with a snapshot-based \size{} operation in two data structures supporting snapshots. 
The first one is \snapskl{}---\cite{petrank2013lock}'s implementation of a skip list with a snapshot mechanism, which we obtained from the paper's authors. Like our skip list implementations (\skl{} and \sskl{}), it builds on Java's \texttt{ConcurrentSkipListMap}. It uses code from an older version of Java, but the performance degradation incurred by the older version is negligible and irrelevant to our measurements, due to the immense performance difference between obtaining the size using their snapshot and using our methodology.
The \size{} operation is implemented in \snapskl{} by taking a snapshot, which produces a snapshot copy of the base level of the skip list, and then iterating it and counting its elements.

The second competitor we compare to is \vcasbst{} \cite{wei2021constant}---a binary search tree with a snapshot mechanism taken from the paper's published implementation \cite{vcaslib}. It is based on the same implementation by Brown that \bst{} and \sbst{} are based on, but uses a modified version of it which batches multiple keys in leaves and stores up to 64 key-value pairs in each tree leaf.
To compute the size, we did not use their implementation as a black box, as their interface supplies a snapshot copy of the tree's elements, but such copying is redundant for retrieving the size. Instead, to compute the size we call their snapshot operation that advances the timestamp, and then traverse the tree and sum the number of elements in leaves with a timestamp no bigger than the snapshot timestamp. By this, we save copying all tree elements, and even save iterating the elements one by one---as we simply read the leaf's number of elements (each batched leaf node keeps the number of contained elements). Even though we use this improved size implementation for \vcasbst{}, and even though \vcasbst{} uses batched leaves which enables it to perform faster than without them, we will show that our size computation method still outperforms it.

\paragraph{Platform}
We conducted our experiments on a machine running Linux (Ubuntu $20.04$) equipped with $4$ AMD Op\-te\-ron(TM) 6376 \SI{2.3}{GHz} processors. Each processor has $16$ cores, resulting in $64$ threads overall.
The machine used \SI{64}{GB} RAM, an L1 data cache of \SI{16}{KB} per core, an L2 cache of \SI{2}{MB} for every two cores, and an L3 cache of \SI{6}{MB} for every $8$ cores.

The implementations were written in Java. We used \codestyle{OpenJDK} $17.0.2$ with the flags \codestyle{-server}, \codestyle{-Xms15G} and \codestyle{-Xmx15G}. 
The latter two flags reduce interference from Java's garbage collection. We used the G1 garbage collector (using ParallelGC yields similar results).

\paragraph{Methodology}
Before each experiment, we fill the data structure with $1$M items, except for the experiments that check dependence on the data-structure size, in which we fill the data structure with a varying number of items---$1$M, $10$M or $100$M. We chose these sizes in order to measure the performance of the data-structure when it does not fit into the L3 cache.

We run two workloads: an update-heavy workload, with $30\%$ \ins{} operations, $20\%$ \del{} operations and $50\%$ \contains{} operations, and a read-heavy workload, with $3\%$ \ins{} operations, $2\%$ \del{} operations and $95\%$ \contains{} operations. These workloads match the read rates suggested by {\em Yahoo! Cloud Serving Benchmark} (YCSB) \cite{cooper2010benchmarking}---update-heavy workloads with $50\%$ reads and read-heavy workloads with $95\%$ reads (YCSB also suggests a $100\%$-read workload, but this is less relevant to our case, since it is less likely to have \size{} calls on a data structure that never changes). 
The left part of \Crefrange{fig:HT overhead}{fig:overhead breakdown} shows results for the read-heavy workload, and the right part shows results for the update-heavy workload.

Similarly to \cite{wei2021constant}, keys for operations during the experiment and for the initial filling are drawn uniformly at random from a range $[1, r]$, where $r$ is chosen to maintain the initial size of the data structure. For example, for $n=1$M initial keys and a workload with $30\%$ inserts and $20\%$ deletes, we use $r = n \cdot (30 + 20)/30 \approx 1.67M$.

In all experiments except for the experiments that check how overhead is split by operation type, we repeatedly choose (by the update-heavy or read-heavy proportion) the type of the next operation. However, in the overhead-split measurements (that appear in \Cref{section:overhead-breakdown-by-op-type}), we repeatedly choose a uniform type for the next 100 operations, because in these measurements we need to obtain the time it took to execute operations of each type, and obtaining the time it took to execute too few consecutive operations of the same type would impair the time measurement accuracy.

In each experiment, we run $w$ workload threads, performing \ins{}, \del{} and \contains{} calls according to the update-heavy or read-heavy workloads, and $s$ size threads, repeatedly calling \size{}, except for executions of the baseline algorithms (\htb{}, \skl{} and \bst{}---evaluated in the overhead and overhead breakdown measurements), for which we run $w$ workload threads only. $w$ and $s$ vary across experiments, and we took $w+s$ to be a power of 2 in most experiments.
In each experiment, the threads perform operations concurrently for $5$ seconds.
Each data point in the graphs represents the average result of $10$ runs, after executing $5$ preliminary runs to warm up the Java virtual machine. 
The coefficient of variation was up to $11\%$ in the experiments we present next, and up to $21\%$ in the experiments presented in \Cref{section:overhead-breakdown-by-op-type}.

\begin{figure*}[t]
  \centering
  \medskip
  \textit{\ \ \ \ \ \ \ \ \ \ \ \ Read heavy}\hfill
  \includegraphics[height=.03\textwidth]{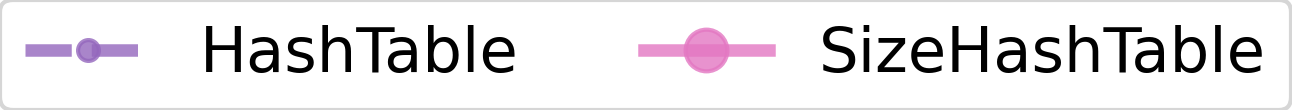}\hfill
  \textit{Update heavy\ \ \ \ }\par
  \medskip
  \text{Without a concurrent size thread}\par
  \hspace*{1.8mm}\includegraphics[width=.473\textwidth]{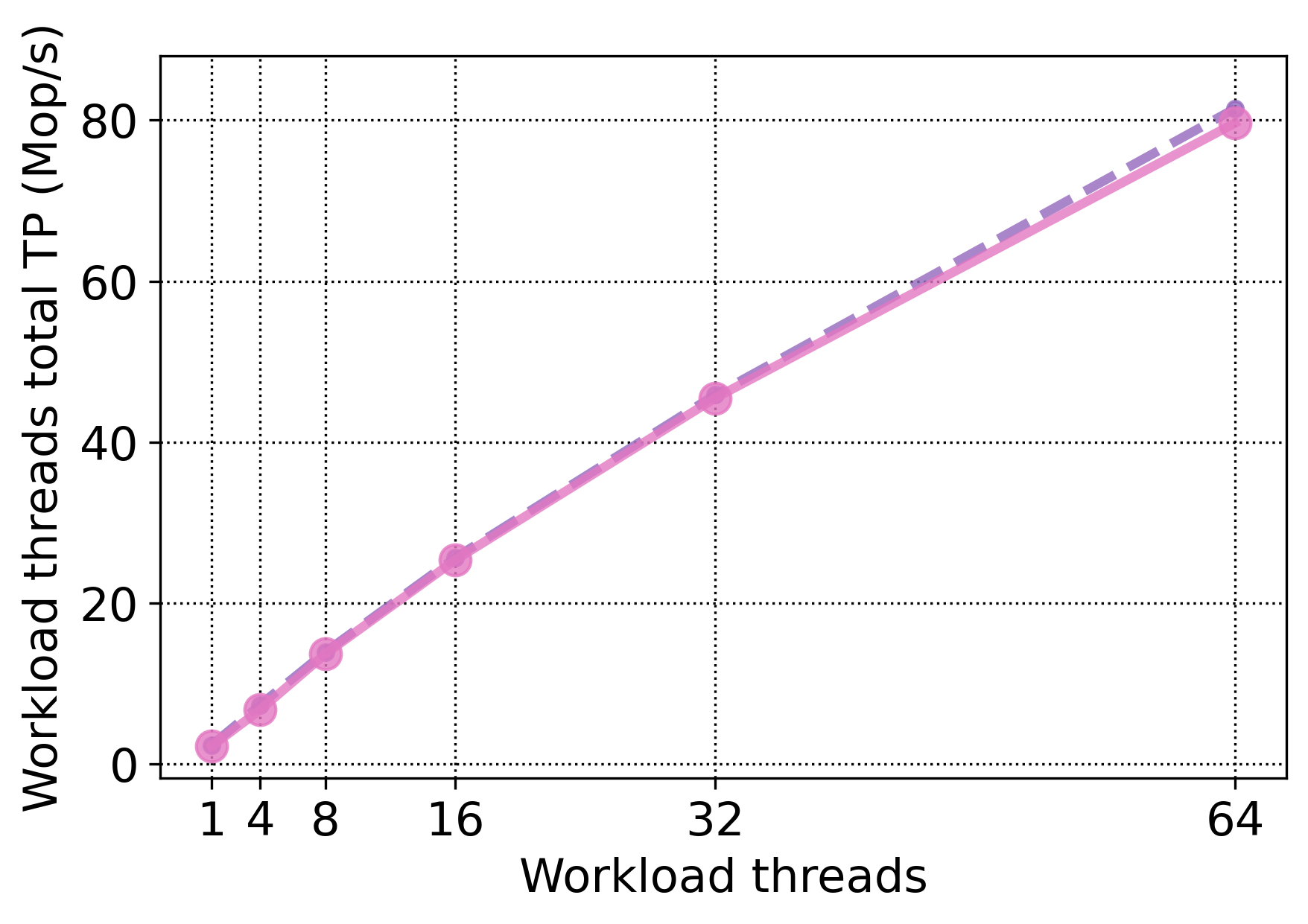}\hspace*{3.2mm}
  \includegraphics[width=.473\textwidth]{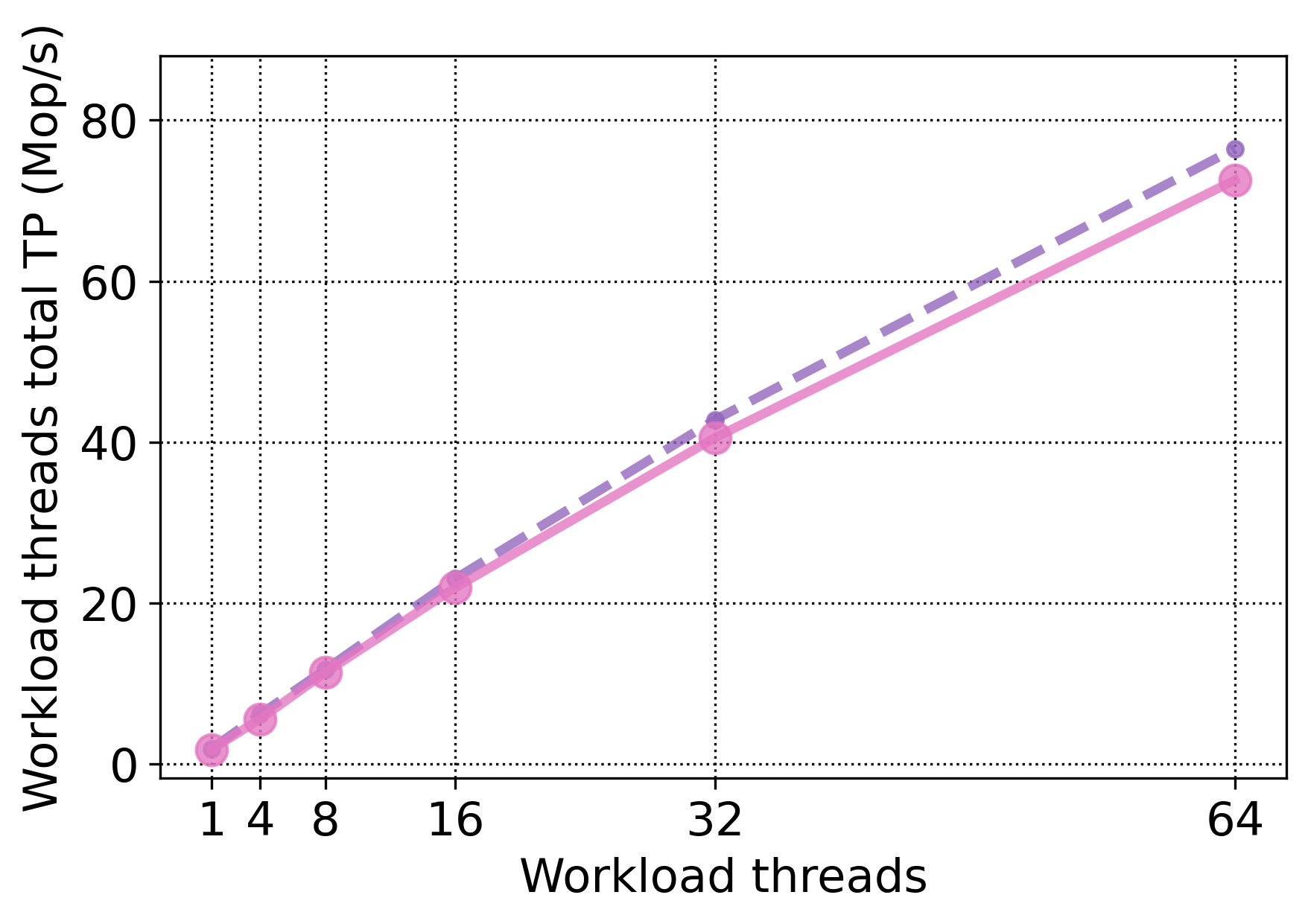}
  \includegraphics[width=.4978\textwidth,trim={0 0 0 .4cm}]{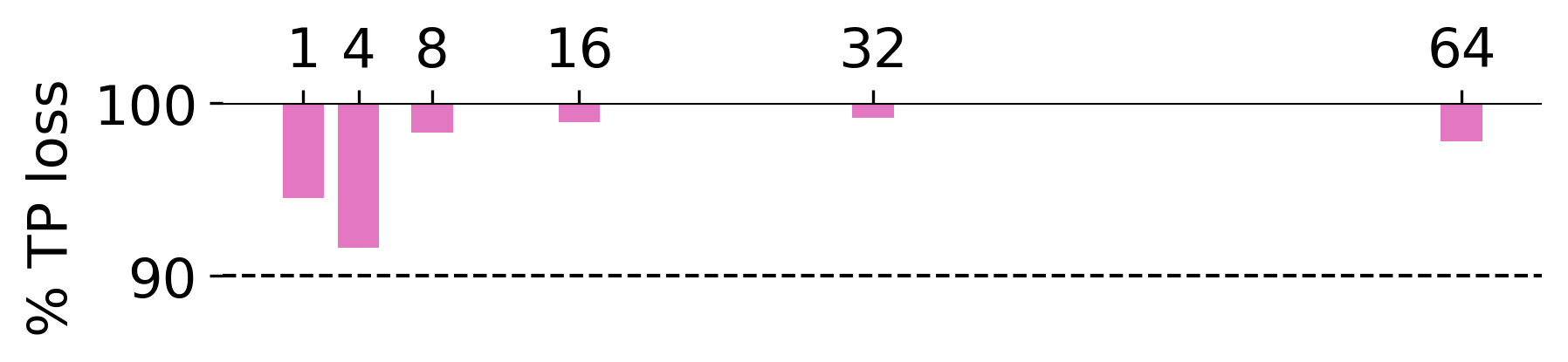}\hspace*{0.001mm}
  \includegraphics[width=.4978\textwidth,trim={0 0 0 .4cm}]{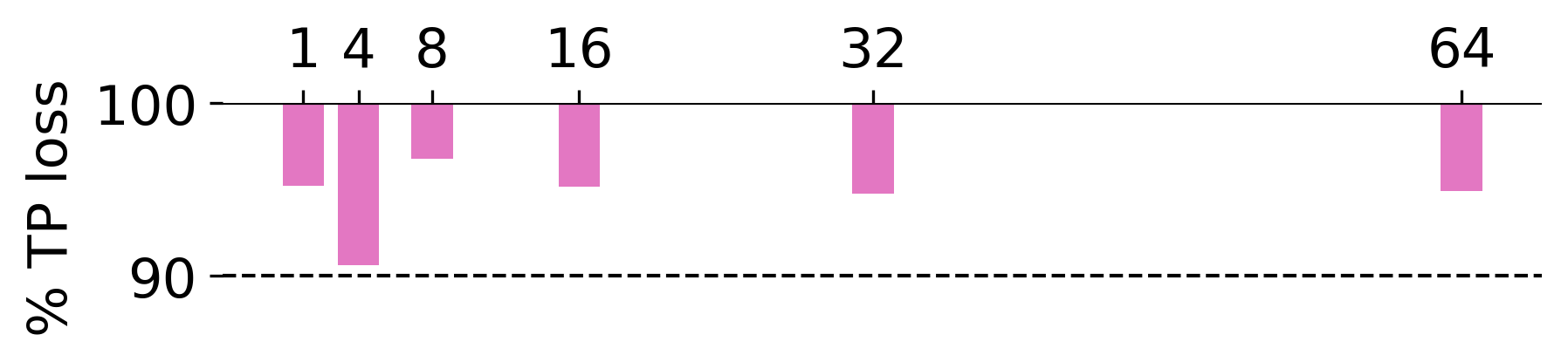}
  \medskip
  \text{With a concurrent size thread}\par
  \hspace*{1.8mm}\includegraphics[width=.473\textwidth,trim={0 0 0 .3cm}]{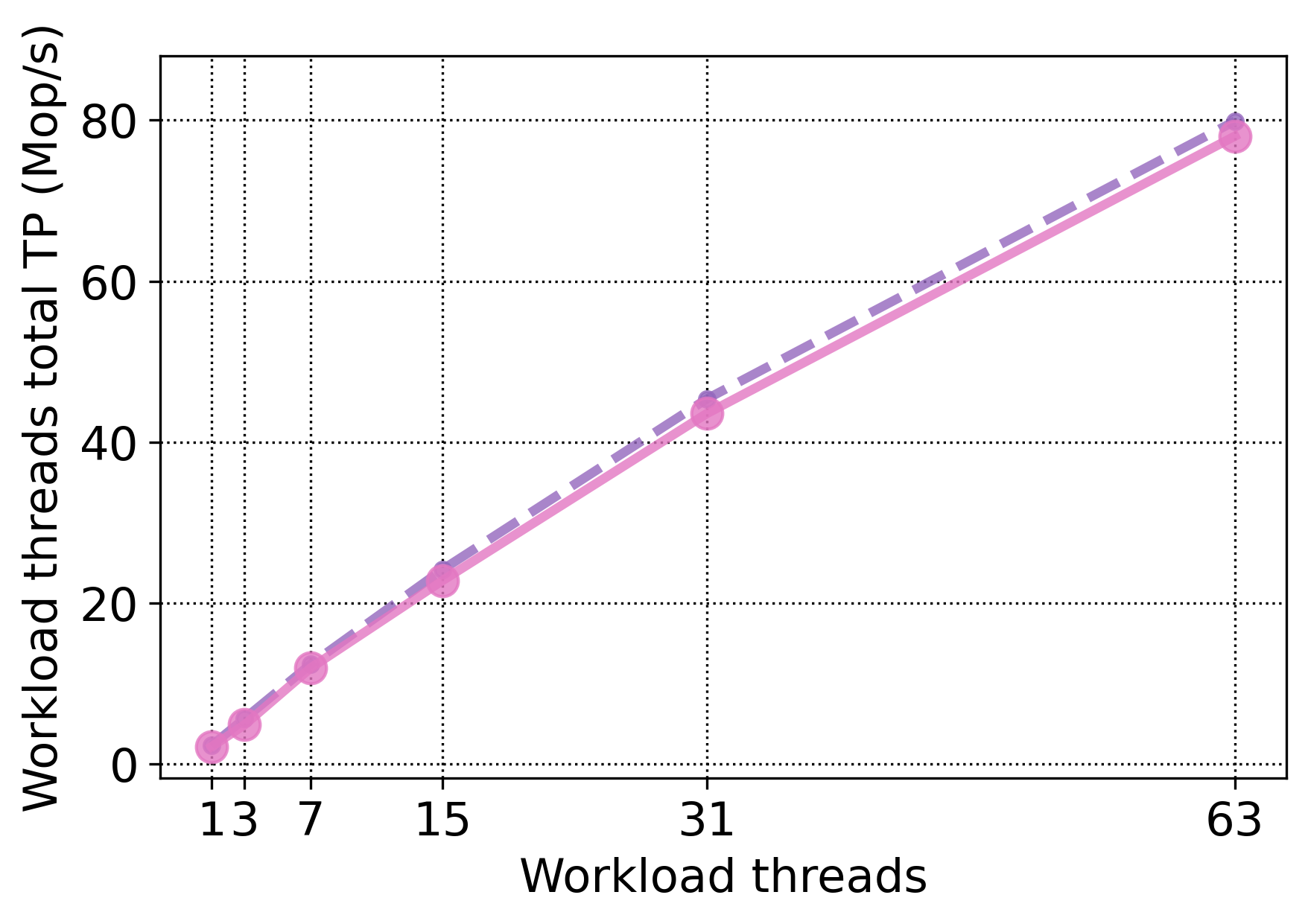}\hspace*{3.2mm}
  \includegraphics[width=.473\textwidth,trim={0 0 0 .3cm}]{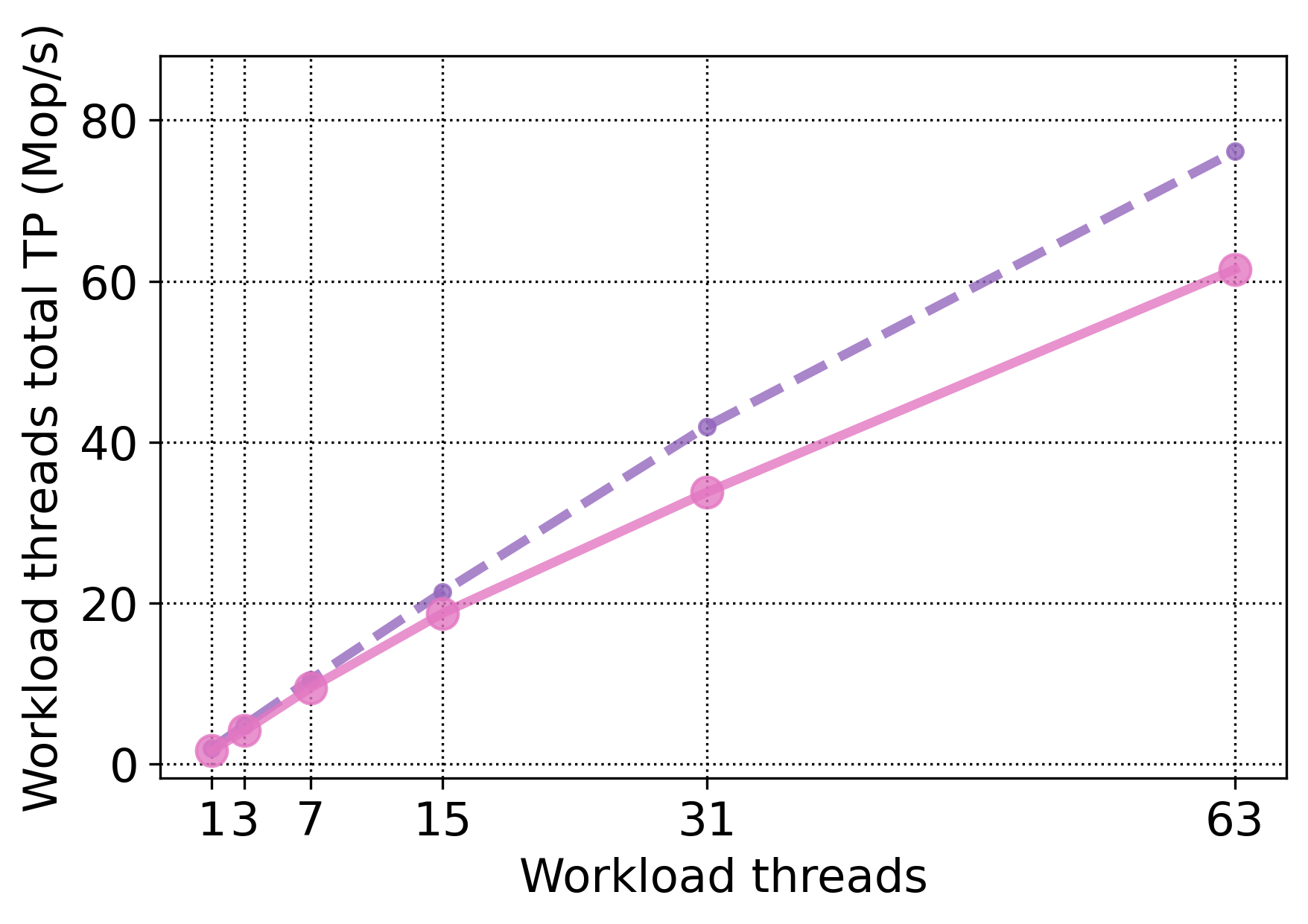}
  \raisebox{.45cm}{\includegraphics[width=.4975\textwidth,trim={0 0 0 .4cm}]{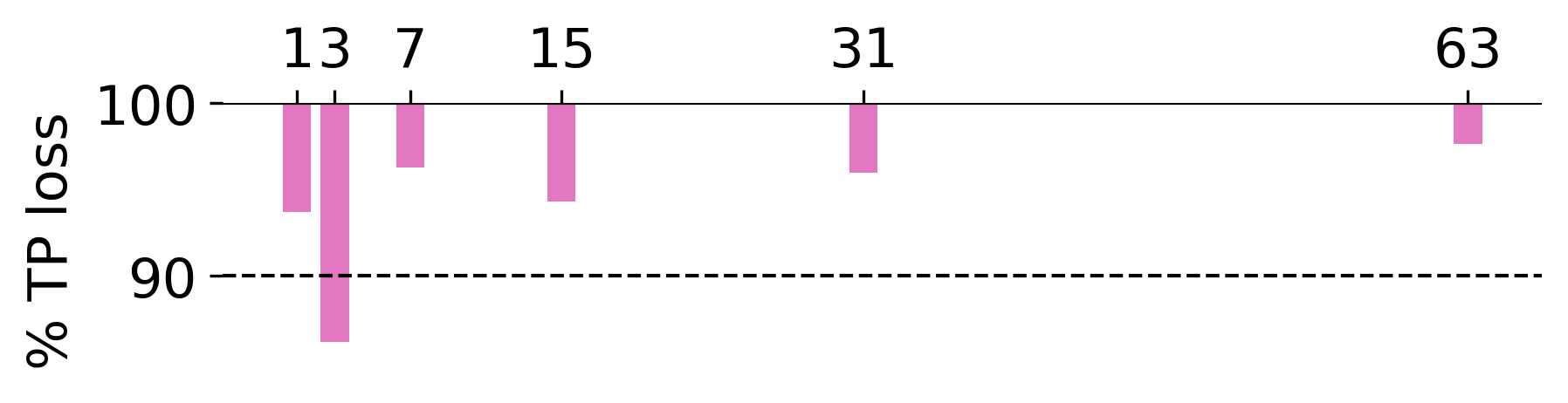}}\hspace*{0.01mm}
  \includegraphics[width=.4975\textwidth,trim={0 0 0 .4cm}]{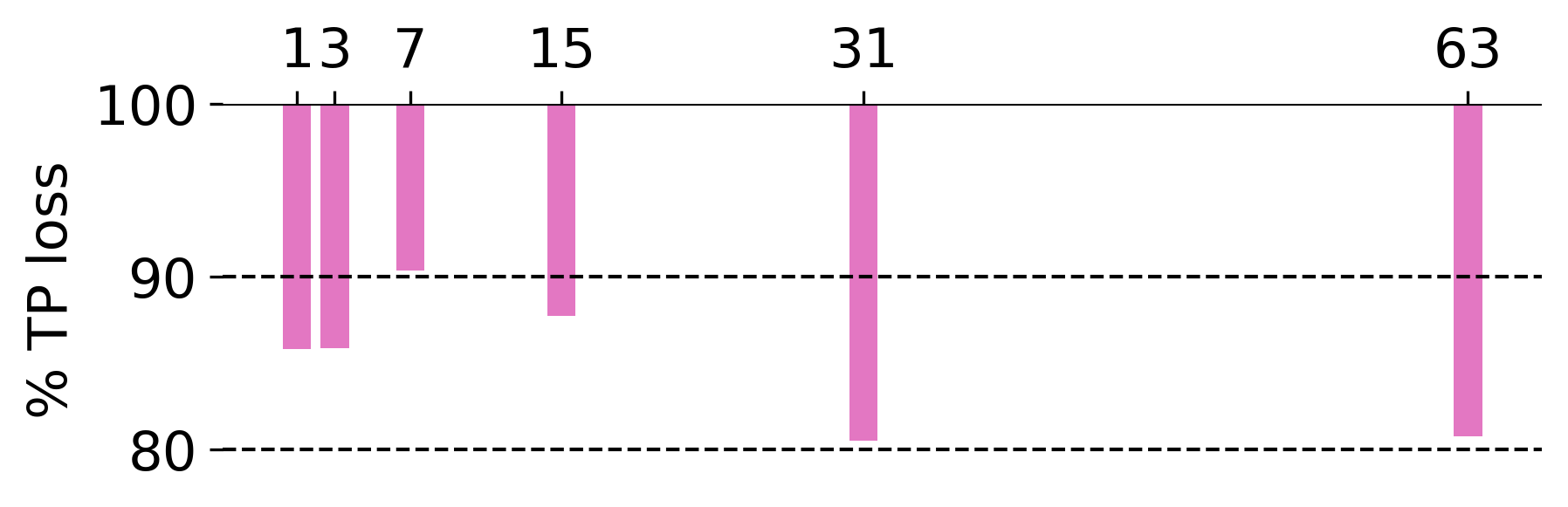}
  \caption{Overhead on hash table operations}
  \label{fig:HT overhead}
\end{figure*}

\begin{figure*}[t]
  \centering
  \medskip
  \textit{\ \ \ \ \ \ \ \ \ \ \ \ Read heavy}\hfill
  \includegraphics[height=.03\textwidth]{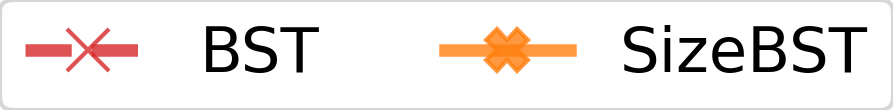}\hfill
  \textit{Update heavy\ \ \ \ }\par
  \medskip
  \text{Without a concurrent size thread}\par
  \hspace*{2mm}\includegraphics[width=.473\textwidth]{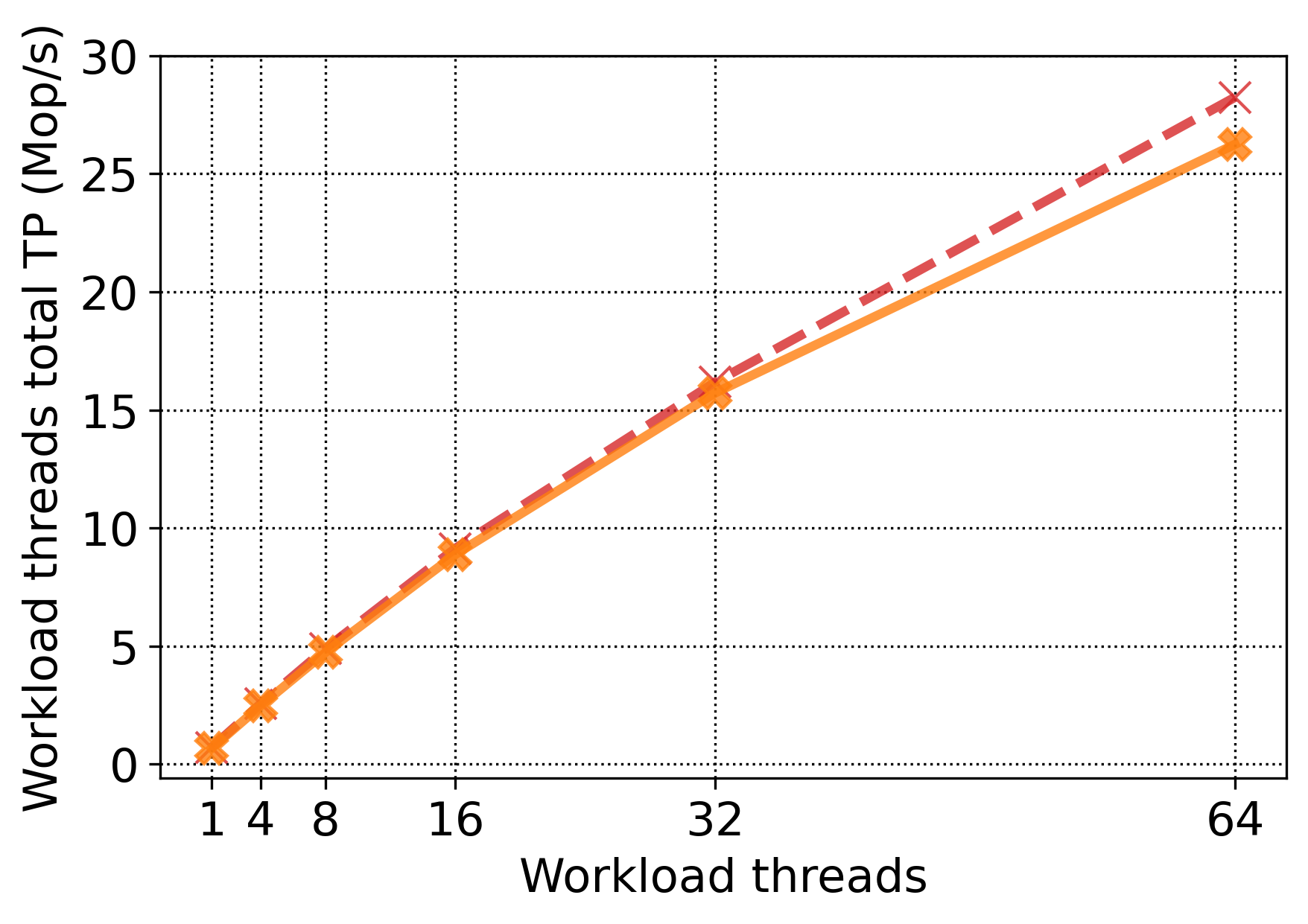}\hspace*{3mm}
  \includegraphics[width=.473\textwidth]{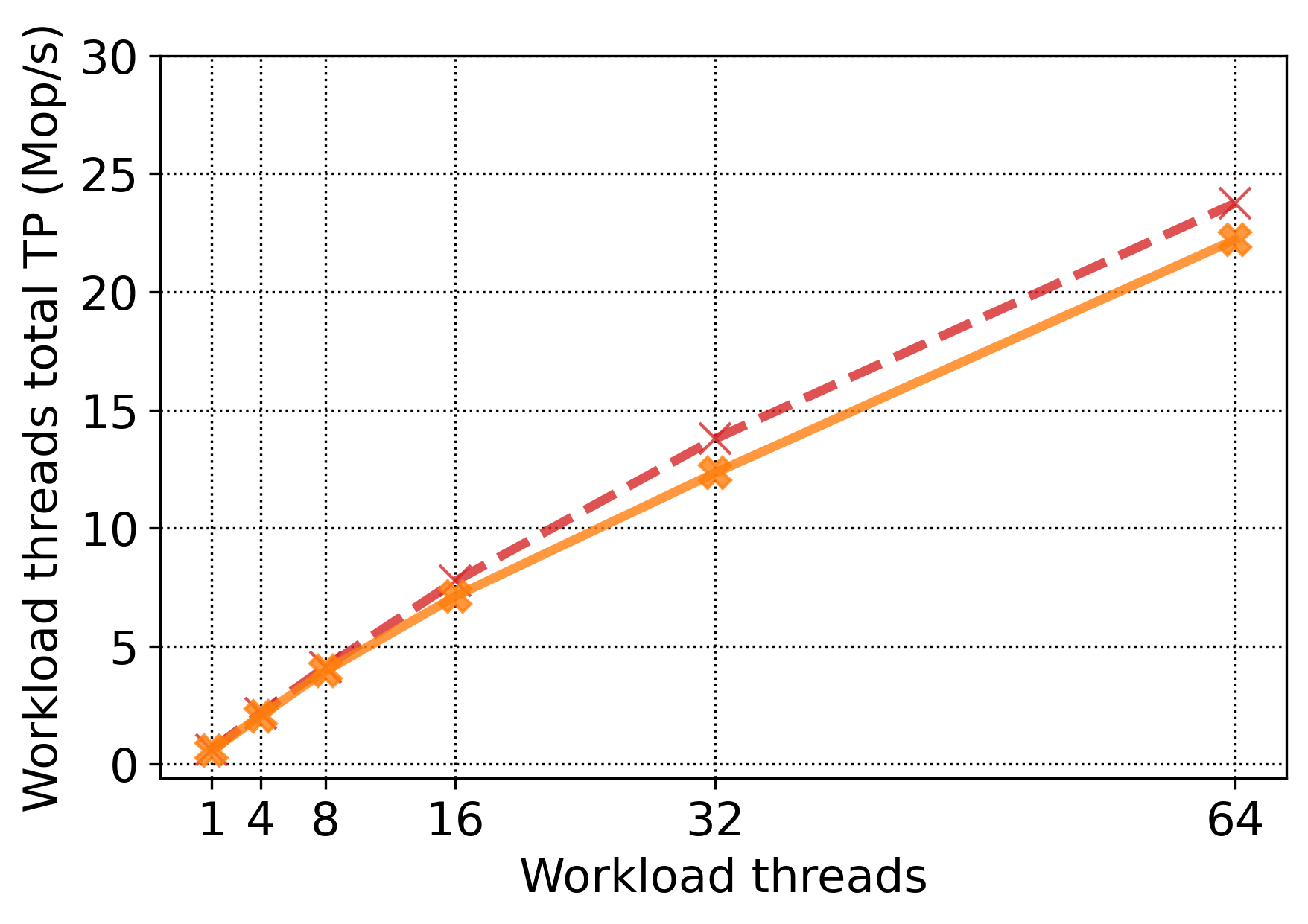}
  \includegraphics[width=.499\textwidth,trim={0 0 0 .4cm}]{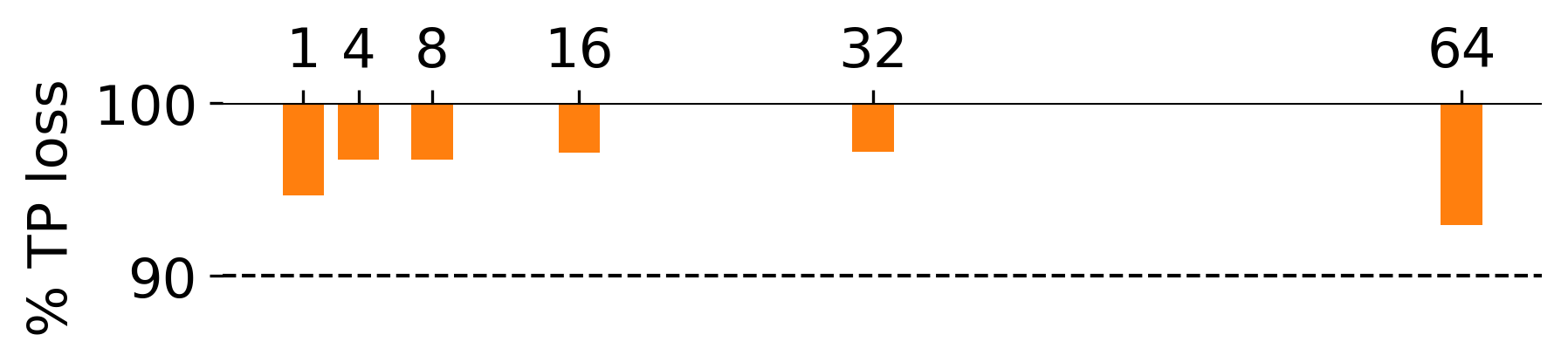}\hspace*{0.001mm}
  \includegraphics[width=.4985\textwidth,trim={0 0 0 .4cm}]{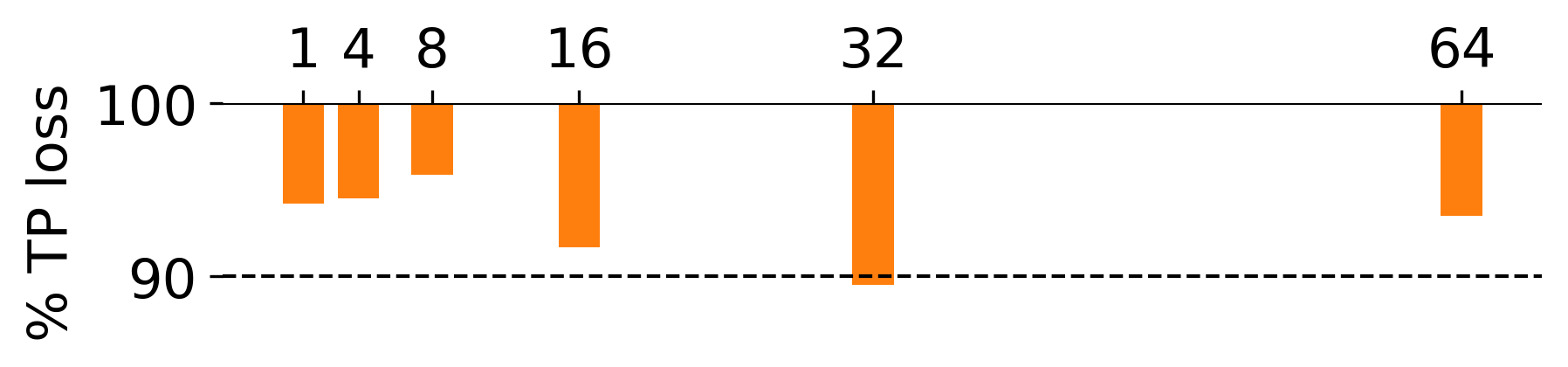}
  \medskip
  \text{With a concurrent size thread}\par
  \hspace*{2mm}\includegraphics[width=.473\textwidth,trim={0 0 0 .3cm}]{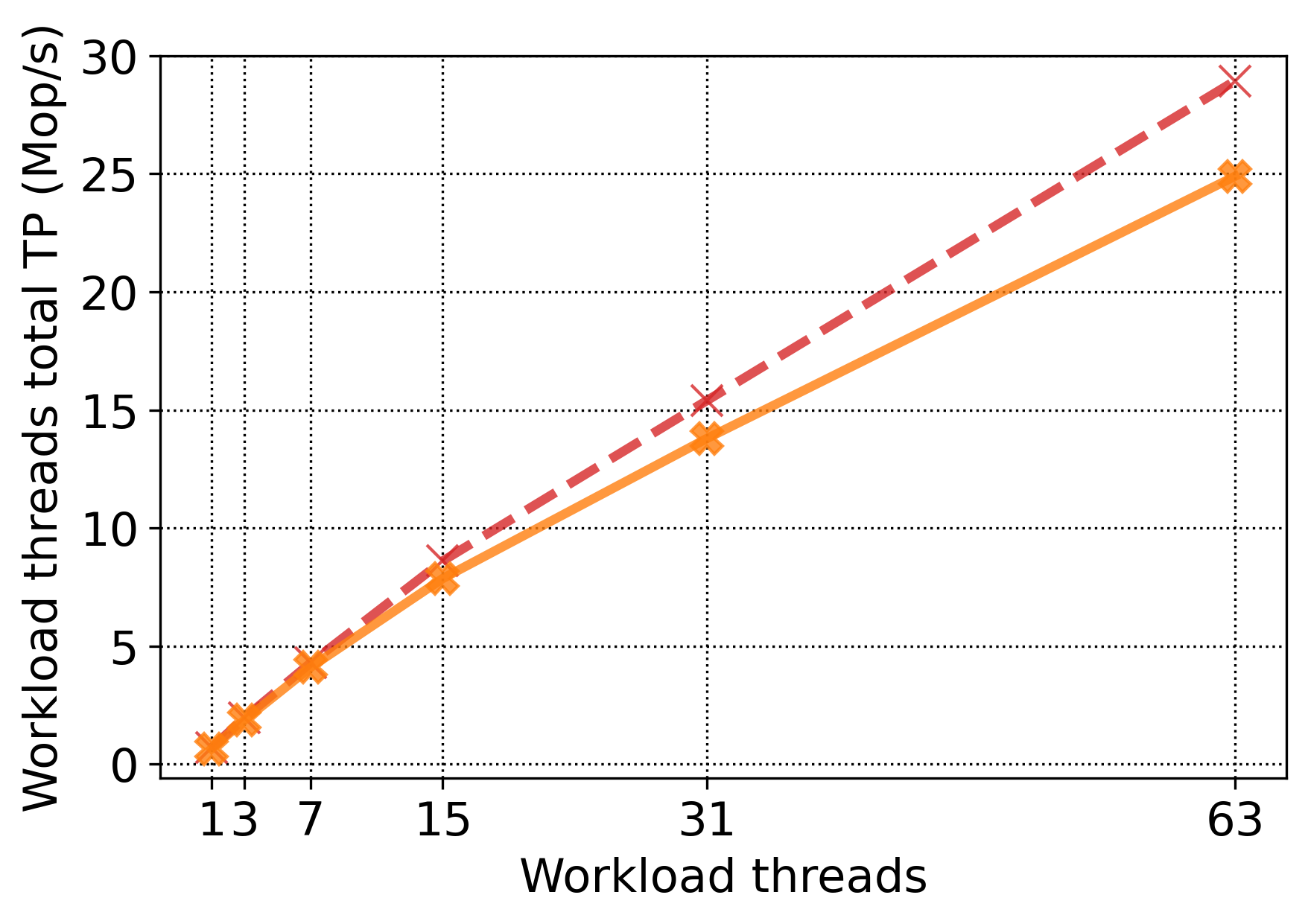}\hspace*{3mm}
  \includegraphics[width=.473\textwidth,trim={0 0 0 .3cm}]{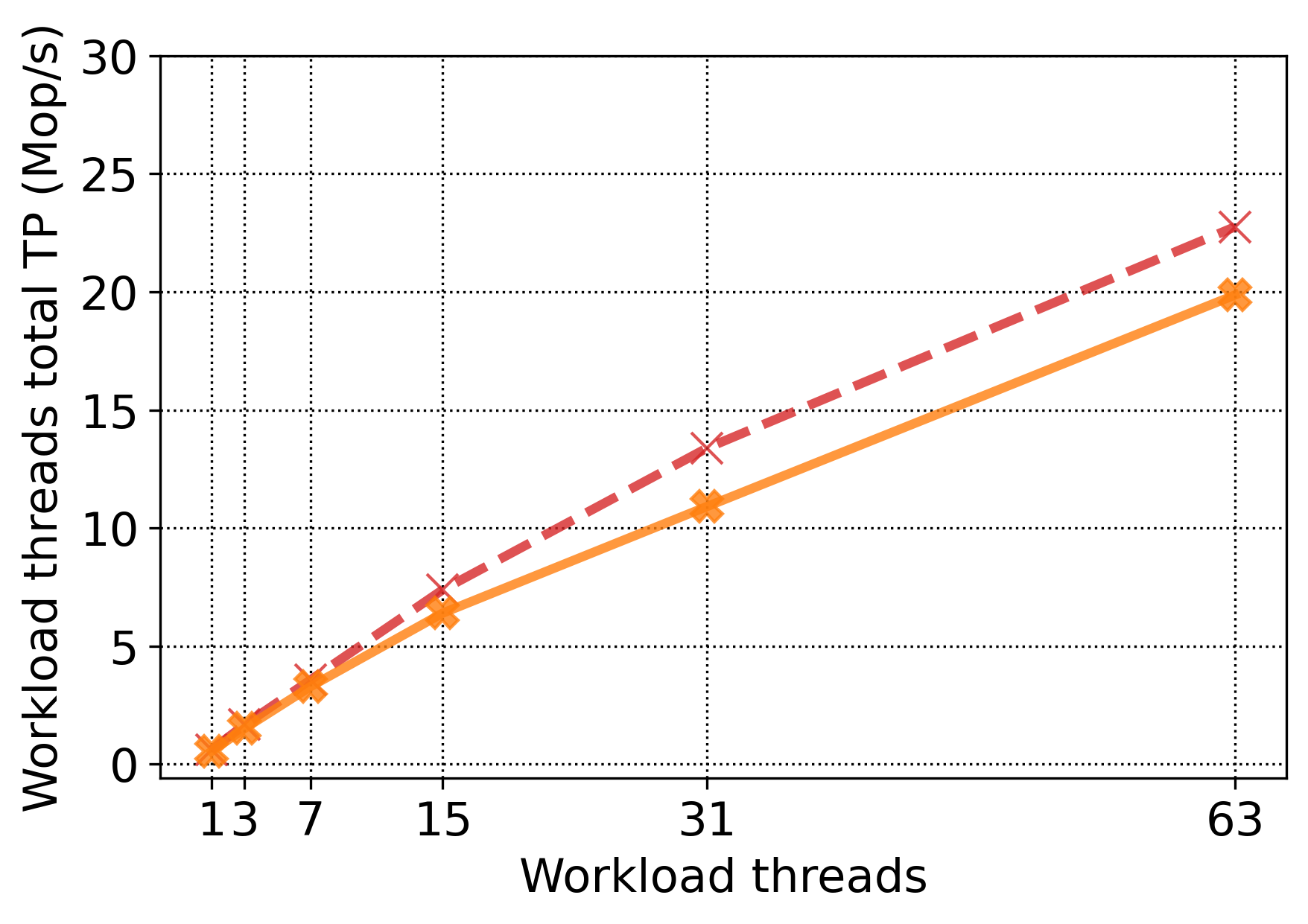}
  \raisebox{.45cm}{\includegraphics[width=.495\textwidth,trim={0 0 0 .4cm}]{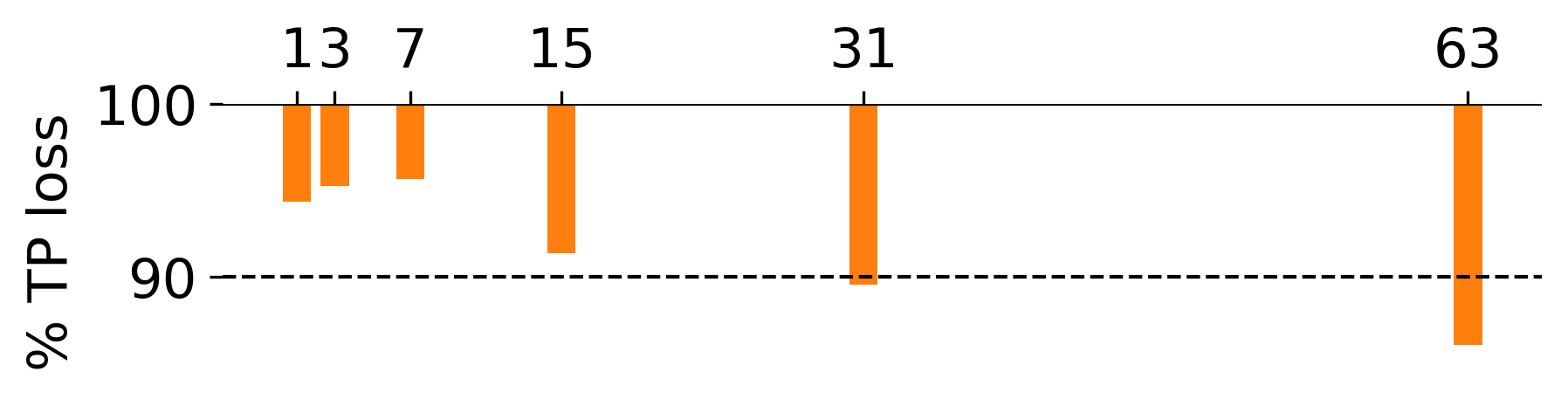}}\hspace*{0.001mm}
  \includegraphics[width=.495\textwidth,trim={0 0 0.5 .4cm}]{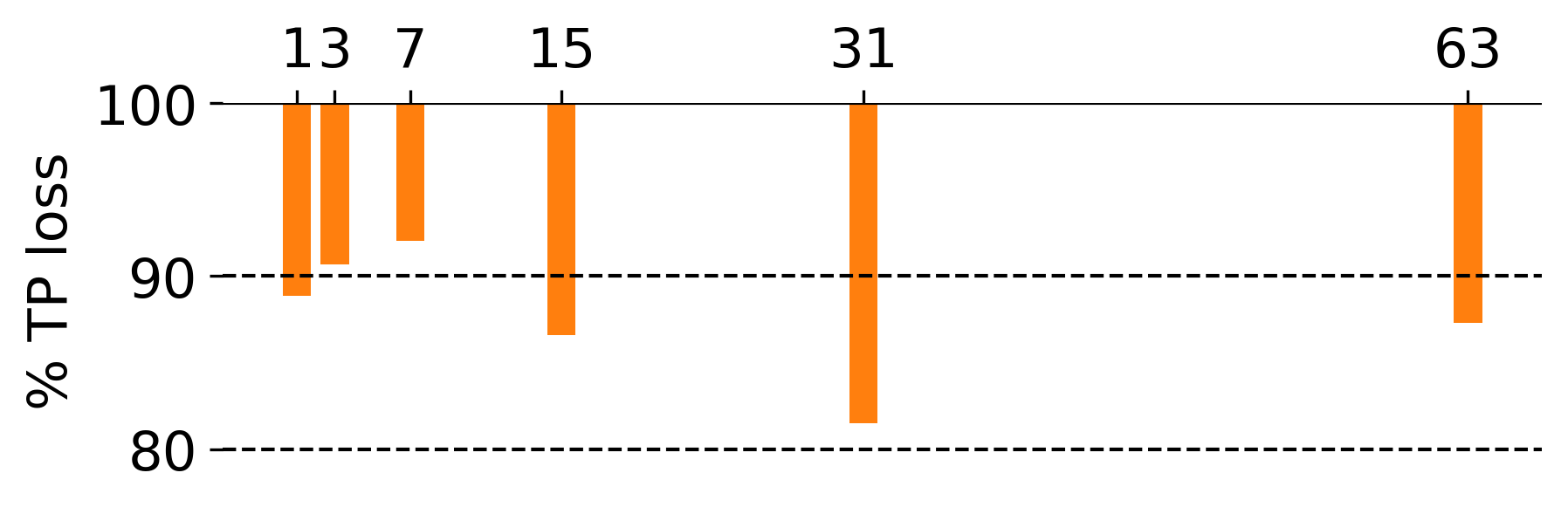}
  \caption{Overhead on BST operations}
  \label{fig:BST overhead}
\end{figure*}

\begin{figure*}[t]
  \centering
  \medskip
  \textit{\ \ \ \ \ \ \ \ \ \ \ \ Read heavy}\hfill
  \includegraphics[height=.03\textwidth]{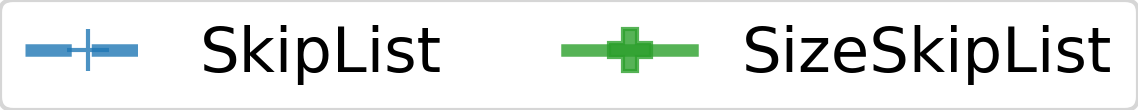}\hfill
  \textit{Update heavy\ \ \ \ }\par
  \medskip
  \text{Without a concurrent size thread}\par
  \includegraphics[width=.49\textwidth]{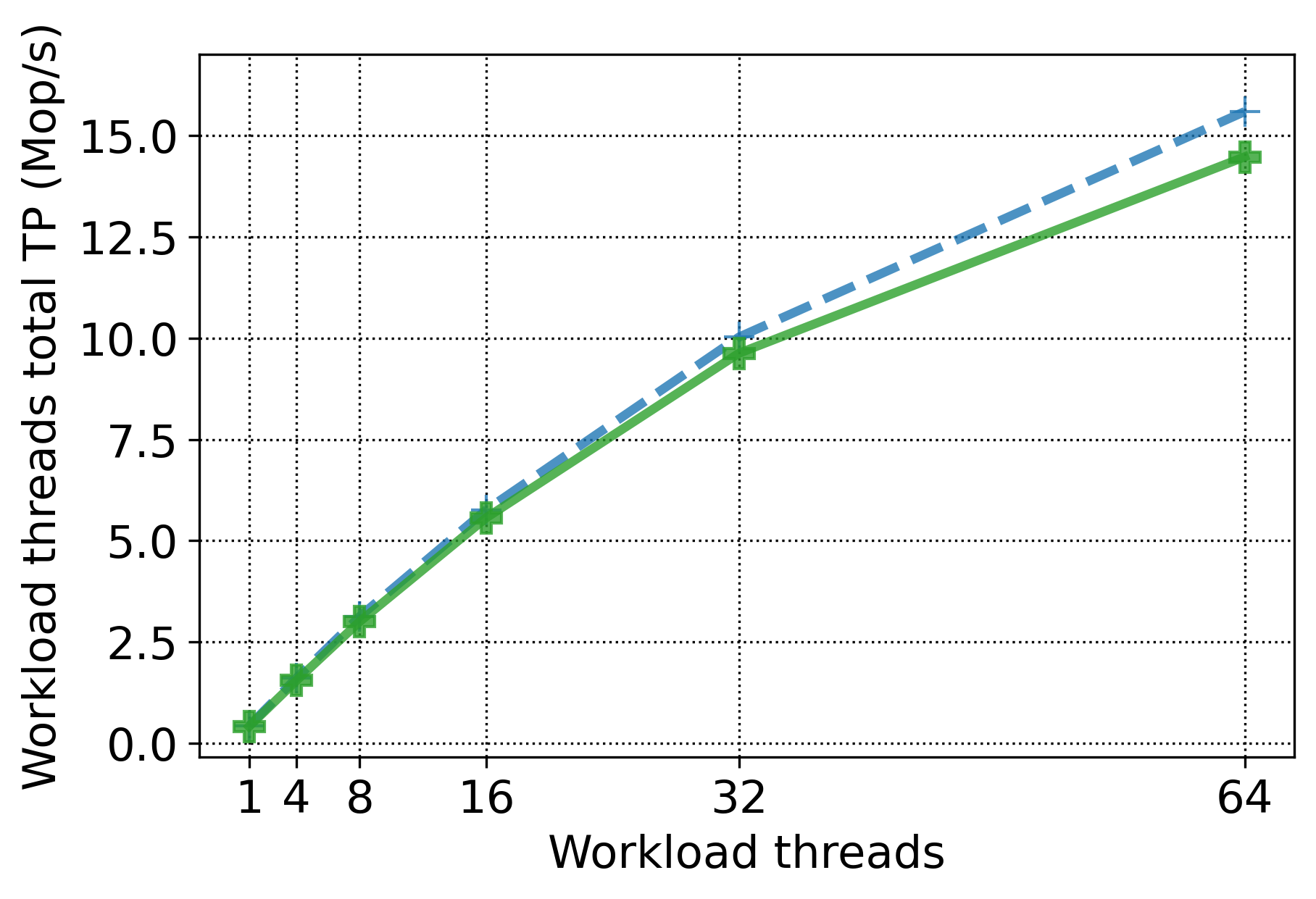}\hspace*{1.5mm}
  \includegraphics[width=.49\textwidth]{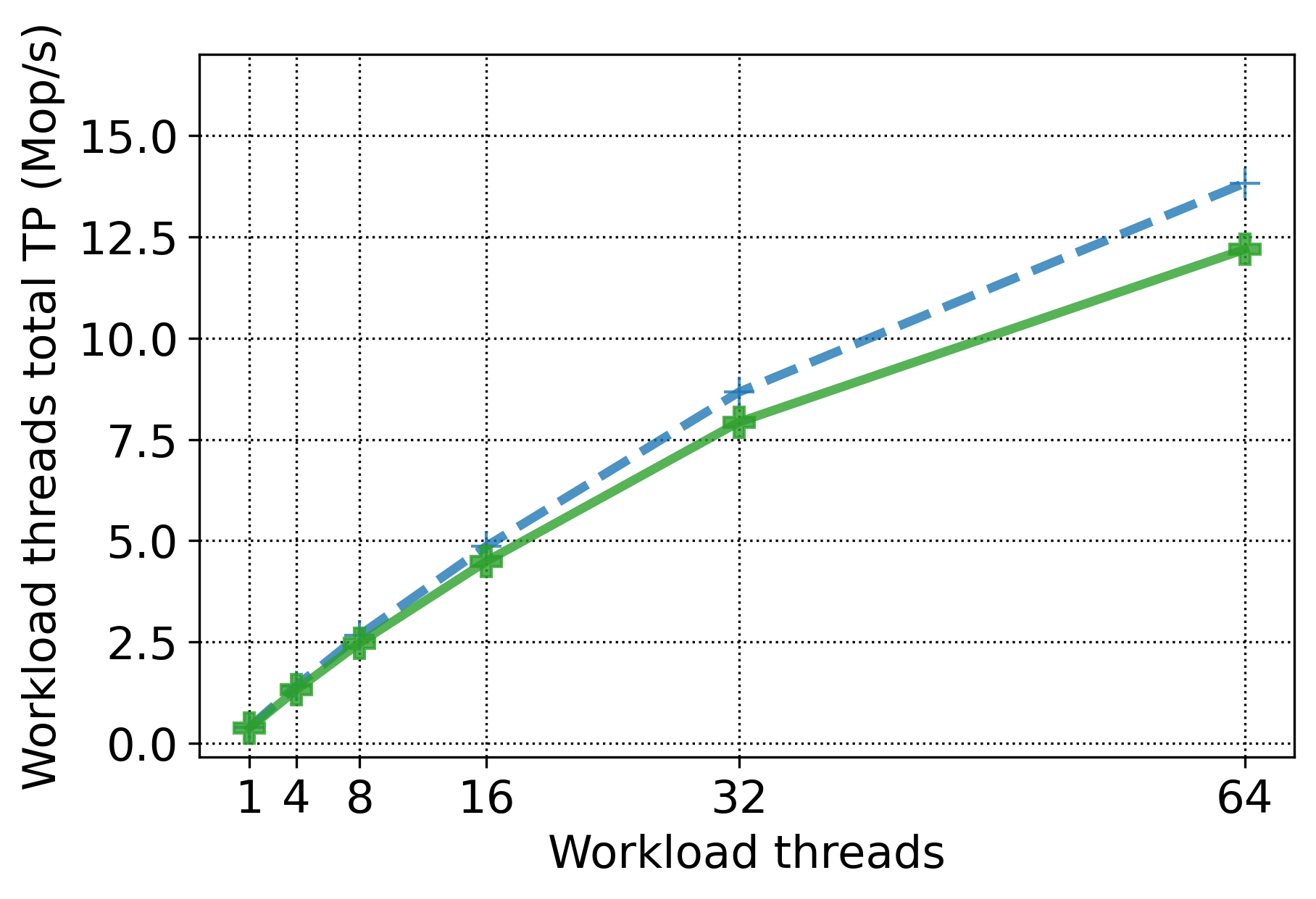}
  \raisebox{.08cm}{\includegraphics[width=.5\textwidth,trim={0.5 0 0 .4cm}]{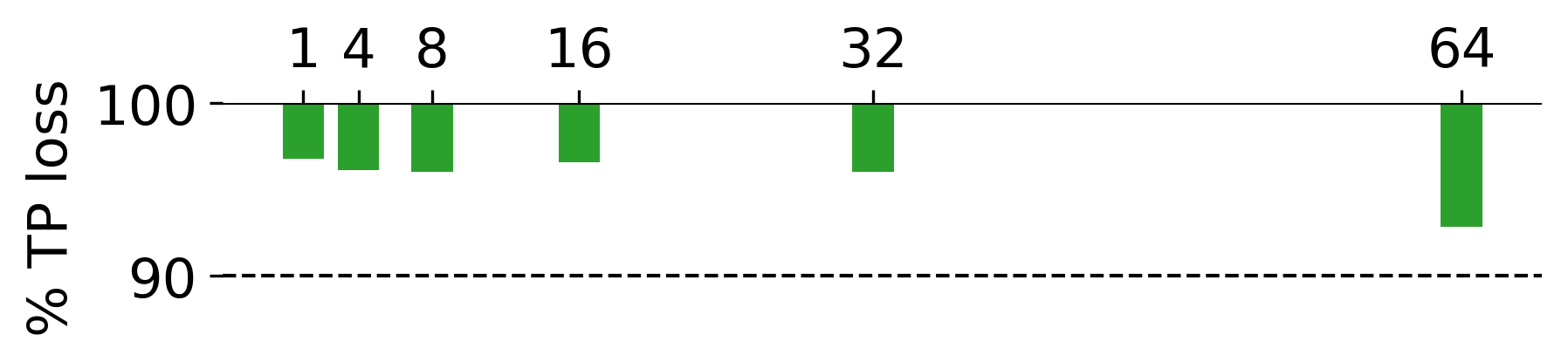}}\hspace*{0.5mm}
  \includegraphics[width=.5\textwidth,trim={0.5 0 0 .4cm}]{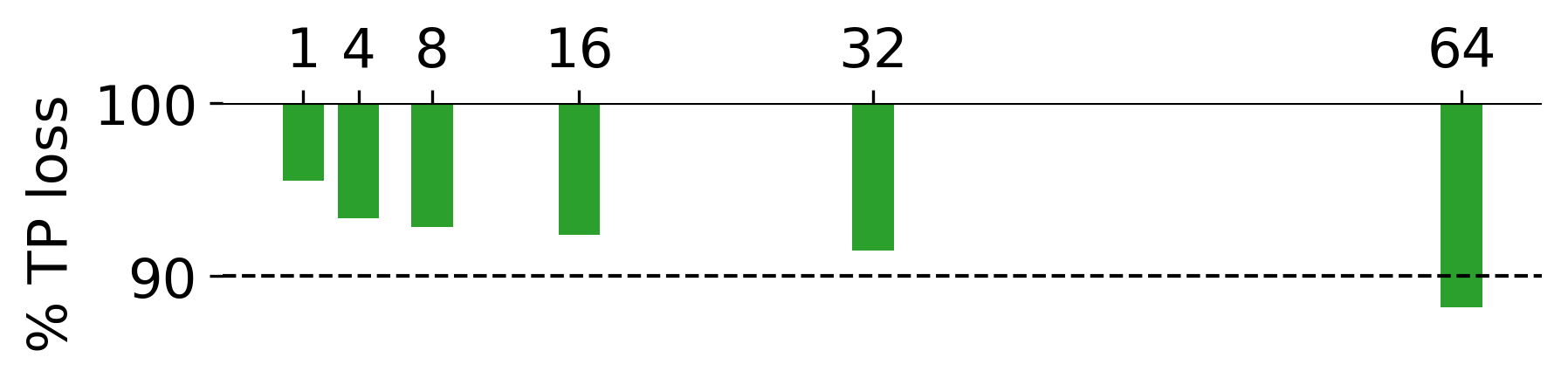}
  \medskip
  \text{With a concurrent size thread}\par
  \includegraphics[width=.49\textwidth,trim={0 0 0 .3cm}]{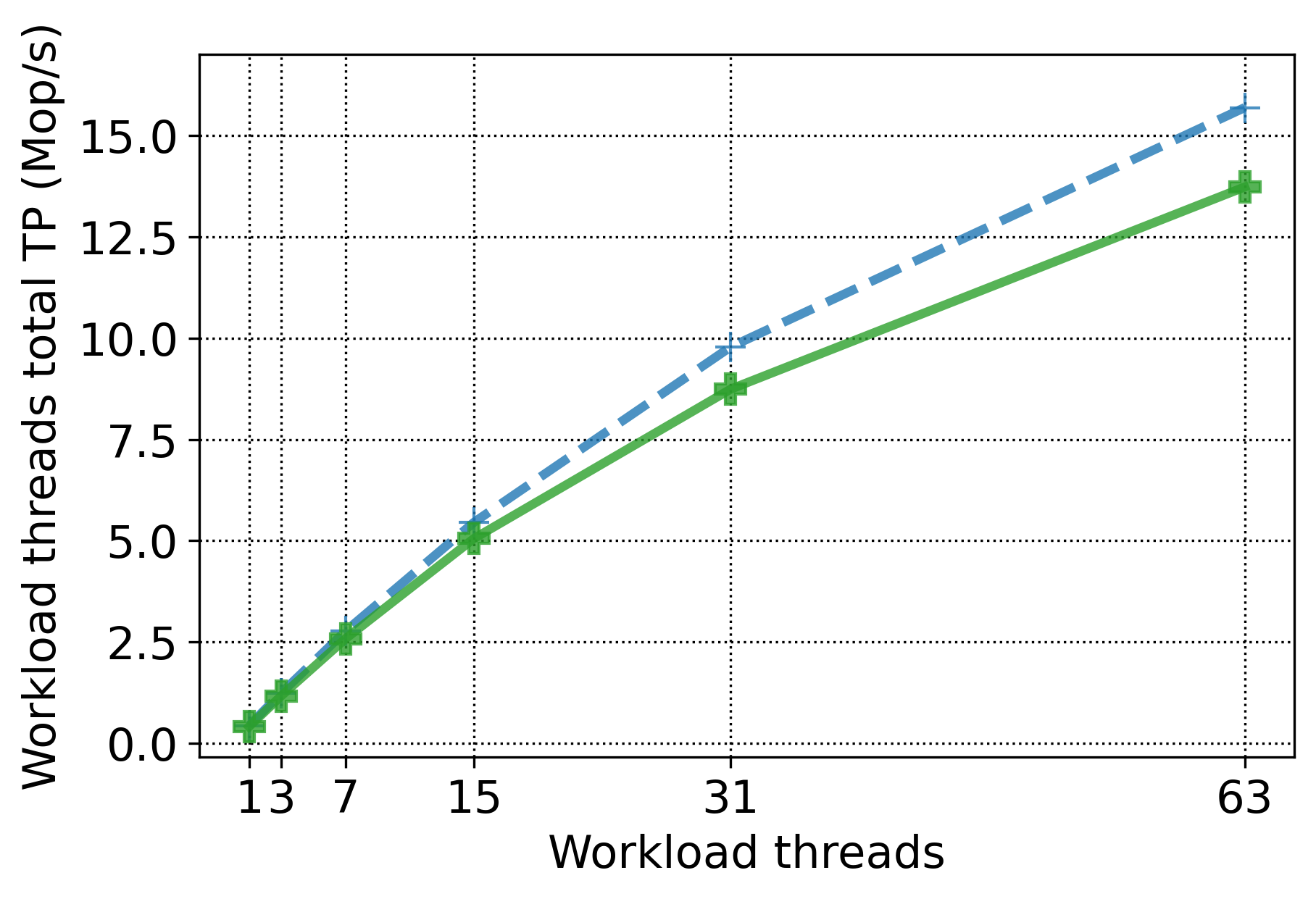}\hspace*{1mm}
  \includegraphics[width=.49\textwidth,trim={0 0 0 .3cm}]{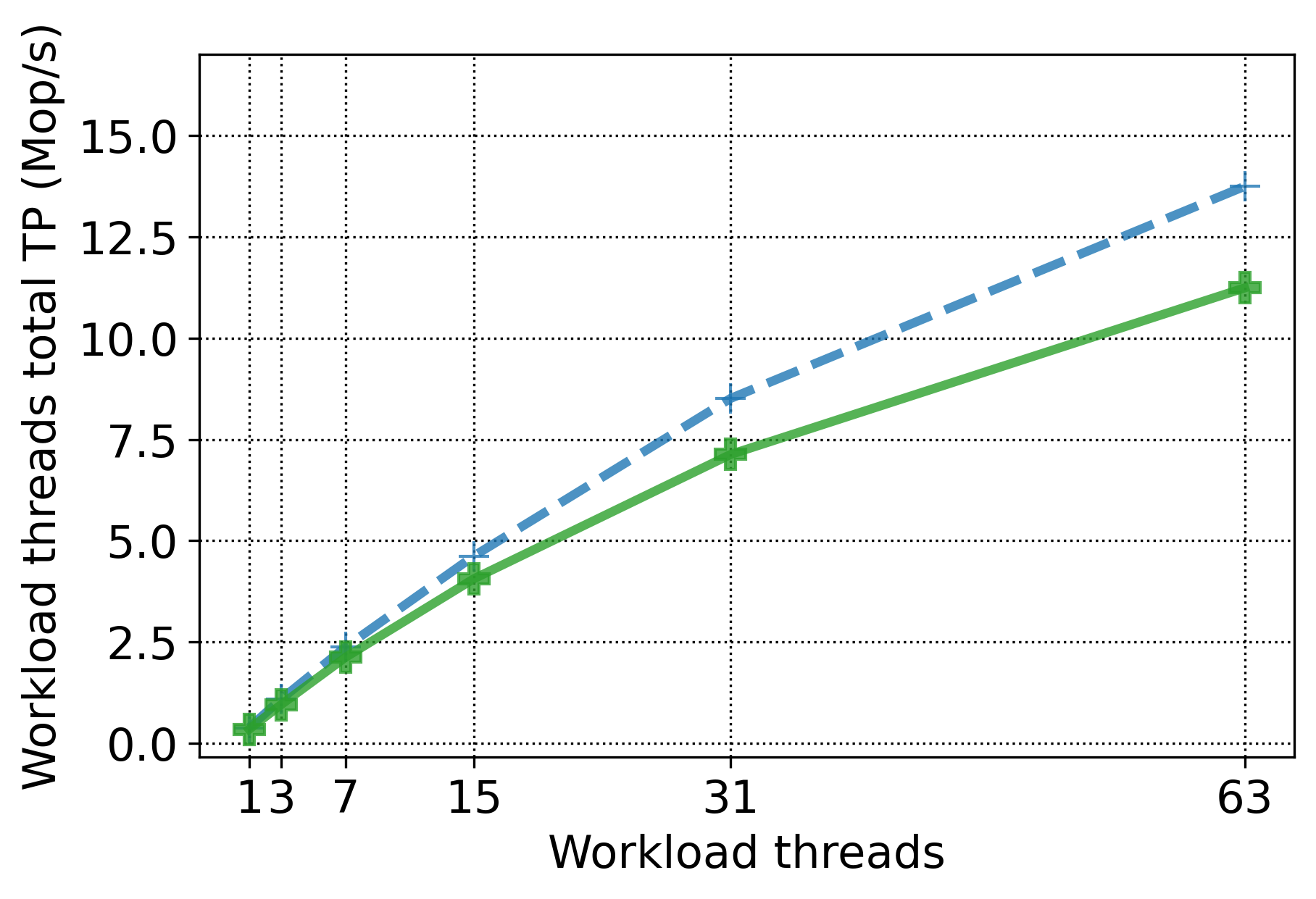}
  \raisebox{.55cm}{\includegraphics[width=.499\textwidth,trim={0 0 0 .4cm}]{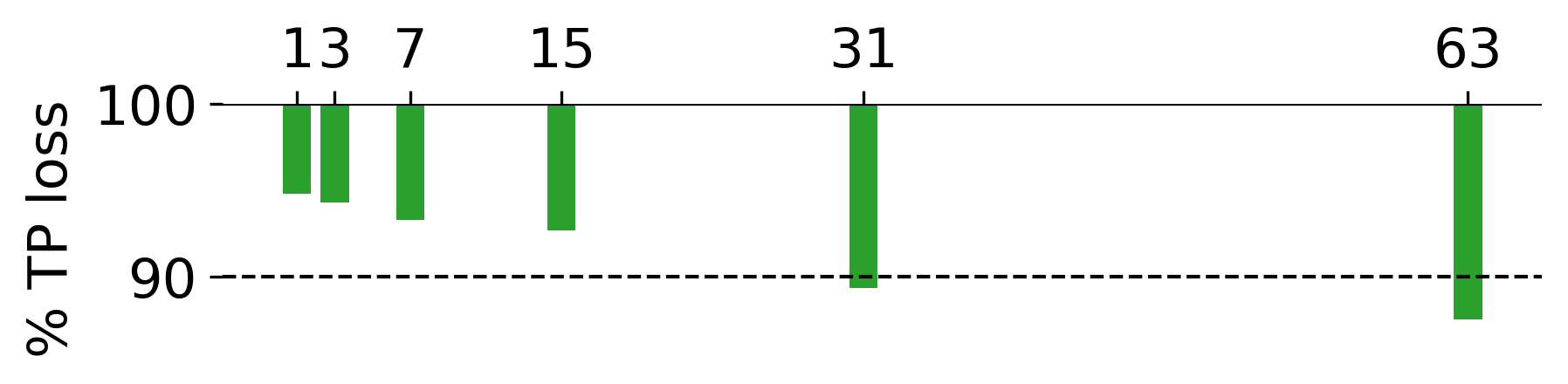}}\hspace*{.05mm}
  \includegraphics[width=.498\textwidth,trim={0 0 0 .4cm}]{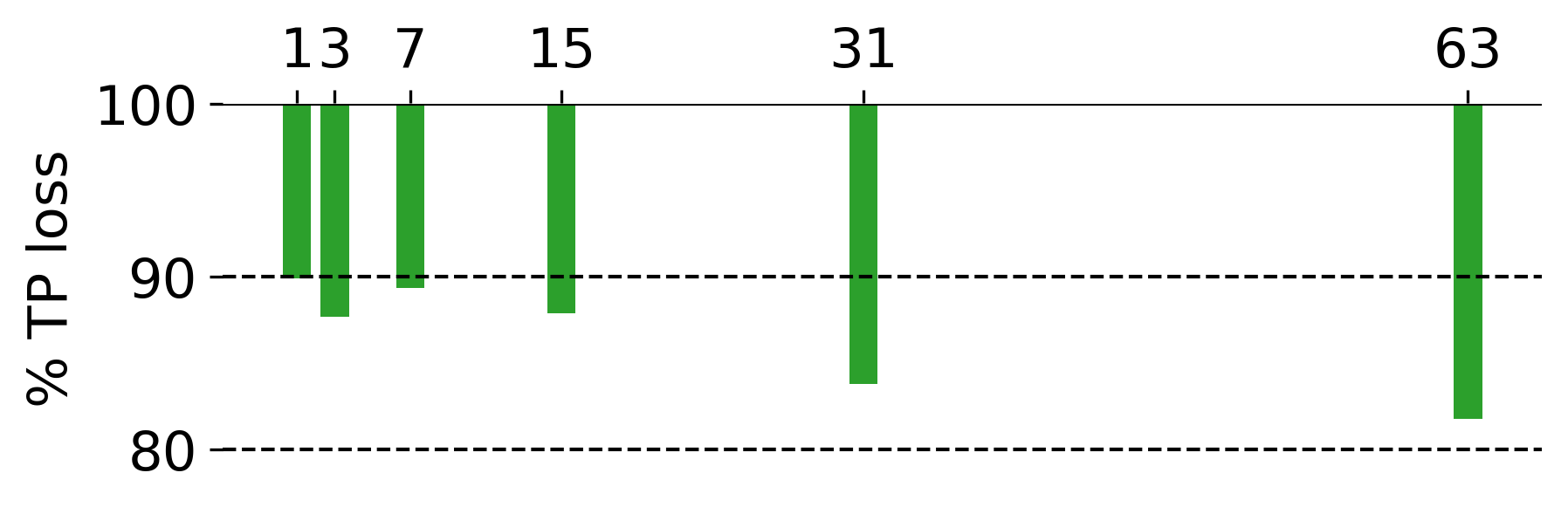}
  \caption{Overhead on skip list operations}
  \label{fig:SL overhead}
\end{figure*}

\paragraph{Overhead}

We measure the overhead of our methodology on the original data-structure operations by measuring the performance of workload threads---executing \ins{}, \del{} and \contains{} operations.
We compare the total throughput of $w$ workload threads, where $w$ varies from $1$ to $64$, for the transformed data structures versus the baseline data structures. The results appear in the top part of \Cref{fig:HT overhead,fig:BST overhead,fig:SL overhead}: the results for \shtb{} versus \htb{} appear in \Cref{fig:HT overhead}, for \sbst{} versus \bst{} in \Cref{fig:BST overhead}, and for \sskl{} versus \skl{} in \Cref{fig:SL overhead}.
They show the overhead when no concurrent \size{} operations are executed. To measure the overhead in the presence of \size{} calls as well, we similarly run $w$ workload threads, where $w$ varies from $1$ to $63$, while also running---for the transformed algorithms only---a concurrent size thread (that executes \size{} calls), and measure the total throughput of the workload threads. The results of these experiments appear in the bottom part of \Cref{fig:HT overhead,fig:BST overhead,fig:SL overhead}.

For each experiment, the top graph depicts the number of operations (\ins{}, \del{} and \contains{}) applied to the data structure per second by the workload threads altogether, measured in million operations per second. The curve of the baseline data structure appears along with the curve of its transformed version with \size{} support.
The bottom bar graph shows the throughput of the transformed data structure divided by that of the baseline data structure (in percentages), to demonstrate the throughput loss of the transformed data structure's operations. For instance, $90\%$ signify that the transformed workload threads reach $90\%$ of the throughput of the baseline workload threads. The throughput loss is worse for an update-heavy workload than for a read-heavy workload, and worse when a concurrent \size{} is executed. Still, the relative throughput in all experiments varies in the range of $80\%$ to $99\%$, i.e., a throughput loss of $1\%$ to $20\%$. We bring a breakdown of the overhead by operation type in \Cref{section:overhead-breakdown-by-op-type}.

\begin{figure*}[t]
  \centering
  \medskip
  \textit{\ \ \ \ \ \ \ \ \ \ \ Read heavy}\hfill
  \includegraphics[height=.03\textwidth]{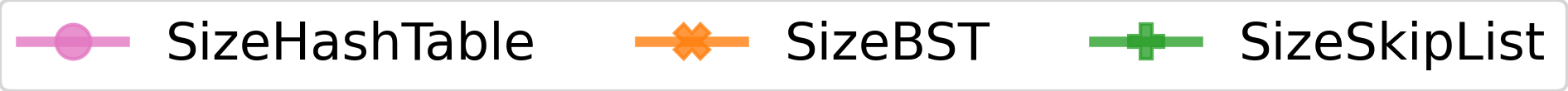}\hfill
  \textit{Update heavy\ \ \ }\par
  \medskip
  \includegraphics[width=.49\textwidth,trim={0 0 0 0.2cm},clip]{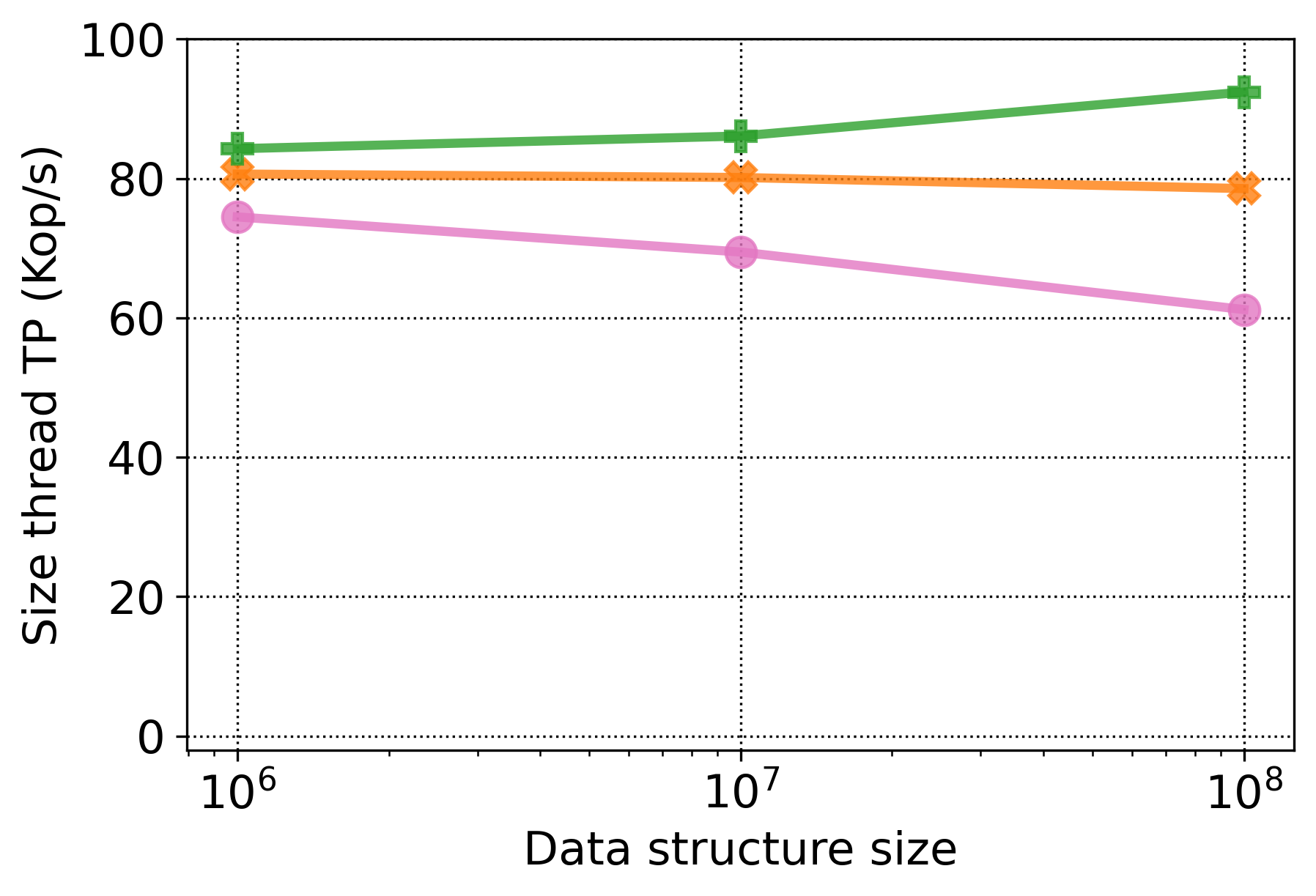}\hfill
  \includegraphics[width=.49\textwidth,trim={0 0 0 0.2cm},clip]{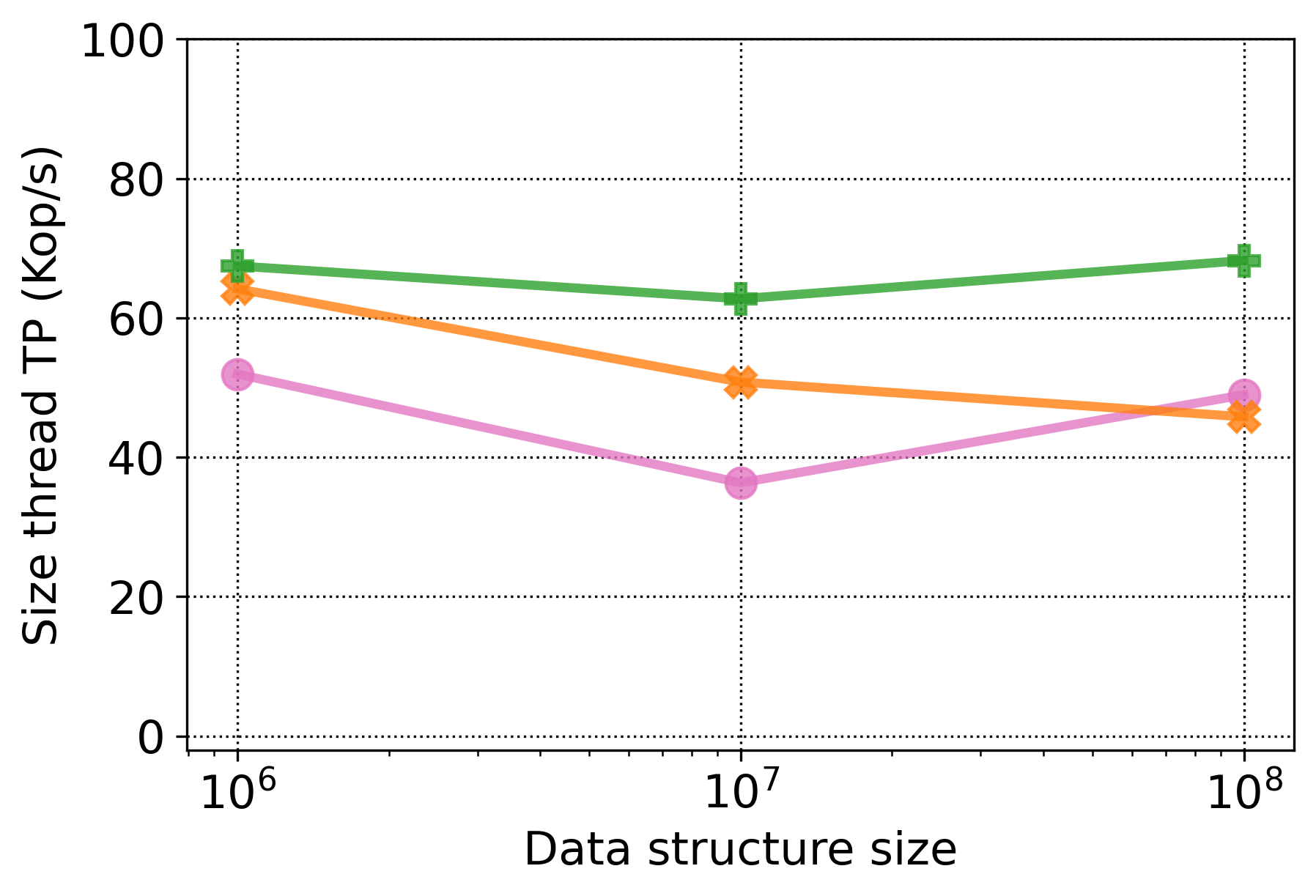}\par
  \caption{Size throughput as a function of data-structure size}
  \label{fig:per-size}
\end{figure*}

\begin{figure*}[t]
  \centering
  \medskip
  \textit{\ \ \ \ \ \ \ \ \ \ \ \ Read heavy}\hfill
  \includegraphics[height=.03\textwidth]{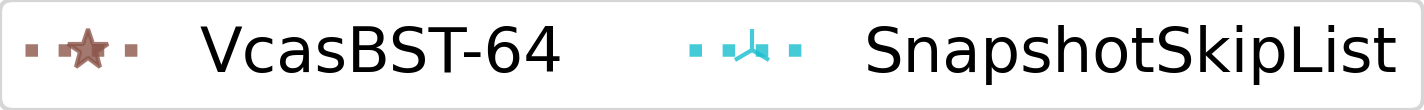}\hfill
  \textit{Update heavy\ \ \ \ }\par
  \medskip
  \includegraphics[width=.49\textwidth,trim={0 0.3cm 0 0},clip]{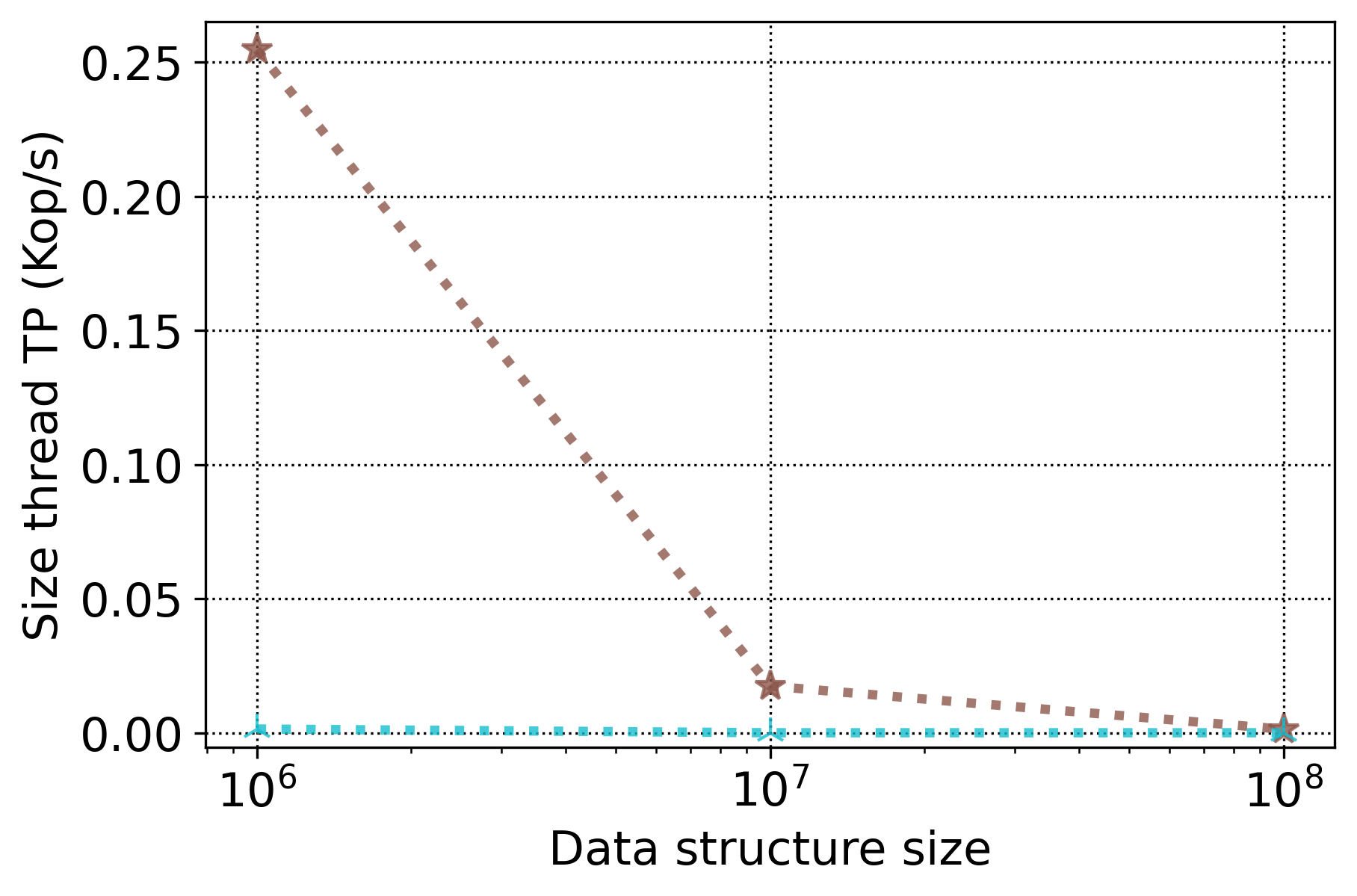}\hfill
  \includegraphics[width=.49\textwidth,trim={0 0.3cm 0 0},clip]{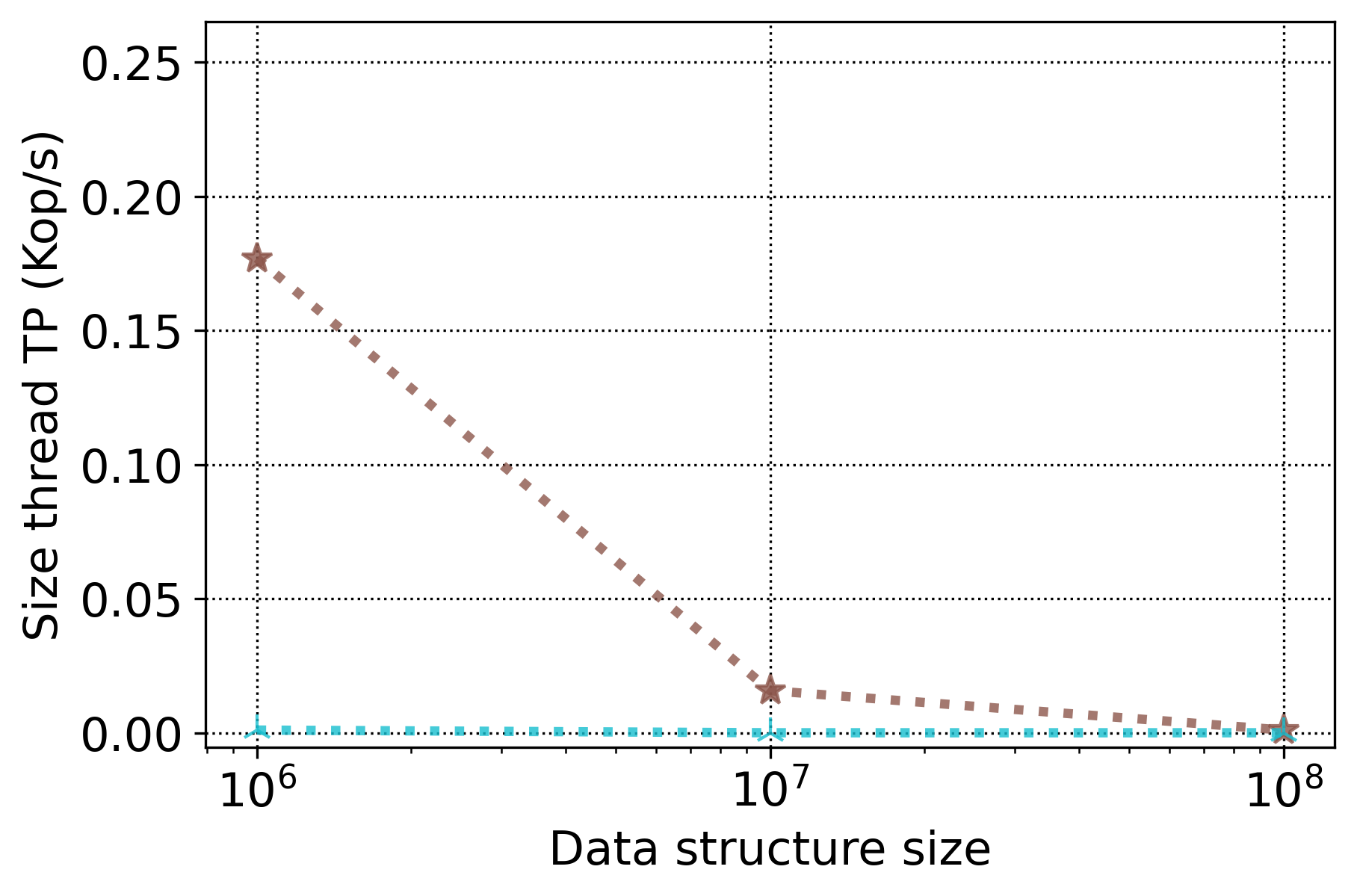}
  \caption{Snapshot-based size throughput as a function of data-structure size}
  \label{fig:per-size-snapshot}
\end{figure*}

\paragraph{Varying data-structure size}
To measure the effect of the number of elements in the data structure on the \size{} throughput,
we run experiments on different data-structure sizes, varying between $1$M and $100$M, with $32$ concurrent threads---one size thread and $31$ workload threads. \Cref{fig:per-size} presents the throughput of the size thread, measured in thousand \size{} operations per second. Each curve shows the size throughput for another transformed data structure, per different initial sizes. The results demonstrate that our size-computation methodology is not sensitive to the data-structure size. This is due to the metadata array, on which the \size{} operates instead of traversing the data structure itself.
In contrast, obtaining the size using a snapshot-based method causes performance degradation as the size increases, as shown for \vcasbst{} in 
\Cref{fig:per-size-snapshot} which presents the corresponding graphs for the competitors. \snapskl{} demonstrates a very low size throughput: for a data-structure size of $1$M it executes $1.4$ \size{} operations per second in average for the read-heavy workload and $1$ \size{} operation per second for the update-heavy workload; and for bigger data-structure sizes it executes less than $1$ \size{} operation per second.

\paragraph{Scalability}
To assess the scalability of the \size{} operation, we run $s$ size threads, where $s$ varies between $1$ and $16$, concurrently with $32$ workload threads. \Cref{fig:scalability} presents the total throughput of all size threads, measured in thousand \size{} operations per second. It shows results for both our transformed data structures, and the snapshot-supporting data structures which demonstrate inferior performance.
For each of our transformed data structures, the throughput improves as number of size threads increases. This demonstrates the scalability of our methodology.

\begin{figure*}
  \centering
  \medskip
  \hfill\includegraphics[height=.03\textwidth]{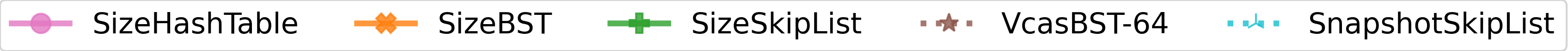}\hspace*{1.5mm}\par
  \textit{\ \ \ \ \ \ \ \ \ \ \ Read heavy}\hfill
  \textit{Update heavy\ \ \ \ }\par
  \medskip
  \includegraphics[width=.49\textwidth]{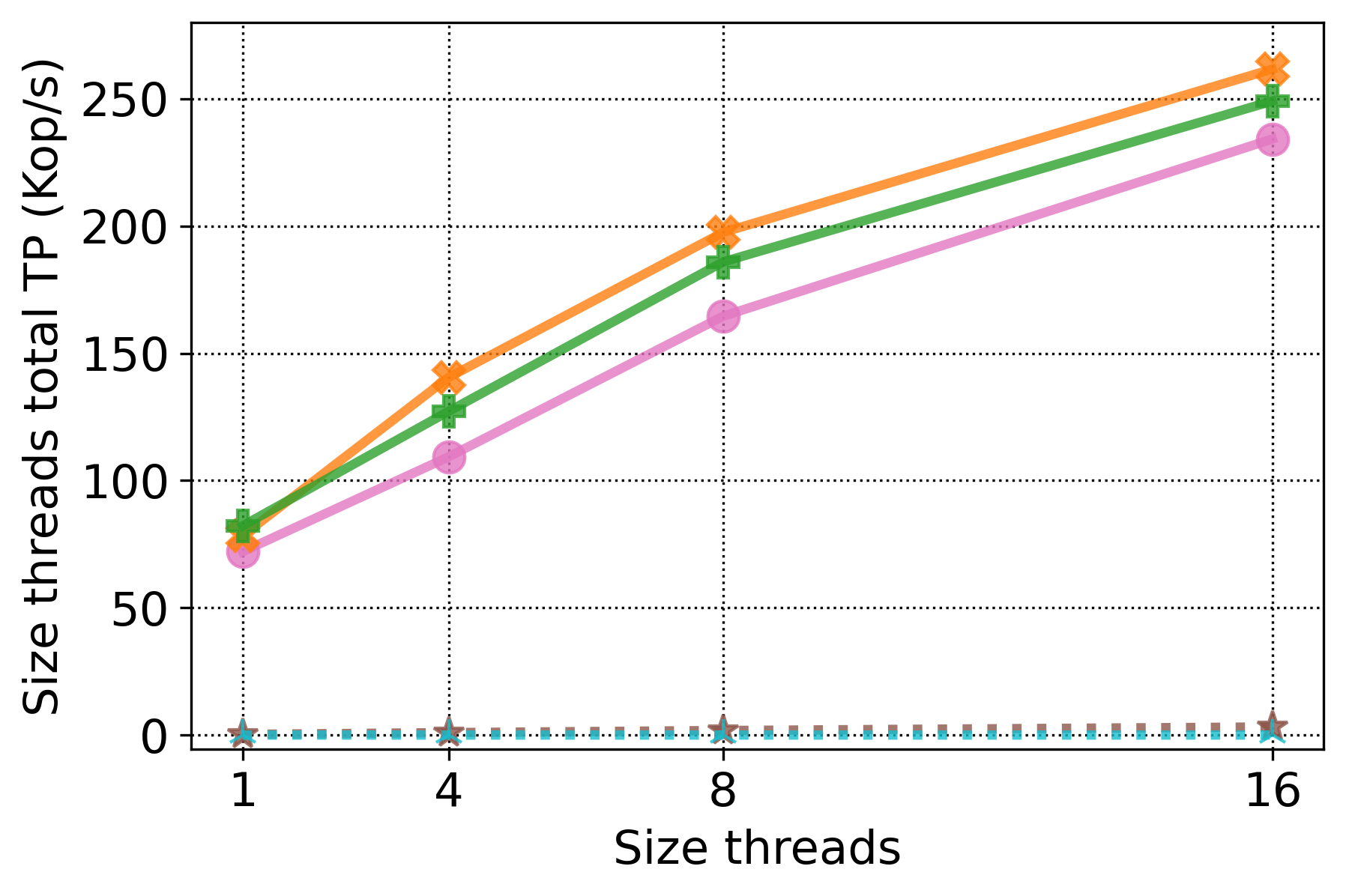}\hfill
  \includegraphics[width=.49\textwidth]{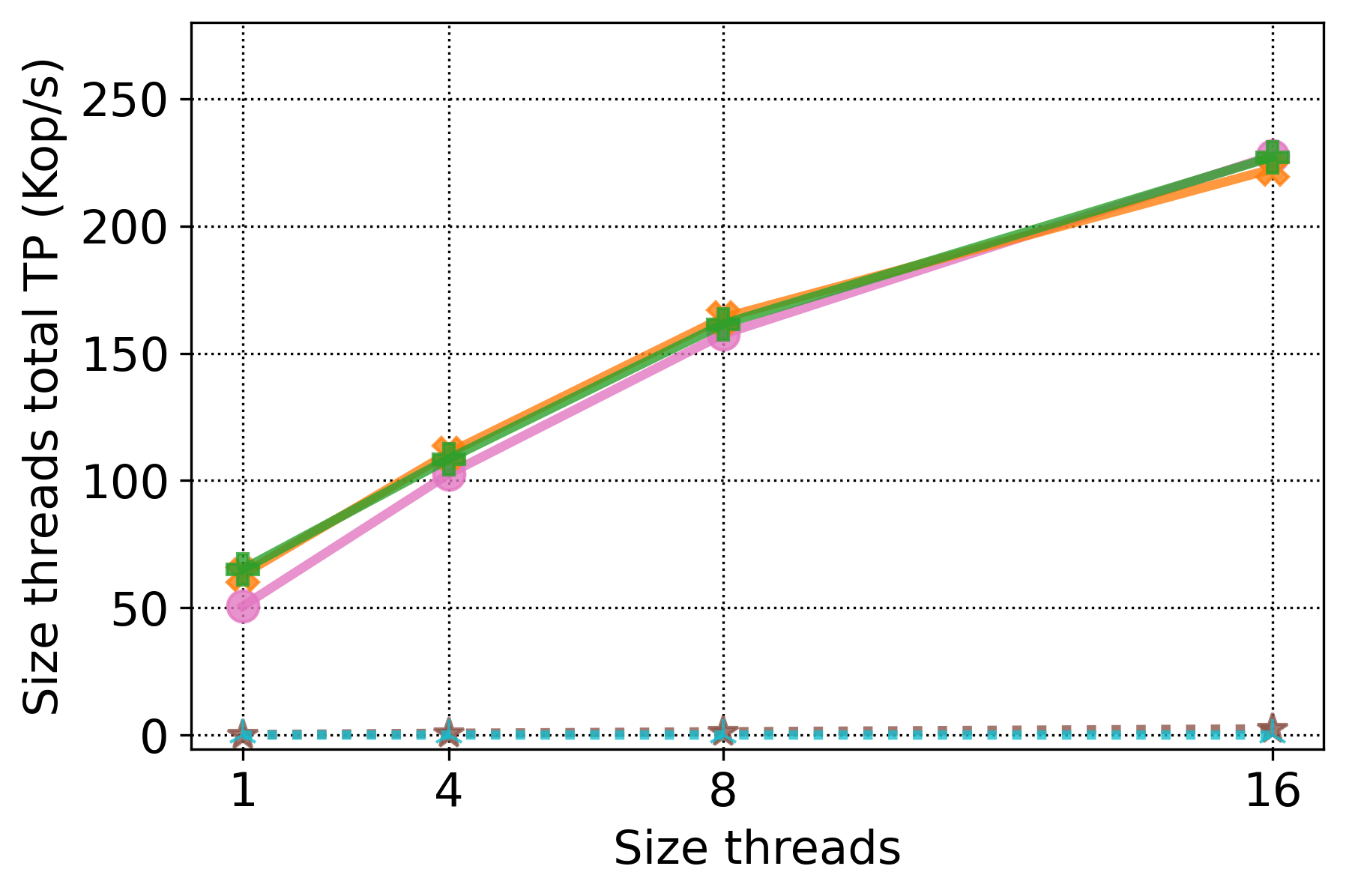}
  \caption{Size scalability}
  \label{fig:scalability}
\end{figure*}

\paragraph{Comparison to snapshot-based size}
Our transformed data structures yield a much better throughput than the competitors, as demonstrated in \Cref{fig:per-size,fig:per-size-snapshot,fig:scalability}:
\sskl{} demonstrates in these experiments a throughput at least $54806\times$ the throughput of \snapskl{} (in some experiments, not even a single \size{} operation on \snapskl{} completed within $5$ seconds).
The throughput of \sbst{} in these experiments is between $83-60423\times$ the throughput of \vcasbst{}. The performance gap between our transformed data structures and \vcasbst{} is not as large as the gap from \snapskl{}, because \vcasbst{} succeeds to improve snapshot performance in comparison to \snapskl{}, but not without a cost---it pays with higher space overhead.

\subsection{Overhead Breakdown by Operation Type}\label{section:overhead-breakdown-by-op-type}

\begin{figure*}
  \centering
  \medskip
  \textit{\ \ \ \ \ \ \ \ \ \ \ \ Read heavy}\hfill
  \includegraphics[height=.03\textwidth]{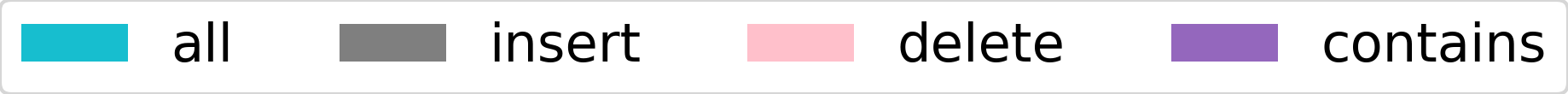}\hfill
  \textit{Update heavy\ \ \ \ }\par
  \vspace{1mm}
  \text{\shtb{} vs \htb{} without a concurrent size thread}\par
  \vspace{2.2mm}
  \includegraphics[width=.4975\textwidth,trim={0 .7cm 0 .4cm}]{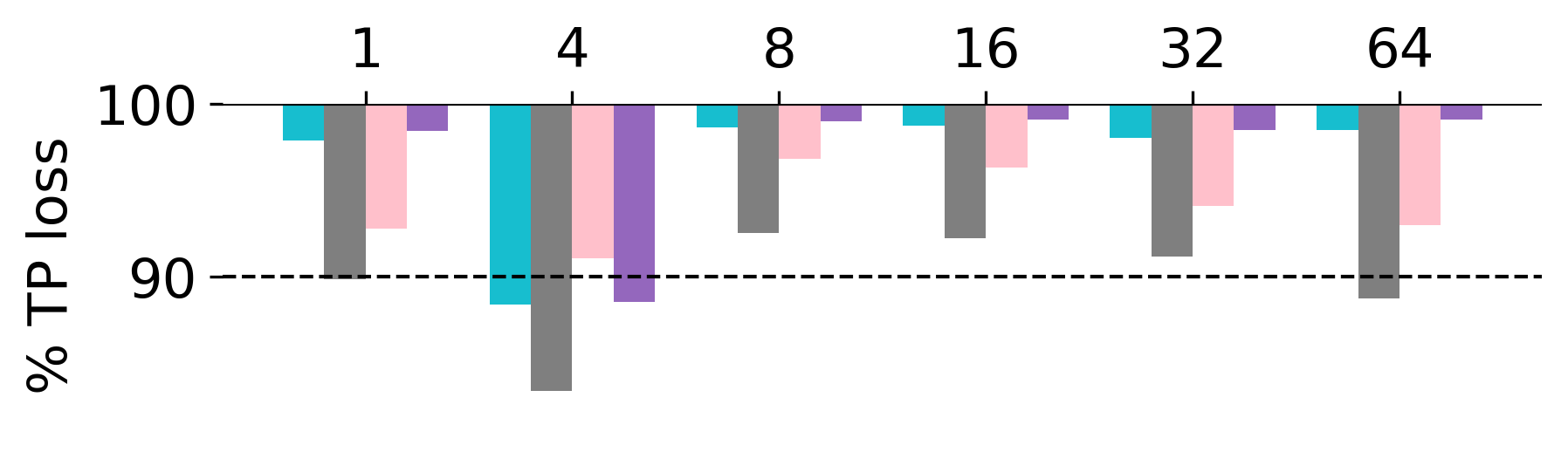}\hspace*{0.001mm}
  \raisebox{.3cm}{\includegraphics[width=.4975\textwidth,trim={0 .7cm 0 .4cm}]{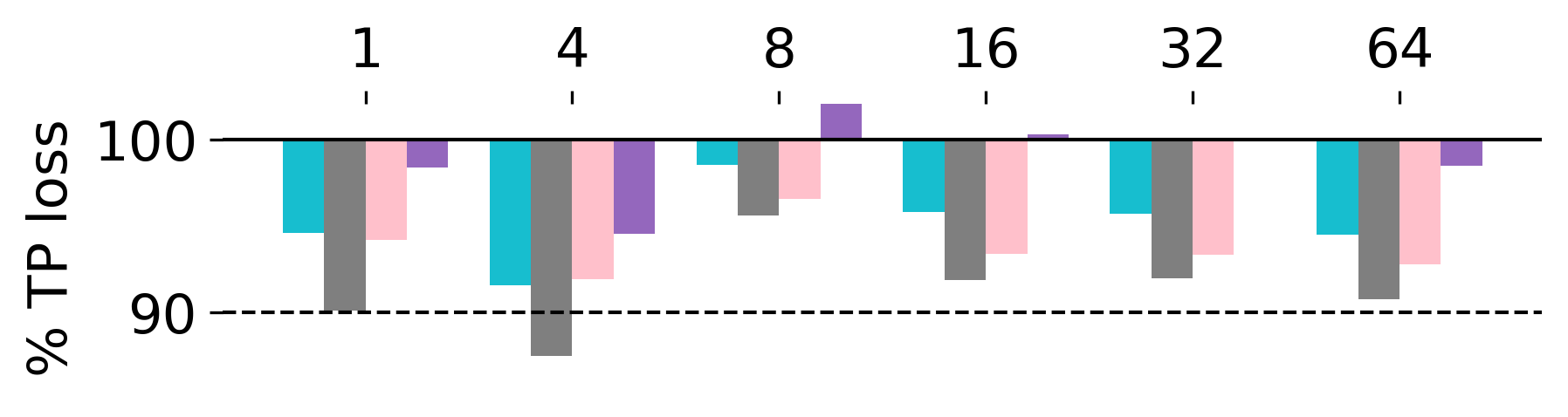}}
  \text{\shtb{} vs \htb{} with a concurrent size thread}\par
  \vspace{2.2mm}
  \raisebox{.32cm}{\includegraphics[width=.4975\textwidth,trim={0 .7cm 0 .4cm}]{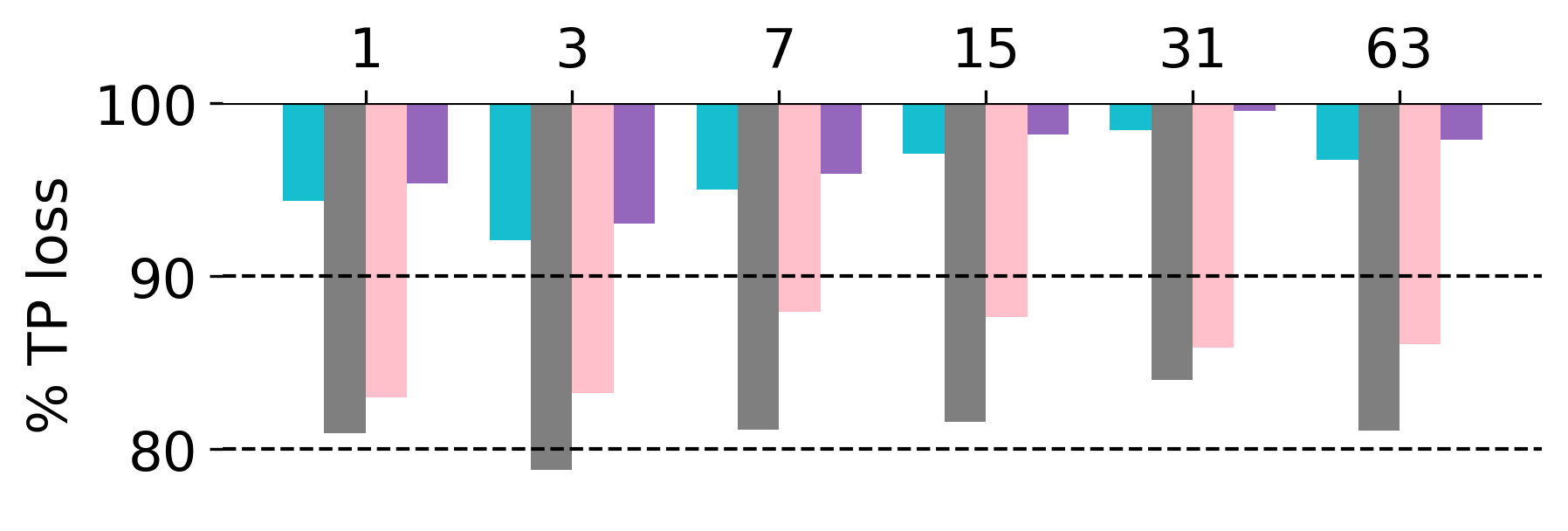}}\hspace*{0.01mm}
  \includegraphics[width=.4975\textwidth,trim={0 .7cm 0 .4cm}]{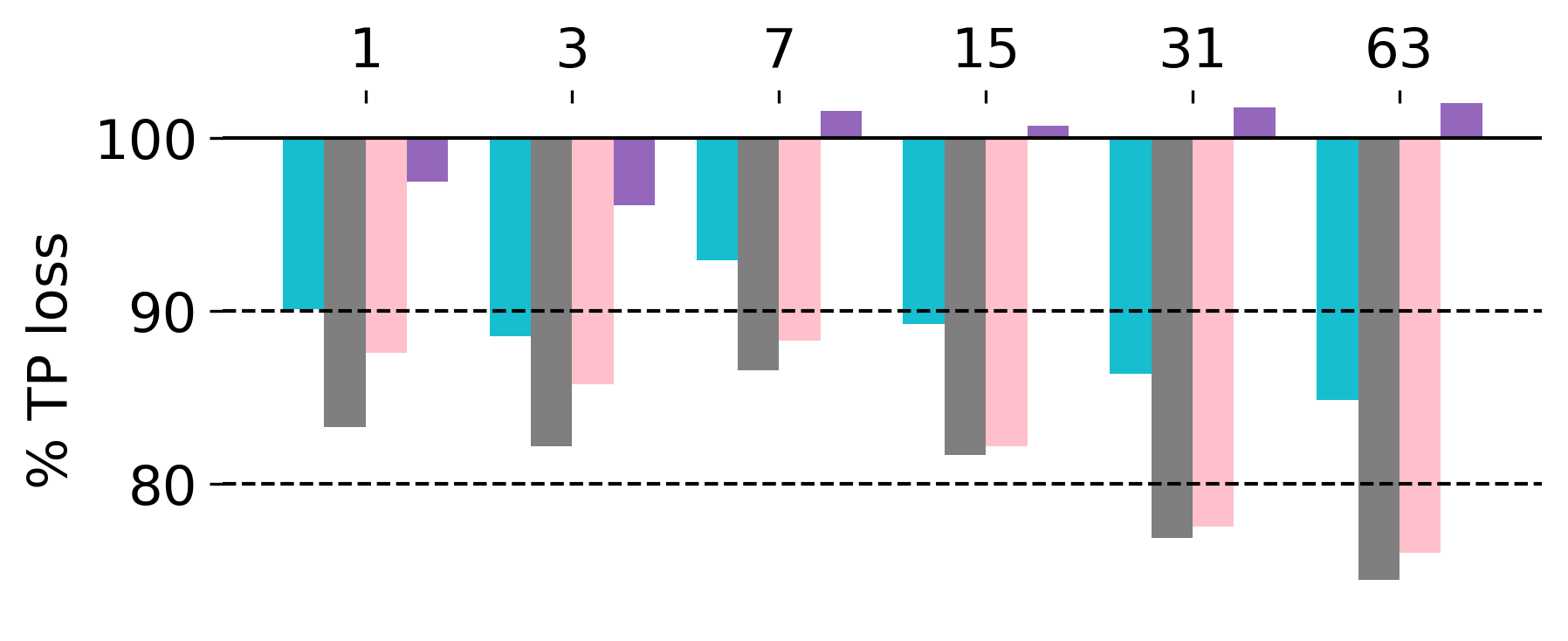}
  \text{\sbst{} vs \bst{} without a concurrent size thread}\par
  \vspace{2.2mm}
  \raisebox{.2cm}{\includegraphics[width=.499\textwidth,trim={0 0.7 0 .4cm}]{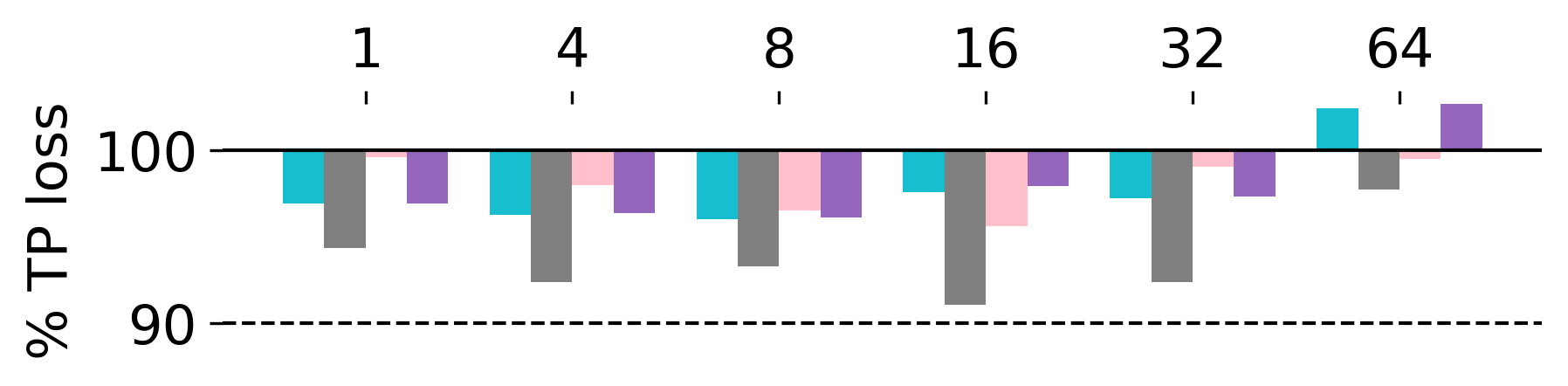}}\hspace*{0.001mm}
  \includegraphics[width=.4985\textwidth,trim={0 0.7 0 .4cm}]{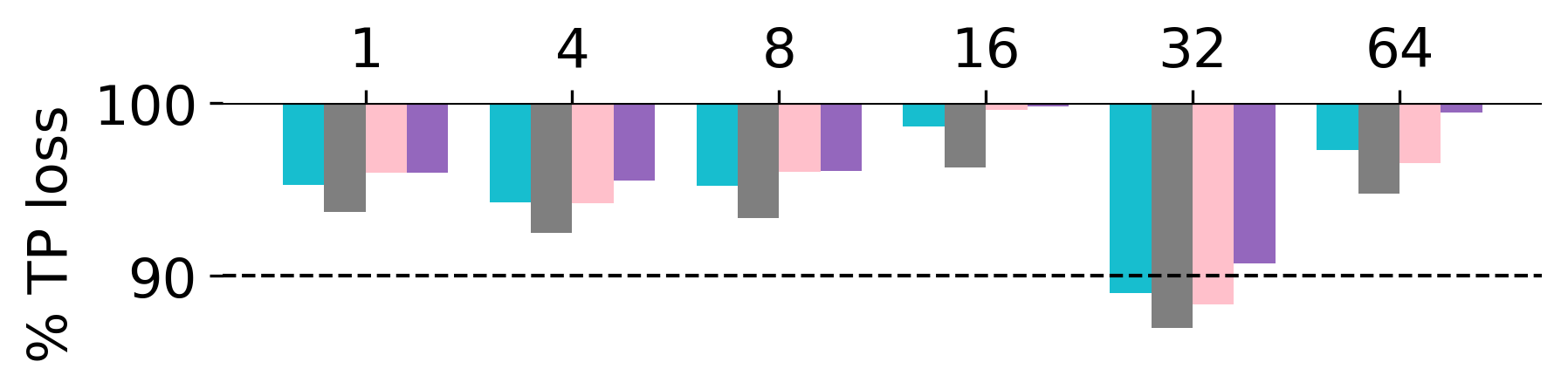}
  \text{\sbst{} vs \bst{} with a concurrent size thread}\par
  \vspace{2.2mm}
  \includegraphics[width=.497\textwidth,trim={0 0.7 0 .4cm}]{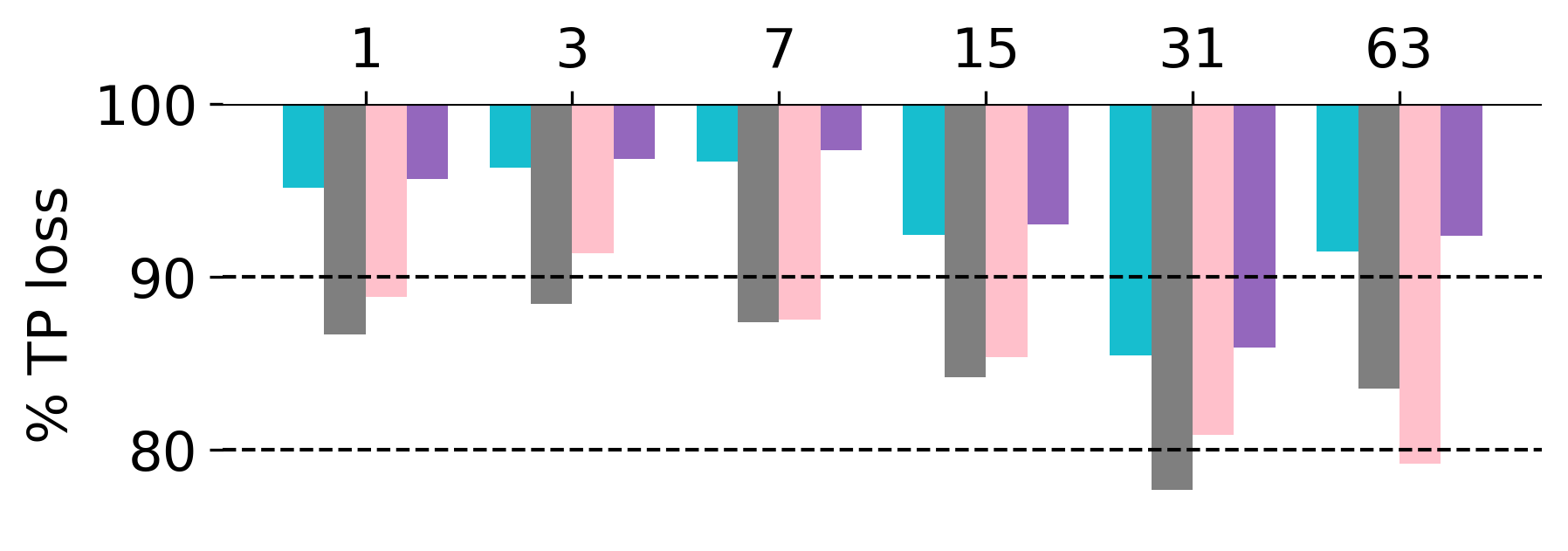}\hspace*{0.01mm}
  \raisebox{.2cm}{\includegraphics[width=.497\textwidth,trim={0 0.7 0 .4cm}]{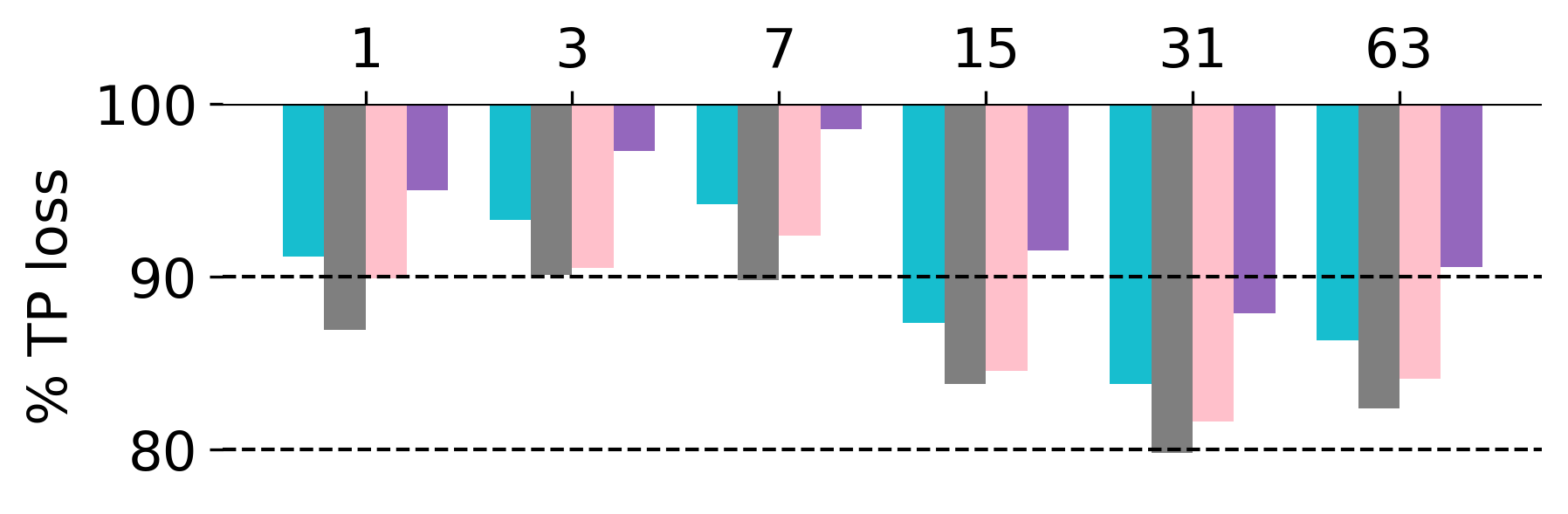}}
  \text{\sskl{} vs \skl{} without a concurrent size thread}\par
  \vspace{2.2mm}
  \raisebox{.2cm}{\includegraphics[width=.4995\textwidth,trim={0 1.3 0 .4cm}]{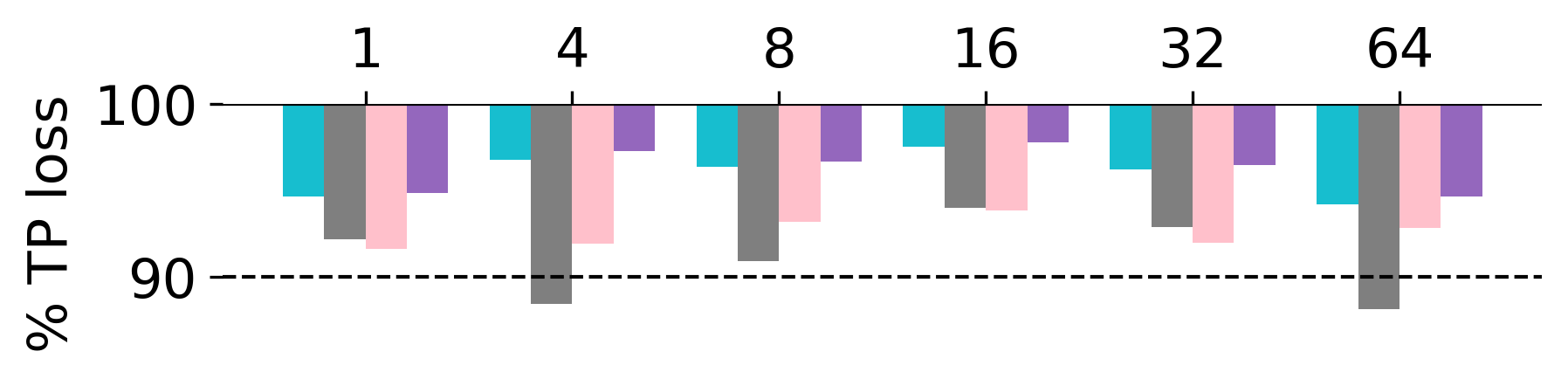}}\hspace*{0.5mm}
  \includegraphics[width=.4995\textwidth,trim={0 1.3 0 .4cm}]{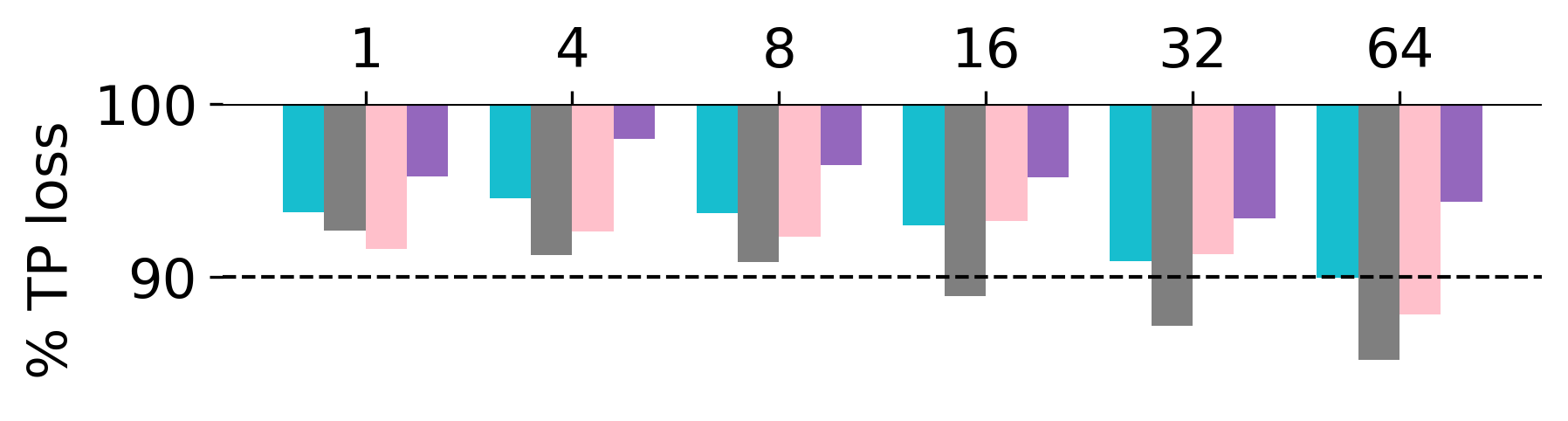}
  \text{\sskl{} vs \skl{} with a concurrent size thread}\par
  \vspace{2.2mm}
  \includegraphics[width=.499\textwidth,trim={0 0.7 0 .4cm}]{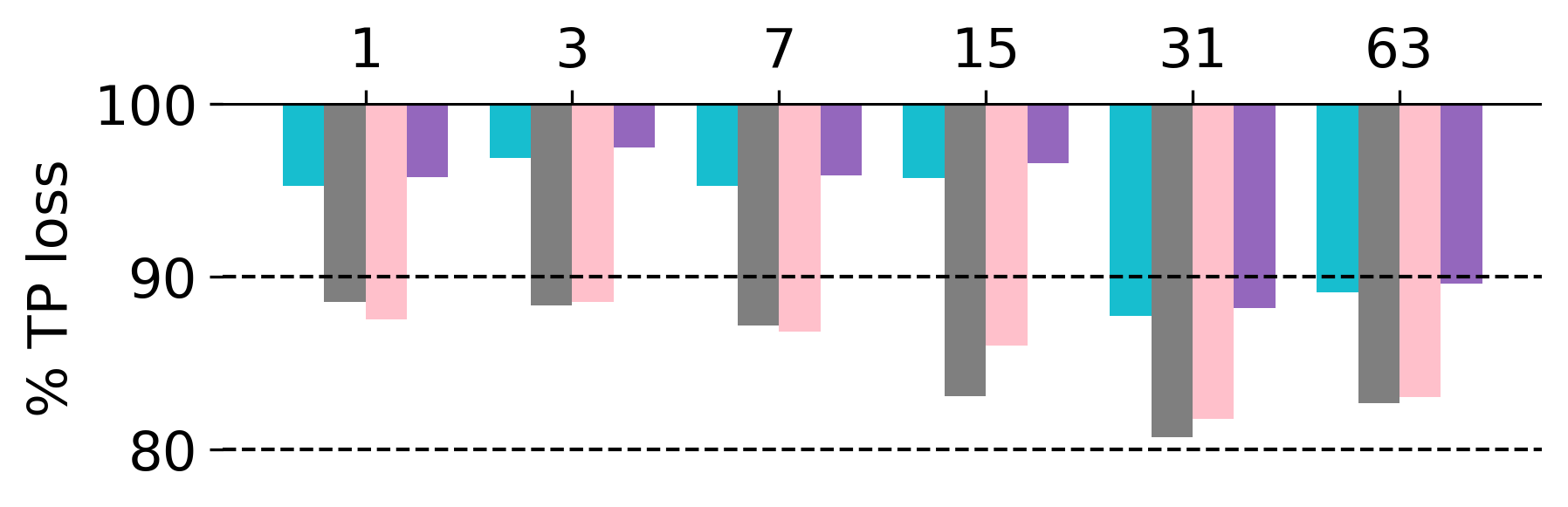}\hspace*{.05mm}
  \raisebox{.2cm}{\includegraphics[width=.498\textwidth,trim={0 0.7 0 .4cm}]{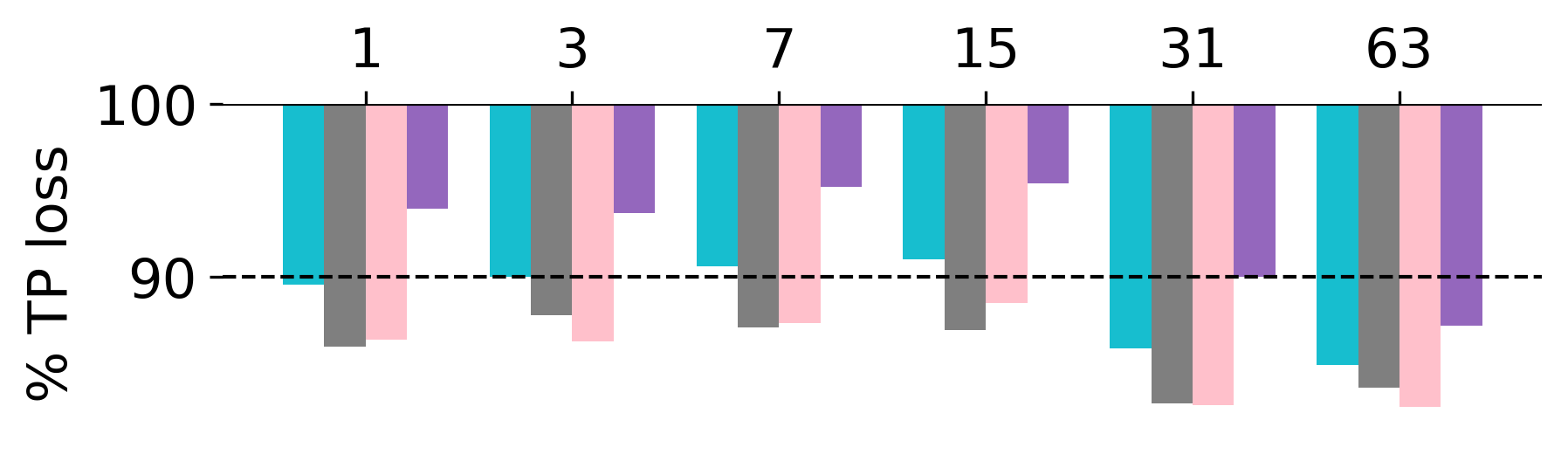}}
  \caption{Overhead breakdown by operation type}
  \label{fig:overhead breakdown}
\end{figure*}

We performed measurements to assess the overhead breakdown by operation type (\ins{} / \del{} / \contains{}).
Similarly to the above overhead measurements, we compare the performance of workload threads for the transformed data structures versus the baseline data structures. But here, in addition to comparing the combined throughput of all three types of operations by all workload threads (i.e., the total number of all operations divided by the total time they ran),
we compare also the total throughput of all workload threads {\em per operation type} (namely, the total number of insertions by all threads divided by the total time the insertions ran, and the same for deletions and for \contains{} calls). The results appear in \Cref{fig:overhead breakdown}.
In most measurements, the throughput loss is highest for \ins{} operations and lowest for \contains{} operations.

\section{Conclusion}\label{section:conclusion}

In this work we addressed the problem of obtaining a correct size of a concurrent data structure. We showed that existing solutions in the literature are either inefficient or incorrect (even in a very liberal sense). We then presented a methodology for adding a linearizable \size{} operation to concurrent data structures that implement sets or dictionaries.
Our methodology was shown to yield attractive theoretical properties in terms of progress guarantees and asymptotic complexity. Evaluation demonstrated that while incurring some overhead on the data-structure's original operations, the methodology yields a \size{} operation that provides an orders-of-magnitude performance improvement over existing solutions. We additionally illustrated that the \size{} operation is scalable and insensitive to the data-structure size.

\bibliographystyle{ACM-Reference-Format}
\bibliography{refs}

\appendix
\section[Inaccuracies of the algorithm in Afek et al.]{Inaccuracies of the algorithm in \cite{Afek2012size}}\label{section:prev paper incorrect}

We showed in \Cref{section:intro} that updating the data structure and the size-related metadata separately, as is done in the algorithm for data-structure size presented in \cite{Afek2012size}, makes the algorithm non-linearizable and also not satisfying the weakest correctness principle defined in \cite{herlihy2008art} which requires method calls to appear to happen in a one-at-a-time sequential order.
But even if update operations somehow update the data structure and the metadata atomically, the algorithm of \citet{Afek2012size} will still not satisfy the above-mentioned correctness principle, due to another issue we elaborate on next. We start with a demonstrating execution, which produces an impossible negative size, hence, no reordering of its method calls forms a legal sequential execution.

Consider an execution with 3 threads executing their operations concurrently: thread $T_0$ executes an insertion of an item to the data structure, thread $T_1$ executes a deletion of the same item from the data structure, and thread $T_2$ executes a \size{} call. 
First, $T_0$ and $T_1$ start executing their operations. $T_0$ inserts the item to the data structure and then $T_1$ successfully removes it. After operating on the data structure, they both call the algorithm's $wait\_free\_update$ method to update their values in the array $g\_mem$. During these method calls, they obtain $g\_seq$ when its value is $0$, and before they proceed to updating $g\_mem$, $T_2$ starts its \size{} execution.
It performs $scan\_seq:=FAI(g\_seq)$ which results in $scan\_seq=1$, and later starts collecting $g\_mem$'s values. It obtains $[0,0]$ from $g\_mem[0].\_recent$ and adds $0$ to $size$. At this point, $T_0$ resumes its execution, and writes $[1,0]$ to $g\_mem[0].\_recent$ (this value was already missed by $T_2$). Then $T_1$ resumes its execution, and writes $[-1,0]$ to $g\_mem[1].\_recent$. Now $T_2$ continues scanning the array. It obtains $[-1,0]$ from $g\_mem[1].\_recent$ and accordingly adds $-1$ to size, and then $[0,0]$ from $g\_mem[2].\_recent$ and adds $0$ to size. Subsequently, it returns the incorrect size $-1$. 

To analyze how this happened, we examine the linearization points of the operations on the array $g\_mem$.
The linearization point of the size operation by $T_2$ is when it increments $g\_seq$ using $FAI$; the linearization point of the insertion by $T_0$ is when it writes to $g\_mem[0].\_recent$.
The problem stems from the linearization point of the deletion by $T_1$. It cannot be placed (like erroneously mentioned in \cite{Afek2012size}) when $T_1$ writes to $g\_mem[1].\_recent$, because it must occur before the linearization point of the size operation that observed the deletion. Instead, it is placed in retrospect right before the linearization point of the size operation. This linearization scheme of possibly placing in retrospect linearization points of updates that the size operation observes in its scan long before they write their value, is adopted from the single-scanner algorithm of Riany et al. \cite{riany2001towards}, on which the size algorithm in \cite{Afek2012size} is based. The scheme was intended for the atomic snapshot problem, where there are no dependencies between the update operations. However, when handling dependent data-structure operations, they cannot be freely reordered when the size observes them in its scan.
In the execution described above, reversing the order of an insertion and a following deletion of the same item is unacceptable, since the deletion cannot succeed if it happens before the insertion, and thus cannot legally decrement the size before the insertion increments it.
In conclusion, the correctness problem of the suggested size algorithm stems from the linearization order the size operation dictates: if a size operation $S$ observes during its scan a value, written after $S$'s linearization point by a delayed update $U$, it dictates in retrospect to place $U$'s linearization point before $S$'s linearization point -- which might occur before linearization points of update operations that $U$ depends on.

\section{Linearizability proof continued}\label{section:abstract set lin proof}
We complete the missing parts in the linearizability proof in \Cref{section:Linearizability proof}.
We start with a proof of \Cref{lemma:counter inc after first 2 lines of updateMetadata}.

\begin{proof}
We prove by induction. Assume the lemma holds for $c-1$.
The metadata counters are modified only in \Cref{code:SizeCalculator updateMetadata update array end} by increments using \cas{}. Hence, it is enough to prove that when the \codestyle{updateMetadata} call starts, the relevant metadata counter's value is $\geq c-1$.
\codestyle{updateMetadata} is called with an \codestyle{UpdateInfo} instance associated with $op$ after this instance has been published in the relevant node. This publication is done by the thread $T$ when it executes $op$, and as each thread executes its operations sequentially, $T$ has completed its ($c-1$)-st successful operation of the same kind (insertion or deletion) by this time. During that operation, $T$ called \codestyle{updateMetadata} on its behalf, so by the induction hypothesis, the relevant metadata counter's value is $\geq c-1$ when $T$ completes that previous operation. Since the counters are monotonically increasing, we are done.
\end{proof}

To complete the linearizability proof (presented in \Cref{section:Linearizability proof}), it remains to prove \Cref{claim:results comply}. 
In what follows, we denote the set's $i$-th successful \ins$(k)$ operation (by {\em $i$-th} we refer to the linearization order, namely, to the $i$-th successful \ins$(k)$ to be linearized) by \ins$_i(k)$, its linearization time by $t_{insert_i(k)}$, and the time of its original linearization by $orig\_t_{insert_i(k)}$. We further denote the analogous \del{} operation and its related times by \del$_i(k)$, $t_{delete_i(k)}$ and $orig\_t_{delete_i(k)}$.

\begin{claim}\label{claim:results comply}
Consider a sequential history formed by ordering an execution's operations (with their results) according to their linearization points defined in \Cref{section:linearization points}. 
Then operation results in this history comply with the sequential specification of a set.
\end{claim}
\begin{proof}
As for the results of successful update operations, their correctness follows directly from \Cref{corollary:alternating ins and del}:
The last successful update operation on $k$ to be linearized before the linearization point of a successful \ins{}($k$) operation is a deletion, thus, the key $k$ is logically not in the set at the moment of the insertion's linearization and the insertion correctly succeeds. Similarly, the last successful update operation on $k$ to be linearized before the linearization point of a successful \del{}($k$) operation is an insertion, thus, the key $k$ is logically in the set at the moment of the deletion's linearization and the deletion correctly succeeds.

Now, let us examine the results of \contains{} operations and failing update operations. Let $op$ be such an operation on a key $k$, and let the operation it depends on (namely, the last successful update operation on $k$ whose original linearization point precedes $op$'s original linearization point) be \ins{}$_i(k)$ for some $i\geq1$ (the proof for a \del{} operation is similar). 
As \ins{}$_i(k)$ is an insertion, $op$ must be a \contains{} operation returning \codestyle{true} or a failing \ins{} operation.
To show that $op$'s result---which reflects that the last operation on the set was an insertion---is legal, we will prove that the linearization point of $op$ occurs when the last linearized successful update operation on $k$ is the insertion \ins{}$_i(k)$.
Let $orig\_t_{op}$ be the original linearization moment of $op$. 
There are two possibilities with regards to $op$'s linearization point:
either $op$ is linearized immediately after $t_{insert_i(k)}$, and we are done, or it is linearized at $orig\_t_{op}$. In the latter case, according to the linearization point definition, \ins{}$_i(k)$ must be linearized by $op$'s original linearization moment, namely, $t_{insert_i(k)} < orig\_t_{op}$.
If no successful \del{}($k$) operation is linearized after $t_{insert_i(k)}$, then we are done. Else, $orig\_t_{op}<orig\_t_{delete_i(k)}$, as \ins{}$_i(k)$ is the last successful update operation on $k$
whose original linearization point precedes $orig\_t_{op}$. Since $orig\_t_{delete_i(k)}<t_{delete_i(k)}$ (by \Cref{lemma:ins and del order}), then $orig\_t_{op}<t_{delete_i(k)}$, and we are done.

Finally, we analyze the linearization of a \size{} operation. Denote such an operation by $op$, the \codestyle{CountersSnapshot} instance it obtains and operates on by \term{countersSnapshot}, and the \size{} call that sets the \term{countersSnapshot}.\codestyle{size} field by \term{determiningSize}. $op$ returns the difference between the sum of insertion counters and the sum of deletion counters that were observed in the \term{countersSnapshot}.\codestyle{snapshot} array by \term{determiningSize}.
Let $j$ be the value that \term{determiningSize} obtained from the insertion counter of some thread $T$ in \term{countersSnapshot}.\codestyle{snapshot}. We will prove that $T$'s $j$-th successful \ins{} is linearized before $op$'s linearization point, and $T$'s ($j+1$)-st successful \ins{} (if such an operation occurs) is linearized after it. We refer to insertions for convenience, but the exact same proof applies  to the deletion counters as well.

We start with $T$'s $j$-th successful \ins{}. 
Since \term{determiningSize} obtained the value $j$ from the relevant snapshot counter, then by \Cref{lemma:update preceding size}, the metadata counter update on behalf of $T$'s $j$-th successful \ins{} happens before $op$'s linearization point. If $T$'s $j$-th successful \ins{} is linearized in its metadata counter update, we are done. Else, it is linearized immediately after the linearization point of a \size{} operation, whose collecting \codestyle{CountersSnapshot} instance is announced when the metadata counter is updated, and which reads a value $<j$ from the relevant snapshot counter. This \size{} operation cannot be $op$ (which reads the value $j$), but rather a preceding \size{} operation whose \codestyle{CountersSnapshot} instance is announced prior to \codestyle{countersSnapshot} (because it is already announced when the metadata counter is updated on behalf of $T$'s $j$-th successful \ins{}, which by \Cref{lemma:update preceding size} happens before \term{countersSnapshot}'s \codestyle{collecting} field is set to \codestyle{false}), so its linearization point precedes $op$'s linearization point.

We proceed to $T$'s ($j+1$)-st successful \ins{} (in case such an operation occurs). If in its metadata counter update, \term{countersSnapshot} is announced in the \codestyle{SizeCalculator} instance held by the set and its \codestyle{collecting} field's value is \codestyle{true}, then this insertion is linearized immediately after $op$'s linearization point, because \term{determiningSize} obtained the value $j$ (which is smaller than $j+1$) from the corresponding \term{countersSnapshot}.\codestyle{snapshot}'s counter. Else, 
this metadata counter update must have occurred after $op$'s linearization point, since the alternative is that \term{countersSnapshot} is announced after that metadata counter update, in which case a value $\geq j+1$ must be collected in \term{countersSnapshot}.\codestyle{snapshot}. 
The insertion is linearized either at its metadata counter update or later, thus, linearized after $op$'s linearization point in this case as well.
\end{proof}

The proof of \Cref{claim:results comply} uses the following:

\begin{observation}\label{observation:alternating original lps}
The original linearization points of successful insertions and deletions of each key $k$ are alternating. 
\end{observation}
This follows from the linearizability of the original data structure and the sequential specification of a set.

\begin{lemma}\label{lemma:ins and del order}
For each key $k$ and each $i\geq1$: 
\[
orig\_t_{insert_i(k)} < t_{insert_i(k)} < orig\_t_{delete_i(k)} < t_{delete_i(k)} < orig\_t_{insert_{i+1}(k)}
\]
\end{lemma}
\begin{proof}
The linearization point of each successful \ins{} or \del{} operation happens after its original linearization point because the linearization point occurs at the metadata counter update or later, and this update is performed in our transformation after the original linearization point.

In addition, before \del$_i(k)$ carries out its own original linearization point, i.e., marking the node it is deleting, it calls \codestyle{updateMetadata} on behalf of the \ins{} operation that inserted that node. By \Cref{observation:alternating original lps}, the last original linearization point of a successful update operation on $k$ before the one of \del$_i(k)$ is that of \ins$_i(k)$. Thus, the node that \del$_i(k)$ deleted was inserted by \ins$_i(k)$, and \del$_i(k)$ calls \codestyle{updateMetadata} with the \codestyle{insertInfo} associated with \ins$_i(k)$.
By \Cref{lemma:lin when updateMetadata returns}, \ins$_i(k)$ will have been linearized by the time this \codestyle{updateMetadata} call returns. Hence, $t_{insert_i(k)} < orig\_t_{delete_{i}(k)}$.

It remains to prove that $t_{delete_i(k)} < orig\_t_{insert_{i+1}(k)}$. 
By \Cref{observation:alternating original lps}, \ins$_{i+1}$($k$)'s original linearization point occurs after \del$_{i}$($k$)'s original linearization point. If the node deleted by \del$_{i}$($k$) has been already unlinked prior to $orig\_t_{insert_{i+1}(k)}$, then prior to the unlinking, \codestyle{updateMetadata} has been called on behalf of \del$_{i}$($k$).
Else, 
we note that in all set implementations we are aware of, if at the original linearization moment of a successful \ins{}($k$) there exists a node with the key $k$ reachable from the data structure's roots, then the \ins{} operation must have observed this node earlier, during its search for $k$.
Thus, \ins$_{i+1}$($k$) observes the node deleted by \del$_{i}$($k$), and calls \codestyle{updateMetadata} on its behalf before carrying out its own original linearization point. In both cases, by \Cref{lemma:lin when updateMetadata returns}, \del$_{i}$($k$) is linearized by the time the \codestyle{updateMetadata} call returns, thus, linearized before $orig\_t_{insert_{i+1}(k)}$.
\end{proof}

\begin{corollary}\label{corollary:alternating ins and del}
The linearization points of successful insertions and deletions of each key $k$ are alternating. 
\end{corollary}

\begin{lemma}\label{lemma:update preceding size}
Let \term{countersSnapshot} be a \codestyle{CountersSnapshot} instance.
Any non-\codestyle{INVALID} value written to a counter in the \term{countersSnapshot}.\codestyle{snapshot} array must have been written to the corresponding counter in the \codestyle{metadataCounters} array (of the \codestyle{SizeCalculator} instance held by the set) \emph{before} the \term{countersSnapshot}.\codestyle{collecting} field is set to \codestyle{false}.
\end{lemma}
Intuitively, this implies that a \size{} operation cannot witness future update operations (namely, ones that are linearized after it).
\begin{proof}
Consider some \term{countersSnapshot}.\codestyle{snapshot}'s cell $C$ of either an insertion or a deletion counter of some thread $T$. We will analyze all possible writes of non-\codestyle{INVALID} values to $C$, and show that they write values that have been written to $T$'s corresponding metadata counter before \term{countersSnapshot}.\codestyle{collecting} is set to \codestyle{false}.
Non-\codestyle{INVALID} values are written to \term{countersSnapshot}.\codestyle{snapshot} in the \codestyle{CountersSnapshot}'s \codestyle{add} and \codestyle{forward} methods, in \Cref{code:CountersSnapshot snapshot cas in add,code:CountersSnapshot snapshot cas in forward} respectively.
Starting with \codestyle{add}, only the first execution of the \cas{} in \Cref{code:CountersSnapshot snapshot cas in add} on $C$ succeeds, and it occurs within a call to \codestyle{SizeCalculator}'s \codestyle{\_collect} method, before the first time \term{countersSnapshot}.\codestyle{collecting} is set to \codestyle{false} in \Cref{code:SizeCalculator compute stop collection}.
As for \codestyle{forward}, it is called by \codestyle{SizeCalculator}'s \codestyle{updateMetadata} method for forwarding some value $val$ to \term{countersSnapshot} after (1) $val$ is written to the relevant $T$'s metadata counter---as guaranteed by \Cref{lemma:counter inc after first 2 lines of updateMetadata}, and then (2) \term{countersSnapshot}.\codestyle{collecting} is verified to bear the value \codestyle{true} in \Cref{code:SizeCalculator updateMetadata check collecting}. Hence, $val$ has been written to the relevant metadata counter before \term{countersSnapshot}.\codestyle{collecting} is set to \codestyle{false}.
\end{proof}

\end{document}